\documentclass{article}

\usepackage[dvipsnames,table,xcdraw]{xcolor}
\usepackage[pdftex]{graphicx}

\usepackage{adjustbox}
\usepackage[noEnd=true,commentColor=black]{algpseudocodex}
\usepackage{algorithm}
\usepackage{amsmath,amssymb,amsthm,amsfonts}
\usepackage{annotate-equations}
\usepackage{authblk}
\usepackage{array}
\usepackage{bibunits}
\usepackage{bold-extra}
\usepackage{booktabs}
\usepackage{cancel}
\usepackage{caption}
\usepackage{circledsteps}
\usepackage{cite}
\usepackage{csvsimple}
\usepackage{etoolbox}
\usepackage{float}
\newfloat{algorithm}{t}{lop} 
\usepackage[margin=1in]{geometry}
\usepackage{hyperref}
\usepackage{import}
\usepackage[inline]{enumitem}
\usepackage{mathtools}
\usepackage{natbib}
\usepackage{orcidlink}
\usepackage{physics}
\usepackage{placeins}
\usepackage[subtle]{savetrees}
\usepackage{siunitx}
\usepackage{stix}
\usepackage{stmaryrd}
\usepackage{subcaption}
\usepackage{tabularx}
\usepackage{textcomp}
\usepackage{tikz}
\usetikzlibrary{arrows.meta}
\usepackage[most]{tcolorbox}
\usepackage{xparse}
\usepackage{pythonhighlight}
\usepackage{empheq}
\usepackage{cleveref}

\definecolor{seabornBrown}{RGB}{134, 86, 75}
\definecolor{seabornPink}{RGB}{255, 119, 193}
\definecolor{seabornPinkDark}{RGB}{152, 34, 116}

\newcolumntype{Y}{>{\centering\arraybackslash}X}

\DeclareRobustCommand{\rchi}{{\mathpalette\irchi\relax}}
\newcommand{\irchi}[2]{\raisebox{\depth}{$#1\chi$}} 

\newtcolorbox{mybox}{
enhanced,
boxrule=0pt,frame hidden,
borderline west={4pt}{0pt}{green!50!black},
colback=green!30!gray!15,
sharp corners,
parbox=false
}

\makeatletter
\def\mathcolor#1#{\@mathcolor{#1}}
\def\@mathcolor#1#2#3{%
  \protect\leavevmode
  \begingroup\color#1{#2}#3\endgroup
}
\makeatother

\newcommand{\adjustedaccent}[1]{%
  \mathchoice{}{}
    {\mbox{\raisebox{-.75ex}[0pt][0pt]{$\scriptscriptstyle#1$}}}
    {\mbox{\raisebox{-.55ex}[0pt][0pt]{\scalebox{.8}{$\scriptscriptstyle#1$}}}}
}
\newcommand\smileacc[1]{\overset{\adjustedaccent{\smile}}{#1}}

\makeatletter
\DeclareRobustCommand{\ctau}{{
  \mathpalette\cap@greek\tau
}}
\DeclareRobustCommand{\csigma}{{
  \mathpalette\cap@greek\sigma
}}
\newcommand{\cap@greek}[2]{%
  \begingroup
  \sbox\z@{$#1t$}
  \resizebox{!}{\ht\z@}{$\m@th#1#2$}
  \endgroup
}
\makeatother

\newcommand{\colort}{\mathcolor{VioletRed}{\mathrm{t}}}
\newcommand{\colortsetofT}{\lBrace \colort \rBrace}
\newcommand{\colortsetofTone}{\lBrace \colort + 1 \rBrace}
\newcommand{\colortau}{\mathcolor{orange}{\ctau}}
\newcommand{\colortausetoft}{\lBrace \colortau \rBrace}
\newcommand{\colortausetofT}{\{\hspace{-0.5ex}\lBrace \colortau \rBrace\hspace{-0.5ex}\}}
\newcommand{\colorT}{\mathcolor{red}{T}}
\newcommand{\colorTbar}{\mathcolor{red}{\smash{\smileacc{T}}}}

\newcommand{\colork}{\mathcolor{purple}{k}}
\newcommand{\colors}{\mathcolor{blue}{\hat{\mathrm{s}}}}
\newcommand{\colorS}{\mathcolor{blue}{S}}  
\newcommand{\colorK}{\mathcolor{purple}{\mathrm{K}}}

\newcommand{\colorh}{\mathcolor{violet}{h}}
\newcommand{\colorH}{\mathcolor{violet}{\mathrm{H}}}
\newcommand{\colorHcal}{\mathcolor{violet}{\mathcal{H}}}

\newcommand{\colorL}{\mathcolor{purple}{\mathrm{L}}}

\newcommand{\colorB}{\mathcolor{olive}{\mathcal{B}}}
\newcommand{\colorBnot}{\mathcolor{olive}{\cancel{\mathcal{B}}}}

\newcommand{\colorg}{\mathcolor{teal}{g}}
\newcommand{\colorG}{\mathcolor{teal}{G}}

\let\mycheckmark\checkmark
\renewcommand{\checkmark}{\textcolor{ForestGreen}{\mycheckmark}}

\newcommand{\nullval}{\texttt{null}}

\theoremstyle{definition}
\newtheorem{proofpart}{Part}
\newtheorem{theorem}{Theorem}[section]
\newtheorem{lemma}{Lemma}[section]
\newtheorem{sublemma}{Sublemma}[lemma]
\newtheorem{corollary}{Corollary}[lemma]
\makeatletter
\@addtoreset{proofpart}{theorem}
\@addtoreset{proofpart}{lemma}
\@addtoreset{proofpart}{corollary}
\makeatother

\hypersetup{breaklinks=true}

\renewcommand{\eqnhighlightheight}{\mathstrut}

\newcommand*{\colorboxed}{}
\def\colorboxed#1#{%
  \colorboxedAux{#1}%
}
\newcommand*{\colorboxedAux}[3]{%
  \begingroup
    \setlength\fboxrule{1pt}
    \colorlet{cb@saved}{.}%
    \color#1{#2}%
    \boxed{%
      \color{cb@saved}%
      #3%
    }%
  \endgroup
}

\makeatletter
\renewcommand{\boxed}[1]{\text{\fboxsep=.2em\fbox{\m@th$\displaystyle#1$}}}
\makeatother

\renewcommand{\stackrel}[2]{%
  \mathrel{\smash{\vbox{\offinterlineskip\ialign{%
    \hfil##\hfil\cr
    $\scriptscriptstyle#1$\cr
    $#2$\cr
}}}}}

\newcommand{\hv}{h.v.\hphantom{}}

\newcommand{\twodots}{\mathinner {\ldotp \ldotp}}

\let\mythealgorithm\thealgorithm

\makeatletter
\newlength{\comment@width}

\renewcommand{\Comment}[1]{%
  \sbox0{#1}
  \ifdim\wd0>\comment@width
    \setlength{\comment@width}{\wd0}%
  \fi
  \ifcsname comment@\arabic{algorithm}@width\endcsname
    \algorithmiccomment{\makebox[\csname comment@\mythealgorithm @width\endcsname][l]{#1}}%
  \else
    \algorithmiccomment{#1}%
  \fi
}
\AtBeginEnvironment{algorithmic}{\setlength{\comment@width}{0pt}}
\AtEndEnvironment{algorithmic}{%
  \immediate\write\@auxout{%
    \string\algcommentwidth{\mythealgorithm}{\the\comment@width}%
  }%
}
\newcommand{\algcommentwidth}[2]{%
  \global\@namedef{comment@#1@width}{#2}%
}
\makeatother

\MakeRobust{\Call}
\makeatletter
\def\algbackskip{\hskip-\ALG@thistlm}
\makeatother

\lstnewenvironment{mypython}[1][]{\lstset{style=mypython,#1}}{}

\definecolor{darkred}{HTML}{E32B60}
\newcommand{\code}[1]{\mbox{%
    \ttfamily
    \color{darkred}
    \tcbox[
        on line,
        boxsep=0pt, left=4pt, right=4pt, top=2pt, bottom=1.5pt,
        toprule=0pt, rightrule=0pt, bottomrule=0pt, leftrule=0pt,
        oversize=0pt, enlarge left by=0pt, enlarge right by=0pt,
        colframe=white, colback=black!12,
        height=.8\baselineskip
    ]{\color{darkred}\detokenize{#1}}%
}}

\definecolor{codegreen}{rgb}{0,0.6,0}
\definecolor{codegray}{rgb}{0.5,0.5,0.5}
\definecolor{codepurple}{rgb}{0.58,0,0.82}
\definecolor{backcolour}{rgb}{0.95,0.95,0.92}

\lstdefinestyle{mystyle}{
    backgroundcolor=\color{backcolour},
    commentstyle=\color{codegreen},
    keywordstyle=\color{magenta},
    numberstyle=\tiny\color{codegray},
    stringstyle=\color{codepurple},
    basicstyle=\ttfamily\footnotesize,
    breakatwhitespace=false,
    breaklines=true,
    captionpos=t,
    keepspaces=true,
    numbers=left,
    numbersep=5pt,
    showspaces=false,
    showstringspaces=false,
    showtabs=false,
    tabsize=2
}

\DeclareCaptionFormat{listing}{\rule{\dimexpr\textwidth\relax}{0.4pt}\par\vskip1pt#1#2#3\par\vspace{-5pt}\rule{\dimexpr\textwidth\relax}{0.4pt}}
\makeatletter
\let\pragma@iinput=\@iinput
\def\@iinput#1{\xdef\@pragmafile{#1}\pragma@iinput{#1} }
\def\@pragmafile{default}
\def\pragmaonce{%
   \csname pragma@\@pragmafile\endcsname
   \global\expandafter\let \csname pragma@\@pragmafile\endcsname =  
}
\makeatother

\ifdefined\mydraft
\mydraft
\fi

\defaultbibliography{bibl}
\defaultbibliographystyle{apalike-ejor}

\begin{document}

\title{ Structured Downsampling for Fast, Memory-efficient Curation of Online Data Streams }

\author[1,2,3,4]{Matthew Andres Moreno\orcidlink{0000-0003-4726-4479}\thanks{Corresponding author: \texttt{morenoma@umich.edu}}}
\author[1,2,4]{Luis Zaman\orcidlink{0000-0001-6838-7385}}
\author[5,6,7]{Emily Dolson\orcidlink{0000-0001-8616-4898}}

\affil[1]{Department of Ecology and Evolutionary Biology}
\affil[2]{Center for the Study of Complex Systems}
\affil[3]{Michigan Institute for Data and AI in Society}
\affil[4]{University of Michigan, Ann Arbor, United States}
\affil[5]{Department of Computer Science and Engineering}
\affil[6]{Program in Ecology, Evolution, and Behavior}
\affil[7]{Michigan State University, East Lansing, United States}
\maketitle

\begin{bibunit}

\begin{abstract}
Operations over data streams typically hinge on efficient mechanisms to aggregate or summarize history on a rolling basis.
For high-volume data steams, it is critical to manage state in a manner that is fast and memory efficient --- particularly in resource-constrained or real-time contexts.
Here, we address the problem of extracting a fixed-capacity, rolling subsample from a data stream.
Specifically, we explore ``data stream curation'' strategies to fulfill requirements on the composition of sample time points retained.
Our ``DStream'' suite of algorithms targets three temporal coverage criteria: (1) steady coverage, where retained samples should spread evenly across elapsed data stream history; (2) stretched coverage, where early data items should be proportionally favored; and (3) tilted coverage, where recent data items should be proportionally favored.
For each algorithm, we prove worst-case bounds on rolling coverage quality.
In contrast to previous work by Moreno, Rodriguez Papa, and Dolson (2024), which dynamically scales memory use to guarantee a specified level of coverage quality, here we focus on the more practical, application-driven case of maximizing coverage quality given a fixed memory capacity.
As a core simplifying assumption, we restrict algorithm design to a single update operation: writing from the data stream to a calculated buffer site --- with data never being read back, no metadata stored (e.g., sample timestamps), and data eviction occurring only implicitly via overwrite.
Drawing only on primitive, low-level operations and ensuring full, overhead-free use of available memory, this ``DStream'' framework ideally suits domains that are resource-constrained (e.g., embedded systems), performance-critical (e.g., real-time), and fine-grained (e.g., individual data items as small as single bits or bytes).
In particular, proposed power-of-two-based buffer layout schemes support $\mathcal{O}(1)$ data ingestion via concise bit-level operations.
To further practical applications, we provide plug-and-play open-source implementations targeting both scripted and compiled application domains.
\end{abstract}

\section{Introduction} \label{sec:introduction}

Efficient operations over data streams are critical in harnessing the ever-increasing volume and velocity of data generation.
Formally, a data stream is considered to be composed of a strictly-ordered sequence of read-once inputs.
Such streams' ordering may be dictated by inherently real-time processes (e.g., physical sensor inputs) or by access patterns for physical storage media (e.g., a tape archive) \citep{henzinger1998computing}.
They may also result from non-reversible computations (e.g., forward-time simulation) \citep{abdulla2004simulation,schutzel2014stream}.
Work with data streams assumes input greatly exceeds memory capacity, with streams often treated as unbounded \citep{jiang2006research}.
Indeed, real-world computing often requires real-time operations on a continuous, indefinite basis \citep{cordeiro2016online}.
Notable application domains involving data streams include sensor networks \citep{elnahrawy2003research}, distributed big-data processing \citep{he2010comet}, real-time network traffic analysis \citep{johnson2005streams,muthukrishnan2005data}, systems log management \citep{fischer2012real}, fraud monitoring \citep{rajeshwari2016real}, trading in financial markets \citep{agarwal2009faster}, environmental monitoring \citep{hill2009real}, and astronomy \citep{graham2012data}.

Here, we focus specifically on subsampling over data streams and introduce three $\mathcal{O}(1)$ ``DStream'' algorithms for space-efficient curation of data items:
\begin{enumerate*}
\item \textit{evenly covering} elapsed history (``\textit{steady}'' algorithm, Section \ref{sec:steady}),
\item \textit{skewed older} over elapsed history (``\textit{stretched}'' algorithm, Section \ref{sec:stretched}), or
\item \textit{skewed newer} over elapsed history (``\textit{tilted}'' algorithm, Section \ref{sec:tilted}).
\end{enumerate*}
Together, these algorithms support a variety of use cases differing in what data is prioritized.
Figure \ref{fig:criteria-intuition} compares steady, stretched, and tilted retention.

For each algorithm, we demonstrate worst-case bounds on error in curated collection composition.
We refer to this rolling subset problem as ``data stream curation,'' which we will define next.

\subsection{Stream Curation Problem}
\label{sec:stream-curation-problem}

Our work concerns online sampling of discrete data items from a one-dimensional data stream.
In selecting retained data items, we seek to ``curate'' a collection containing samples spanning the first items ingested from the data stream through the most recently ingested items \citep{moreno2024algorithms}.
The objective, ultimately, is to preserve a representative, approximate record of stream history.
We consider a retained collection's coverage over history solely in terms of the timepoints (i.e., sequence positions) of retained data items.
Note that we disregard data items' semantic values in this work, as they are immaterial under this timepoint-based framing.

We define three cost functions on the timepoints of discarded data items:
\begin{empheq}[left={\hspace{1.5in}\displaystyle \mathsf{cost}(\colorT) \coloneq \empheqlbrace}]{align}
  &\max_{\colorTbar \in [0\twodots\colorT)} \colorG_{\colorT}(\colorTbar) &&\text{for \textit{``steady''} curation,}  &&& ~ &&& ~ \label{eqn:steady-cost} \\
  &\max_{\colorTbar \in [1\twodots\colorT)} \frac{\colorG_{\colorT}(\colorTbar)}{\colorTbar} &&\text{for \textit{``stretched''} curation, and}  &&& ~ &&& ~ \label{eqn:stretched-cost} \\
  &\max_{\colorTbar \in [0\twodots\colorT - 1)} \frac{\colorG_{\colorT}(\colorTbar)}{\colorT - 1 - \colorTbar} &&\text{for \textit{``tilted''} curation,}  &&& ~ &&& ~ \label{eqn:tilted-cost}
\end{empheq}
where $\colorT$ is ``logical time'' (how many data items have been ingested), $\colorTbar$ is the timepoints of an ingested data item, and $\colorG_{\colorT}(\colorTbar)$ is ``gap size'' in the curated record around timepoint $\colorTbar$ at logical time $\colorT$.
Section \ref{sec:notation} provides a full introduction of notation, with formal definitions.
Analysis of these cost functions, including best-case lower bounds on cost, accompanies presentation of steady, stretched, and tilted algorithms targeting each in Sections \cref{sec:steady,sec:stretched,sec:tilted}.

Formally, our objective is to maintain cost function $\mathsf{cost}(\colorT)$ below an upper bound $\mathsf{bound}(\colorT)$ across logical time $\colorT$.
We specify $\mathsf{bound}(\colorT)$ on a per-algorithm basis,
We assume curation as an online process where new items are ingested on an ongoing basis, and a properly curated archive is needed at all times.
In practice, such fully online curation can be necessary when either (a) stream records are consulted frequently or (b) time point(s) for which stream records are needed are not known \textit{a priori} (i.e., query- or trigger-driven events).

\begin{figure}
\begin{minipage}[]{\textwidth}
    \noindent\fcolorbox{gray!3}{gray!3}{%
    \begin{minipage}[c][0.55in]{0.02\textwidth}
    \hspace{-0.5ex}%
    \rotatebox{90}{$\colorT = 100$}
    \end{minipage}}%
    \hspace{-1.5ex}%
    {\vrule width 1pt}%
    \noindent\fcolorbox{orange!4}{orange!4}{%
    \begin{minipage}[c][0.55in]{0.34\textwidth}
    \includegraphics[trim={0 1.2cm 0 0},clip,width=\linewidth]{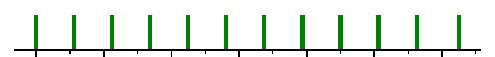}%
    \end{minipage}}%
    \hspace{-1ex}%
    {\vrule width 1pt}%
    \noindent\fcolorbox{pink!6}{pink!6}{%
    \begin{minipage}[c][0.55in]{0.31\textwidth}
    \includegraphics[trim={0 1.2cm 0 0},clip,width=\linewidth]{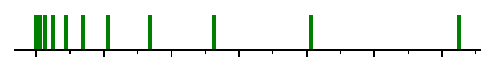}%
    \end{minipage}}%
    \hspace{-1ex}%
    {\vrule width 1pt}%
    \noindent\fcolorbox{teal!5}{teal!5}{%
    \begin{minipage}[c][0.55in]{0.31\textwidth}
    \includegraphics[trim={0 1.2cm 0 0},clip,width=\linewidth]{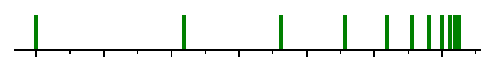}%
    \end{minipage}\hspace{-1ex}}%
    \end{minipage}\vspace{-1ex}

    \begin{minipage}[]{\textwidth}
    \noindent\fcolorbox{gray!11}{gray!11}{%
    \begin{minipage}[c][0.55in]{0.02\textwidth}
    \hspace{-0.5ex}%
    \rotatebox{90}{$\colorT = 50$}
    \end{minipage}}%
    \hspace{-1.5ex}%
    {\vrule width 1pt}%
    \noindent\fcolorbox{orange!10}{orange!10}{%
    \begin{minipage}[c][0.55in]{0.34\textwidth}
    \includegraphics[trim={0 0 0 0},clip,width=\linewidth]{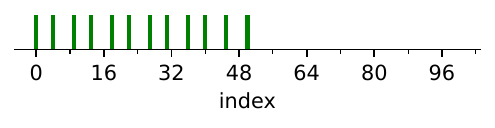}%
    \end{minipage}}%
    \hspace{-1ex}%
    {\vrule width 1pt}%
    \noindent\fcolorbox{pink!15}{pink!15}{%
    \begin{minipage}[c][0.55in]{0.31\textwidth}
    \includegraphics[trim={0 0 0 0},clip,width=\linewidth]{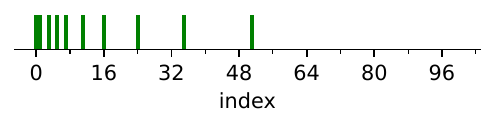}%
    \end{minipage}}%
    \hspace{-1ex}%
    {\vrule width 1pt}%
    \noindent\fcolorbox{teal!11}{teal!11}{%
    \begin{minipage}[c][0.55in]{0.31\textwidth}
    \includegraphics[trim={0 0 0 0},clip,width=\linewidth]{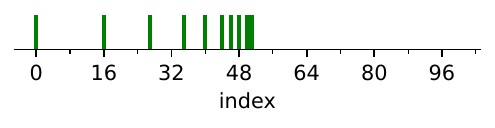}%
    \end{minipage}\hspace{-1ex}}%
    \end{minipage}\vspace{-1ex}

    \begin{minipage}[]{\textwidth}
    \noindent\fcolorbox{white!100}{white!100}{%
    \begin{minipage}[b][0.2cm]{0.02\textwidth}
    \hspace{-0.5ex}%
    \rotatebox{90}{$~$}
    \end{minipage}}%
    \hspace{-1.5ex}%
    {\vrule width 1pt}%
    \noindent\fcolorbox{orange!2}{orange!2}{%
    \begin{minipage}[]{0.34\textwidth}
    \hspace{-1ex}\begin{subfigure}[b]{\linewidth}
    \caption{steady criterion}
    \label{fig:criteria-intuition-steady}
    \end{subfigure}%
    \end{minipage}}%
    \hspace{-1ex}%
    {\vrule width 1pt}%
    \noindent\fcolorbox{pink!3}{pink!3}{%
    \begin{minipage}[]{0.31\textwidth}
    \hspace{-1ex}\begin{subfigure}[b]{\linewidth}
    \caption{stretched criterion}
    \label{fig:criteria-intuition-stretched}
    \end{subfigure}%
    \end{minipage}}%
    \hspace{-1ex}%
    {\vrule width 1pt}%
    \noindent\fcolorbox{teal!3}{teal!3}{%
    \begin{minipage}[]{0.31\textwidth}
    \hspace{-1ex}\begin{subfigure}[b]{\linewidth}
    \caption{tilted criterion}
    \label{fig:criteria-intuition-tilted}
    \end{subfigure}%
    \end{minipage}\hspace{-1ex}}%
    \end{minipage}
    \vspace{-0.15in}
    \caption{%
      \textbf{Surveyed target coverage criteria.}
      \footnotesize
      Ideal distributions of ingestion time points for retained data items under each criterion are shown at $\colorT=50$ (bottom) and $\colorT=100$ (top).
      Vertical bars represent a retained data item.
      In this illustration, collection size is 12 retained items.
      All other ingested data items have been discarded.
      The steady criterion (\ref{fig:criteria-intuition-steady}) seeks to minimize largest absolute gap size.
      So, ideal retention maintains items spread evenly across data stream history.
      The stretched criterion (\ref{fig:criteria-intuition-stretched}) calls for greater retention of early data items to minimize gap size proportional to data item ingestion time $\colorTbar$.
      In contrast, under the tilted criterion (\ref{fig:criteria-intuition-tilted}) recency-proportional gap size is to be minimized, necessitating over-retention of recent data items.
      }
    \label{fig:criteria-intuition}
\end{figure}



\subsection{Applications of Stream Curation}

Efficient stream curation operations benefit a variety of use cases requiring synopses of data stream history.
A straightforward application of stream curation is in unattended or sporadically uplinked sensor devices, which must record incoming observation streams on an indefinite or indeterminate basis, with limited memory capacity \citep{jain2022survey}.
In practice, however, even well-resourced centralized systems require thinning of full fidelity data --- raising the possibility of use cases in long-term telemetry and log management \citep{kent2006guide,miebach2002hubble}.
Checkpoint-rollback state might also be managed through stream curation in scenarios where the possibility of non-halting silent errors requires support for arbitrary rollback extents \citep{aupy2013combination}.
Extensions could be imagined to support more general aggregation and approximation operations over stream history besides sampling \citep{schoellhammer2024lightweight}, although we do not directly investigate these possibilities here.


Algorithms reported here stem from work on ``\textit{hereditary stratigraphy},'' a recently-developed technique for tracking of digital ancestry trees in highly-distributed systems --- for instance, in analysis of many-processor agent-based evolution simulations, content in decentralized social networks, peer-to-peer file sharing, or computer viruses \citep{moreno2022hereditary}.
Although beyond the scope of objectives here, we will briefly motivate this particular use case of stream curation.
Hereditary stratigraphy annotates surveilled artifacts with checkpoint data, which is extended by a new ``fingerprint'' value with each copy event.
Comparing two artifacts' accreted records reconstructs the duration of their common ancestry, with the first mismatched fingerprints signifying divergence from common descent.

This use case relies on stream curation to prevent unbounded growth of generational fingerprint records.
These records can be considered a data stream in that they accrue indefinitely, piece by piece.
Downsampling fingerprints saves memory, but introduces uncertainty in estimating the timing of lineage divergence.
For this reason, spacing of retained checkpoints across generational history is crucial to inference quality.
Minimizing per-item storage overhead is also critical to hereditary stratigraphy, with \citet{moreno2024guide} finding that single-bit checkpoint values maximize reconstruction quality (i.e., by allowing more fingerprints to be retained).
Both of these concerns are prioritized in present work.

\subsection{Prior Work}
\label{sec:prior-work}

Given the broad applicability of the data stream paradigm, many algorithms exist for analysis and summarization over sequenced input --- such as rolling summary statistic calculations \citep{lin2004continuously}, on-the-fly data clustering \citep{silva2013data}, live anomaly detection \citep{cai2004maids}, and rolling event frequency estimation \citep{manku2002approximate}.
Stream curation touches in particular on two broad paradigms data stream processing:
\begin{enumerate}
\item \textit{sampling}, where the data stream corpus is coarsened through extraction of exemplar data items \citep{sibai2016sampling}; and
\item \textit{binning/windowing}, where data stream content is aggregated (e.g., summarized, compressed, or sampled) with respect to discrete time spans over stream history \citep{gama2007data}.
\end{enumerate}

Although curated data items are, indeed, a sample of a data stream, work here is orthogonal to the question of $\ell_p$ sampling (e.g., $\ell_0$, $\ell_1$ sampling) in that our objective is to optimize for temporal balance rather than stochastic composition.
Indeed, well-established techniques exist to extract rolling $\ell_p$-representative samples over the distribution of data values from a stream, such as reservoir sampling, sketching, and hash-based methods \citep{gaber2005mining,muthukrishnan2005data,cormode2019lp}.
Note also that stream curation pertains to logical time rather than real time \citep{sibai2016sampling}, as retention objectives are organized vis-a-vis sequence index rather than clock time.

Owing to dimension reduction's fundamental role in supporting more advanced data stream operations, substantial work exists addressing the question of downsampling via temporal binning.
Notably, schemes for fixed-capacity steady (``equi-segmented'') and tilted (``vari-segmented'') retention appear in \citep{zhao2005generalized}, with the latter resembling additional ``pyramidal,'' ``logarithmic,'' and ``tilted'' time window schemes appearing elsewhere \citep{aggarwal2003framework,han2005stream,giannella2003mining,phithakkitnukoon2010recent}.
Although congruities exist in objectives and aspects of algorithm structure, no existing work prescribes non-iterative layout and update procedures that emphasize minimization of representational overhead (e.g., avoiding storage of timestamps, segment length values, etc.) --- as pursued here.
Work on ``amnesic approximation,'' a generalized scheme for downsampling satisfying an arbitrary temporal cost function, has related objectives but caters to a substantially more resource-intensive use case \citep{palpanas2004online}.


\citet{moreno2024algorithms} presented earlier stream curation techniques in service of hereditary stratigraphy.
Whereas that earlier work also focuses on minimizing the representational footprint around stored data, it caters better to variable-capacity storage, rather than fixed-capacity.
Although configurations oriented to fixed-capacity use cases targeted here are also explored in \citet{moreno2024algorithms}, they require a more expensive update process that keeps data in sorted order and can leave buffer capacity unused.
Indeed, head-to-head benchmark trials demonstrate improved execution speed (by an order of magnitude) and enhanced buffer space utilization under tilted retention \citep{moreno2024guide,moreno2024trackable}.

\subsection{Proposed Approach}


Our proposed DStream approach adopts a strong simplifying constraint: Once stored, we do not allow data items to be subsequently inspected or moved.
We assume a fixed number of buffer sites where items ingested from a data stream may be written.
The only further event that may occur after a data item is stored is being overwritten by a later data item.
We also allow ingested data items to be discarded without storage.
Under this regime, the composition of retained data emerges implicitly as a consequence of items targeted for overwrite.
Put another way, curation policy is exercised solely through ``\textit{site selection}'' when picking a buffer index for the $n$th received data item.

Note that this operational scheme supports particularly efficient storage of fine-grained data items, as it inherently forgoes overhead from explicit data labeling, timestamping, or other structure (e.g., pointers).
Instead, we require site selection to be computable \textit{a priori}.
As a further consequence, efficient attribution of data items' origin time hence requires support for efficient ``inverse'' decoding of a stored data item's origin time based solely on its buffer index and how many items have been ingested from the data stream.
We term this operation ``\textit{site lookup}.''
Figure \ref{fig:ingest-and-lookup} schematizes our ``site selection'' and ``site lookup'' operations.

\begin{figure*}[htbp!]
  \centering
  \begin{minipage}[]{\textwidth}
    \centering
    \begin{minipage}{0.36\textwidth}
    \includegraphics[width=\textwidth]{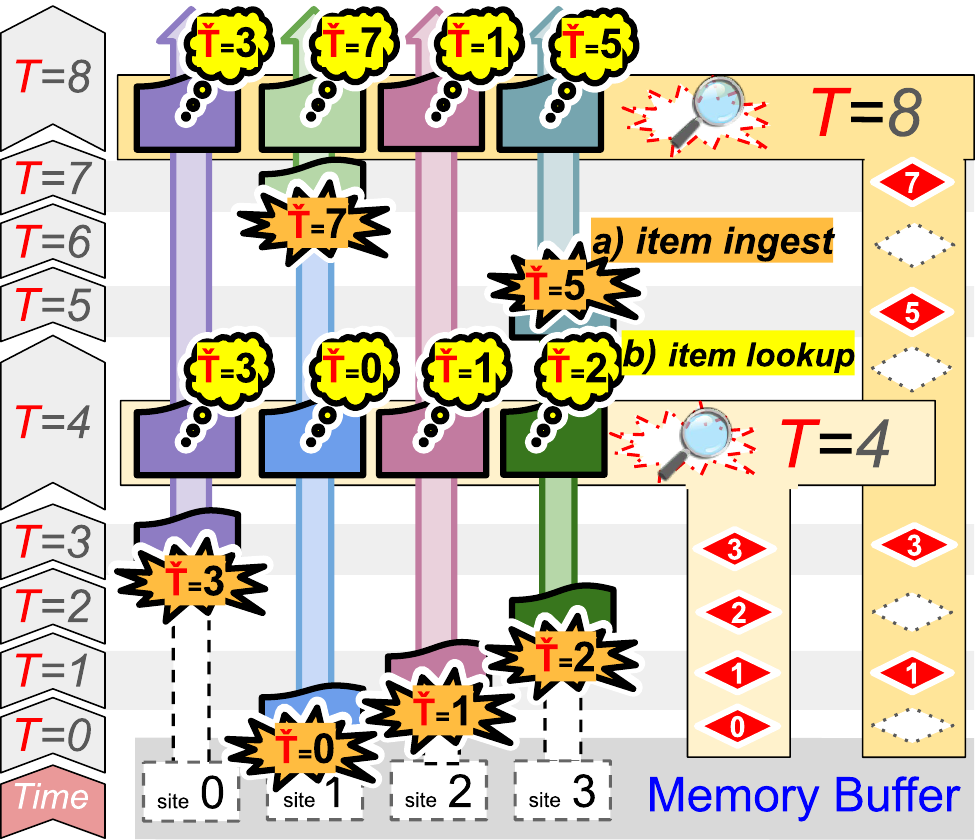}

    \begin{subfigure}{0in}
      \captionsetup{labelformat=empty}
      \caption{}\label{fig:surface-site-ingest}
    \end{subfigure}
    \begin{subfigure}{0in}
      \captionsetup{labelformat=empty}
      \caption{}\label{fig:ingest-rank-calculation}
    \end{subfigure}

  \end{minipage}%
    \hfill
    \begin{minipage}{0.63\textwidth}
      \caption{%
      \textbf{Core stream curation algorithm operations.}
      \footnotesize
      The ingest site selection operation (operation shown as \textit{item ``a''}) takes the current time $\colorT$ and determines the buffer site $\colork$ to store the ingested data item.
      Data items may also be discarded without storage, as are $\colorTbar=4$ and $\colorTbar=6$ in this example.
      This operation is performed when storing data into a curated buffer, once for each data item received from the data stream.
      Data is not moved after it is stored.
      The ingested time calculation operation (operation shown as \textit{item ``b''}) provides the previous time $\colorTbar$ when the data item present at buffer site $\colork$ was ingested, given the current time $\colorT$.
      This operation is performed when reading data from a curated buffer in order to identify the provenance of stored data.
      Note that which data item $\colorTbar$ occupies a buffer site $\colork$ at time $\colorT$ results solely from the sequence of ingest storage sites selected up to that point.
      As such, the site lookup operation $\colorL$ can be considered, in a loose sense, as ``decoding'' or ``inverse'' to the site selection operation $\colorK$.
      Panels with diamond markers on the right show curated collection composition at $\colorT=4$ and $\colorT=8$.
      Figure \ref{fig:criteria-intuition} shows the target curated collection compositions considered in this work.
      }
      \label{fig:ingest-and-lookup}
      \end{minipage}
  \end{minipage}
\end{figure*}

\subsection{Major Results}

This paper contributes three site selection algorithms for stream curation, with corresponding site lookup procedures.
These algorithms differ in the temporal composition of retained data items, targeting steady, stretched, and tilted distributions, respectively.
All three proposed algorithms support $\mathcal{O}(1)$ site selection.
Accompanying site lookup is $\mathcal{O}(\colorS)$ to decode all $\colorS$ buffer sites' ingest times.
We provide worst-case upper bounds on curation quality, with the steady algorithm notable in guaranteeing performance matching best case within a constant factor.

\section{Preliminaries, Notations, and Terminology} \label{sec:notation}

The core function of proposed algorithms is to dynamically filter out a bounded-size subset of incoming data that, according to a desired \textbf{coverage criterion}, manages the structure of gaps in history created by discarding items.
Incoming data is assumed to arrive on a rolling basis, as a \textbf{data stream} comprised of sequential \textbf{data items} $v_i$.
We assume the data stream to be ephemeral (i.e., ``read once''), and refer to the act of reading an item from the data stream as \textbf{ingesting} it.
As mentioned above, we term this scenario the \textbf{data stream curation problem}.

We consider data items according only to their logical sequence position.
We do not consider data items' actual semantic values or real-time arrival.
We assume data items to be fixed size and thus interchangeable in memory buffer slots.

The remainder of this section will proceed to overview key notations used throughout this work, summarized in Table \ref{tab:notation}.

\begin{table}[]
\begin{tabular}{lllll}
\hline
\textit{Description} & \textit{Type} & \textit{Notation} & \textit{Definition} & \textit{Domain} \\ \hline
\rowcolor{gray!20}
\multicolumn{5}{c}{\textbf{Space}} \\ \hline
Buffer Size & int & $\colorS$ & user-defined & $\in \{2^{\mathbb{N}}\}$ \\
Log Buffer Size & int & $\colors$ & $\log_2 \colorS$ & $\in \mathbb{N}$ \\
Buffer Site & int & $\colork$ & index position in buffer & $\in [0\twodots\colorS)$ \\ \hline
\rowcolor{gray!20}
\multicolumn{5}{c}{\textbf{Time}} \\ \hline
Current Time & int & $\colorT$ & num elapsed data item ingests & $\in$ \textsuperscript{\textdagger}$\mathbb{N}$ or \textsuperscript{\textdaggerdbl}$[0 \twodots 2^{\colorS - 1})$ \\
Data Item Ingest Time & int & $\colorTbar$ & num ingests preceding data item & $\in [0 \twodots \colorT)$ \\
Hanoi Value (\hv{}) of Time & int & $\colorh = \colorH(\colorT)$ & Formula \ref{eqn:hanoi-defn} & $\in$ \textsuperscript{\textdagger}$\mathbb{N}$ or \textsuperscript{\textdaggerdbl}$[0 \twodots \colorS)$ \\
Time Epoch & int & $\colort = \mathit{f}_{\!\Circled[fill color=VioletRed!70!lightgray,inner color=white,outer color=VioletRed!70!lightgray,inner xsep=1pt,inner ysep=1pt]{\mathrm{t}}}(\colorT)$ & Formula \ref{eqn:epoch-defn} & $\in$ \textsuperscript{\textdagger}$\mathbb{N}$ or \textsuperscript{\textdaggerdbl}$[0 \twodots \colorS - \colors)$ \\
& set & $\colortsetofT$ & $\{\colorT^{\prime} \in \mathbb{N} : \mathit{f}_{\!\Circled[fill color=VioletRed!70!lightgray,inner color=white,outer color=VioletRed!70!lightgray,inner xsep=1pt,inner ysep=1pt]{\mathrm{t}}}(\colorT^{\prime}) = \colort \}$ & $\subseteq [\colorT' \twodots \colorT' + n]$  \\
Time Meta-epoch & int & $\colortau = \mathit{f}_{\!\Circled[fill color=orange!70!lightgray,inner color=white,outer color=orange!70!lightgray,inner xsep=1pt,inner ysep=1pt]{\tau}}(\colort)$ & Formula \ref{eqn:meta-epoch-defn} & $\in$ \textsuperscript{\textdagger}$\mathbb{N}$ or \textsuperscript{\textdaggerdbl}$[0 \twodots \colors)$ \\
& set & $\colortausetoft$ & $\{\colort^{\prime} \in \mathbb{N} : \mathit{f}_{\!\Circled[fill color=orange!70!lightgray,inner color=white,outer color=orange!70!lightgray,inner xsep=1pt,inner ysep=1pt]{\tau}}(\colort^{\prime}) = \colortau \}$ & $\subseteq [\colort' \twodots \colort' + n]$ \\
& set & $\colortausetofT$ & $\{\colorT^{\prime} \in \mathbb{N} : \mathit{f}_{\!\Circled[fill color=orange!70!lightgray,inner color=white,outer color=orange!70!lightgray,inner xsep=1pt,inner ysep=1pt]{\tau}} \circ \mathit{f}_{\!\Circled[fill color=VioletRed!70!lightgray,inner color=white,outer color=VioletRed!70!lightgray,inner xsep=1pt,inner ysep=1pt]{\mathrm{t}}}(\colorT^{\prime}) = \colortau \}$ & $\subseteq [\colorT' \twodots \colorT' + n]$ \\ \hline
\rowcolor{gray!20}
\multicolumn{5}{c}{\textbf{Layout}} \\ \hline
Hanoi Value Reserved at Site & int & $\colorh = \colorHcal_{\colort}(\colork)$ & algorithm-defined & $\in$ \textsuperscript{\textdagger}$\mathbb{N}$ or \textsuperscript{\textdaggerdbl}$[0 \twodots \colorS)$ \\
Storage Site Selected for Data Item & int & $\colorK(\colorT)$ & algorithm-defined & $\in [0 \twodots \colorS) \cup \{\nullval\}$ \\
Ingest Times of Stored Data Items by Site & seq & $\colorL(\colorT)$ & {\footnotesize $\max\{\colorTbar \in [0 \twodots \colorT) : \colorK(\colorTbar) = \colork\}$ for $\colork \in [0 \twodots \colorS)$} & $\in [0 \twodots \colorT) \cup \{\nullval\}$ \\
Initial Reservation Segment Size & int & $r$ & \textsuperscript{\textdagger}N/A or \textsuperscript{\textdaggerdbl}Formula \ref{eqn:stretched-segment-sizes} & \textsuperscript{\textdagger}N/A or \textsuperscript{\textdaggerdbl}$\in [1 \twodots \colors]$ \\
Mature Reservation Segment Size & int & $R(r)$ & \textsuperscript{\textdagger}N/A or \textsuperscript{\textdaggerdbl} $2^{r} - 1$ (Lemma \ref{thm:stretched-meta-epoch}) & \textsuperscript{\textdagger}N/A or \textsuperscript{\textdaggerdbl}$\in [1 \twodots \colorS]$ \\ \hline
\rowcolor{gray!20}
\multicolumn{5}{c}{\textbf{Curation Quality}} \\ \hline
Retained Data Items & set & $\colorB_{\colorT}$ & algorithm-consequent & $\subseteq [0\twodots\colorT)$ \\
Discarded Data Items & set & $\colorBnot_{\colorT}$ & $[0\twodots\colorT) \setminus \colorB_{\colorT}$ & $\subseteq [0\twodots\colorT)$ \\
Gap Size in Curated Collection & int & $\colorg = \colorG_{\colorT}(\colorTbar)$ & Formula \ref{eqn:gap-size-defn} & $\in [0 \twodots \colorT)$
\end{tabular}
\centering
\caption{
Summary of notation used.
\textsuperscript{\textdagger}steady algorithm; \textsuperscript{\textdaggerdbl}stretched and tilted algorithms.
}
\label{tab:notation}
\end{table}

\subsection{Buffer Storage $\colorS$}
\label{sec:notation-buffer}

We assume a fixed \textit{number of available buffer sites}, sufficient to store $\colorS$ data items.%
\footnote{%
In associated materials, the fixed-size buffer used to store curated data items is referred to as a ``surface.''
Space-efficient solutions for the stream curation problem under extensible memory capacity have been considered in other work \citep{moreno2024algorithms}.%
}
Proposed algorithms require buffer size $\colorS$ as an even power of two, larger than 4. That is, $\colorS = 2^{\colors}$ for some integer $\colors \in \mathbb{N}_{\geq 2}$.
On occasion, it will become necessary to refer to a specific buffer position $\colork$.
We will take a zero-indexing convention, so $\colork \in [0\twodots\colorS)$.

We consider only one update operation on the buffer: storage of an ingested data item at a buffer site $\colork$.
Under this scheme, control of what data is retained and for how long occurs solely as a consequence of \textit{ingestion site selection} --- picking where (and if) to store incoming data items.
Let $\colorK(\colorT) \in [0\twodots\colorS)\cup\{\nullval\}$ denote the site selection operation to place data item $\colorT$ --- with $\nullval$ denoting a data item dropped without storing.%
\footnote{%
A more exacting notation would reflect that site selection depends on buffer size (i.e., as $\colorK_{\colorS}(\colorT)$), but we omit this in our notation for brevity.
}
A schematic of site selection is provided in Figure \ref{fig:surface-site-ingest}.

As a space-saving optimization, we store only the data items themselves in buffer space --- no metadata (e.g., ingestion time) or data structure components (e.g., indices or pointers) are stored.
This optimization is critical, in particular, when data items are small --- such as single bits or single bytes \citep{moreno2022hereditary}.
Without metadata, however, identifying stored data items requires capability to deduce ingest time solely from buffer position $\colork$.
We denote \textit{site lookup} this operation as $\colorL(\colorT)$, yielding the data item ingest times $\colorTbar_{\colork=0}, \colorTbar_{\colork=1}, \;\;\ldots, \colorTbar_{\colork=\colorS-1}$.
Note that if no data item has yet been stored at a site (i.e., when first filling the buffer $\colorT < \colorS$), $\colorL(\colorT)$ may include $\nullval$ values.%
\footnote{%
Although omitted for brevity, it is the case that lookup depends on buffer size (i.e., as $\colorL_{\colorS}(\colorT)$).
}
Figure \ref{fig:ingest-rank-calculation} visualizes the relationship of \textit{site selection} and \textit{site lookup} operations.

\subsection{Logical Time $\colorT$ and Item Ingest Time $\colorTbar$}
\label{sec:notation-time}

We will refer to each data item's stream sequence index as its \textbf{ingest time} $\colorTbar$ and the number of items ingested as the \textbf{current logical time} $\colorT$.
In other contexts, a data item's ingest time $\colorTbar$ might be referred to as its ``sequence position'' within the data stream.
However, we avoid that terminology to prevent confusion of sequence position with buffer position $\colork$.

We use a zero-indexing convention.
Logical time begins at $\colorT=0$, when no data items have yet been ingested.
The first element of the data stream $v_0$ is assigned ingestion time $\colorTbar=0$.
After the first item $v_0$ is ingested, logical time advances to $\colorT=1$.
We assume $\colorT$ to be known at every point, which can be accomplished trivially in practice with a simple counter.
Because we are only concerned with the sequence order of data items (and not their actual data values), we will shorthand $\colorTbar$ as referring to $v_{\colorTbar}$ (i.e., the data item ingested at that time).

\subsection{Gap Size $\colorg$}
\label{sec:notation-gapsize}

We define coverage criteria in terms of \textbf{gap sizes} in the retained record.
Formally, we define gap size as a count of consecutive data items that have been discarded or overwritten.
Let $\colorB_{\colorT}$ denote data items retained in buffer at time $\colorT$ (including $v_{\colorT}$) and $\colorBnot_{\colorT}$ refer to data items discarded (i.e., overwritten) up to that point.
Gap size for record index $\colorTbar \in [0 \twodots \colorT)$ at time $\colorT$ follows as
\begin{align}
\colorG_{\colorT}(\colorTbar)
&\coloneq
\max
\{
  i + j
  \text{ for }
  i,\;\; j \in \mathbb{N}
  :
  [\colorTbar-i \twodots \colorTbar+j) \subseteq \colorBnot_{\colorT}
\}.
\label{eqn:gap-size-defn}
\end{align}
Note that if $\colorTbar \in \colorB_{\colorT}$, then $\colorG_{\colorT}(\colorTbar) = 0$.

\subsection{Time Hanoi Value $\colorh$}
\label{sec:notation-hanoi}

Proposed algorithms make heavy use of OEIS integer sequence A007814 \citep{oeis}, formulated as
\begin{align}
\colorH(\colorT)
\coloneq
\max \{ n \in \mathbb{N} : (\colorT + 1) \bmod 2^n = 0 \}.
\label{eqn:hanoi-defn}
\end{align}
We refer to $\colorH(\colorT) = \colorh$ as the ``\textbf{hanoi value}'' (``\textbf{\hv{}}'') of $\colorT$, in reference to parallels with the famous ``Tower of Hanoi'' puzzle \citep{lucas1889jeux}.

Terms of this sequence correspond to the number of trailing zeros in the binary representation of $\colorT + 1$.%
\footnote{%
As such, in implementation, $\colorH(\colorT)$ can be calculated in fast $\mathcal{O}(1)$ using bit-level operations --- e.g., in Python \texttt{($\sim$T \& T-1).bit\_length()} \citep{oeis}.
}
The first terms are $0,\allowbreak 1,\allowbreak 0,\allowbreak 2,\allowbreak 0,\allowbreak 1,\allowbreak 0,\allowbreak 3,\allowbreak 0,\allowbreak 1,\allowbreak 0,\allowbreak 2,\allowbreak 0,\allowbreak 1,\allowbreak 0,\allowbreak 4,\allowbreak 0,\allowbreak \;\;\ldots \;\;$.
We continue our zero-indexing convention, so $\colorH(0) = 0$, $\colorH(1) = 1$, $\colorH(2) = 0$, etc.

Some intuition for the structure of the Hanoi sequence will benefit the reader.
As depicted in Figure \ref{fig:hanoi-intuition}, the hanoi sequence exhibits recursively-nested fractal structure.
Element 0 appears every 2nd entry, element 1 appears every 4th entry, and in the general case element $\colorh$ appears every $2^{\colorh+1}$th entry.
So, a hanoi value $\colorh$ appears twice as often as value $\colorh + 1$.
When hanoi value $\colorh$ appears for the first time, the value $\colorh - 1$ has appeared exactly once.
So, we have seen precisely one instance of $\colorh$ and also precisely one instance of $\colorh - 1$.
At this point, the value $\colorh - 2$ has appeared exactly twice and, in general, the value $\colorh - n$ has appeared $2^{n - 1}$ times.

DStream algorithms use the \hv{} of data items' ingestion times $\colorH(\colorT)$ as the basis to prioritize items for retention.
Figure \ref{fig:hanoi-intuition} provides intuition for how this core aspect of structure manifests in proposed \textit{steady}, \textit{stretched}, and \textit{tilted} algorithms.

\begin{figure*}
  \centering
\begin{minipage}[]{\textwidth}
\noindent\fcolorbox{gray!3}{gray!3}{%
\begin{minipage}[c][0.97in]{0.02\textwidth}
\hspace{-0.5ex}%
\rotatebox{90}{$\colorT = 100$}
\end{minipage}}%
\hspace{-1.5ex}%
{\vrule width 1pt}%
\noindent\fcolorbox{orange!4}{orange!4}{%
\begin{minipage}[]{0.34\textwidth}
\includegraphics[trim={0.2cm 1.2cm 0 0},clip,height=0.97in]{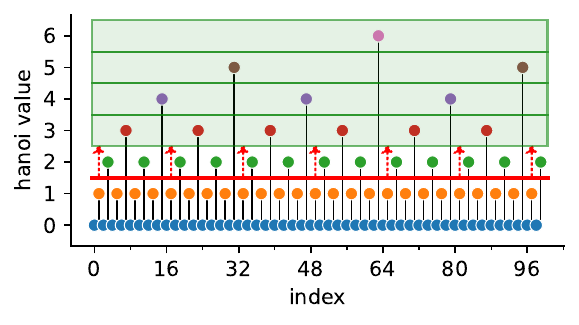}%
\end{minipage}}%
\hspace{-1ex}%
{\vrule width 1pt}%
\noindent\fcolorbox{pink!6}{pink!6}{%
\begin{minipage}[]{0.31\textwidth}
\includegraphics[trim={1cm 1.2cm 0 0},clip,height=0.97in]{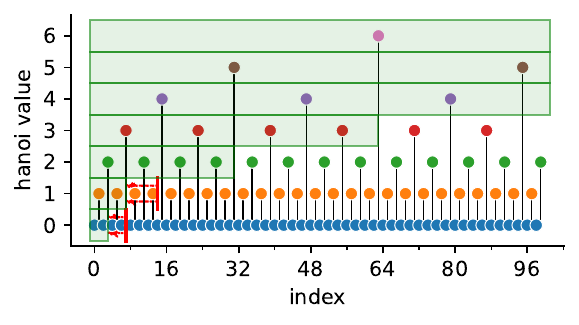}%
\end{minipage}}%
\hspace{-1ex}%
{\vrule width 1pt}%
\noindent\fcolorbox{teal!5}{teal!5}{%
\begin{minipage}[]{0.31\textwidth}
\includegraphics[trim={1cm 1.2cm 0 0},clip,height=0.97in]{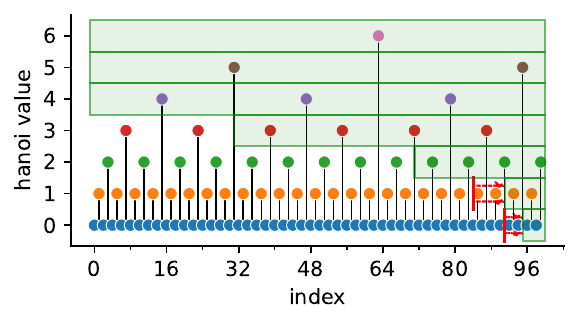}%
\end{minipage}\hspace{-1ex}}%
\end{minipage}\vspace{-1ex}

\begin{minipage}[]{\textwidth}
\noindent\fcolorbox{gray!11}{gray!11}{%
\begin{minipage}[c][1.25in]{0.02\textwidth}
\hspace{-0.5ex}%
\rotatebox{90}{$\colorT = 50$}
\end{minipage}}%
\hspace{-1.5ex}%
{\vrule width 1pt}%
\noindent\fcolorbox{orange!10}{orange!10}{%
\begin{minipage}[c]{0.34\textwidth}
\includegraphics[trim={0.2cm 0 0 0},clip,height=1.25in]{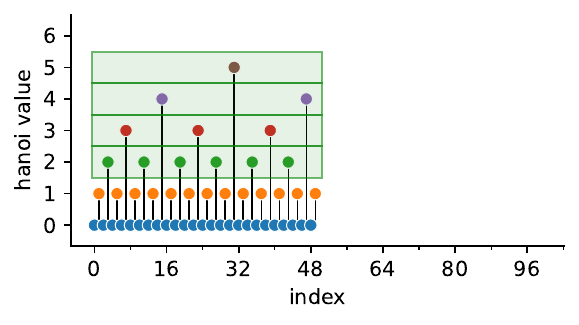}%
\end{minipage}}%
\hspace{-1ex}%
{\vrule width 1pt}%
\noindent\fcolorbox{pink!15}{pink!15}{%
\begin{minipage}[c]{0.31\textwidth}
\includegraphics[trim={1cm 0 0 0},clip,height=1.25in]{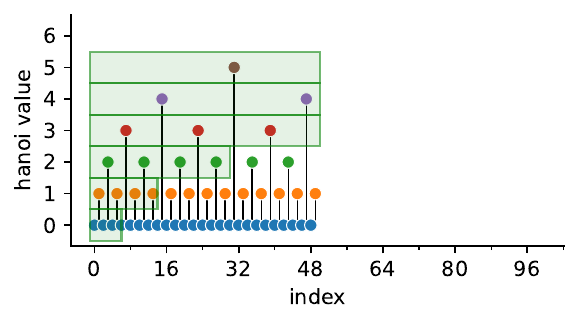}%
\end{minipage}}%
\hspace{-1ex}%
{\vrule width 1pt}%
\noindent\fcolorbox{teal!11}{teal!11}{%
\begin{minipage}[]{0.31\textwidth}
\includegraphics[trim={1cm 0 0 0},clip,height=1.25in]{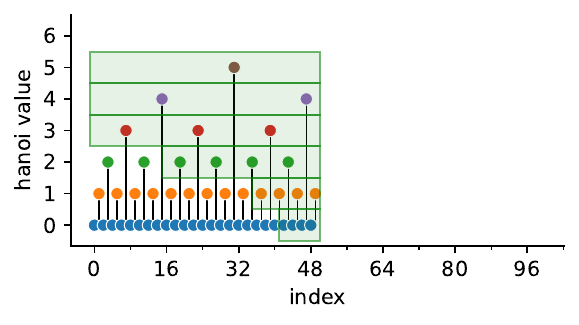}%
\end{minipage}\hspace{-1ex}}%
\end{minipage}\vspace{-1ex}

\begin{minipage}[]{\textwidth}
\noindent\fcolorbox{white!100}{white!100}{%
\begin{minipage}[b][0.2cm]{0.02\textwidth}
\hspace{-0.5ex}%
\rotatebox{90}{$~$}
\end{minipage}}%
\hspace{-1.5ex}%
{\vrule width 1pt}%
\noindent\fcolorbox{orange!2}{orange!2}{%
\begin{minipage}[]{0.34\textwidth}
\centering
\hspace{-1.2ex}\begin{subfigure}[b]{\linewidth}
\captionsetup{justification=centering}
\vspace{-2ex}
\caption{steady strategy\\ \footnotesize all data items of top $n$ hanoi values}
\label{fig:hanoi-intuition-steady}
\end{subfigure}%
\end{minipage}}%
\hspace{-1ex}%
{\vrule width 1pt}%
\noindent\fcolorbox{pink!3}{pink!3}{%
\begin{minipage}[]{0.31\textwidth}
\centering
\hspace{-1.2ex}\begin{subfigure}[b]{\linewidth}
\captionsetup{justification=centering}
\vspace{-2ex}
\caption{stretched strategy\\ \footnotesize first $n$ data items of all hanoi values}
\label{fig:hanoi-intuition-stretched}
\end{subfigure}%
\end{minipage}}%
\hspace{-1ex}%
{\vrule width 1pt}%
\noindent\fcolorbox{teal!3}{teal!3}{%
\begin{minipage}[]{0.31\textwidth}
\centering
\hspace{-1.2ex}\begin{subfigure}[b]{\linewidth}
\captionsetup{justification=centering}
\vspace{-2ex}
\caption{tilted strategy\\ \footnotesize last $n$ data items of all hanoi values}
\label{fig:hanoi-intuition-tilted}
\end{subfigure}%
\end{minipage}\hspace{-1ex}}%
\end{minipage}
\vspace{-0.1in}
\caption{%
  \textbf{Hanoi value retention strategies.}
  \footnotesize
  Data item retention can be prioritized based on ``hanoi value'' of ingestion time $\colorT$.
  Here, ``lollipop'' bars show data item hanoi values, $\colorH(\colorTbar)$.
  To satisfy the steady criterion, our proposed strategy discards data items with \hv{} below a threshold $n(\colorT)$ (\ref{fig:hanoi-intuition-steady}).
  Red arrows show the threshold $n$ increasing as time elapses, purging low \hv{} data items to respect available buffer space.
  Our strategy for the stretched criterion retains the first $n'(\colorT)$ data item instances of all observed \hv{}'s (\ref{fig:hanoi-intuition-stretched}).
  As time elapses, $n'(\colorT)$ is halved across \hv{}'s' in a rolling fashion --- also shown by red arrows above.
  Our strategy to satisfy the tilted criterion operates similarly to the stretched strategy, except the \textit{last} $n'(\colorT)$ data item instances of each \hv{} are retained (\ref{fig:hanoi-intuition-tilted}).
  The bottom and top panels compare example retention at $\colorT=50$ and $\colorT=100$, respectively.
  Green boxes indicate retained data items.
  }
\label{fig:hanoi-intuition}
\end{figure*}

\subsection{Time Epoch $\colort$}
\label{sec:notation-epoch}

Owing to our algorithms' incorporation of \hv{}-based abstractions, it is useful to track a measure related to the binary magnitude of elapsed time $\colorT$ (i.e., $\sim \log_2(\colorT)$).
We call this measure the \textbf{epoch} $\colort$ of time $\colorT$,
\begin{equation}
\colort
\coloneq
\begin{cases}
\left\lfloor \log_2(\colorT) \right\rfloor - \colors + 1 & \text{if $\colorT \geq \colorS$} \\
0 & \text{otherwise.}
\end{cases}
\label{eqn:epoch-defn}
\end{equation}

Under this definition, epochs begin exactly at even powers of two (e.g., $\colorT = 16$) for $\colorT \geq \colorS$.
Correction is applied to begin epoch $\colort=1$ at $\colorT = \colorS$.

\subsection{Site Reservations $\colorHcal_{\colort}(\colork)$}
\label{sec:notation-reservation}

Algorithm design is structured around ``reserving'' (setting aside) buffer sites $\colork \in [0 \twodots \colorS)$ to host data items whose time index $\colorTbar$ has a specific \hv{}, $\{\colorTbar : \colorH(\colorTbar) = \colorh \}$, on an epoch-to-epoch-basis.
Denote site $\colork$'s \textbf{hanoi value reservation} during epoch $\colort$ as $\colorHcal_{\colort}(\colork)$.%
\footnote{
A careful reader may wonder if the notation for site $\colork$'s hanoi value reservation $\colorHcal_{\colort}(\colork)$ should also be qualified by overall buffer size $\colorS$ as $\colorHcal_{\colort,\colorS}(\colork)$, in addition to current epoch $\colort$.
Although omitted from our notation for brevity, this is indeed the case.
}
Note that a data item $\colorTbar \not\in \colortsetofT$ may occupy site $\colork$ during epoch $\colort$ with $\colorHcal_{\colort}(\colork) \neq \colorH(\colorTbar)$, having been held over from the previous epoch $\colort - 1$ before being overwriten with an instance of \hv{} $\colorh = \colorHcal_{\colort}(\colork)$ during the current epoch $\colort$.

A substantial fraction of implementation for presented algorithms relates to how hanoi value reservations $\colorHcal_{\colort}$ are arranged over buffer space $\colork \in [0\twodots\colorS)$ as epochs $\colort$ elapse.
Each algorithm organizes buffer space into contiguous \textbf{reservation segments}.
Within a single reservation segment, all hanoi value reservations are distinct.
That is, no two sites share the same reserved hanoi value.
Reservation segments are themselves further organized into \textbf{segment bunches}.
All segments within a bunch are the same length and have the same left-to-right hanoi value reservation layout.
However, unlike sites in a segment, segments in a bunch may not be laid out contiguously.
Reservation segments in a bunch are contiguous in buffer space under the steady algorithm, but are not under the stretched and tilted algorithms.

Beyond the commonalities above, the precise makeup and layout of segments and segment bunches differs between the steady algorithm versus the stretched and tilted algorithms.
(The latter two algorithms share large commonalities.)
Figures \ref{fig:hsurf-steady-intuition-heatmap} and \ref{fig:hsurf-stretched-intuition-reservations} sketch the makeup of hanoi value reservations, reservation segments, and segment bunches in buffer space over time for the steady algorithm and stretched/tilted algorithms, respectively.
Further details are covered separately for each algorithm in Sections \cref{sec:steady,sec:stretched,sec:tilted}.

\subsection{Time Meta-epoch $\colortau$}
\label{sec:notation-metaepoch}

In the case of the \textit{stretched} and \textit{tilted} algorithms, it becomes useful to group sequential epochs $\colort$ together as \textbf{meta-epochs} $\colortau$.
We define $\colortau=0$ as corresponding to epoch $\colort=0$.
Meta-epoch $\colortau=1$ therefore begins at epoch $\colort = 1$.
As later motivated in Lemma \ref{thm:stretched-meta-epoch}, we define meta-epochs $\colortau\geq1$ as lasting $2^{\colortau} - 1$ epochs.
Under this definition, we have $\colortau\geq1$ as beginning at epoch
\begin{align}
\min(\colort \in \colortausetoft)
&= 1 + \sum_{i=1}^{\colortau - 1} (2^{i} - 1) \nonumber \\
&= 2^{\colortau} - \colortau.
\label{eqn:meta-epoch-defn}
\end{align}
For epoch $\colort > 0$, we can thus calculate the current meta-epoch $\colortau$ exactly as
\begin{align*}
\colortau
=
\begin{cases}
\left\lfloor \log_2(\colort) \right\rfloor + 1 & \text{if } \colort = 2^{\left\lfloor \log_2(\colort) \right\rfloor} - \left\lfloor \log_2(\colort) \right\rfloor \\
\left\lfloor \log_2(\colort) \right\rfloor & \text{otherwise.}
\end{cases}
\end{align*}

\subsection{Restrictions on Logical Time $\colorT$, Epoch $\colort$, and Meta-epoch $\colortau$}

Ideally, data stream curation would support indefinite ingestions, $\colorT \in \mathbb{N}$.
Our proposed \textit{steady curation} algorithm, introduced below, operates in this fashion.
However, our proposed \textit{stretched} and \textit{tilted curation} algorithms accept only $2^{\colorS} - 2$ ingestions.
We expect this capacity to suffice for many applications using even moderately sized buffers.
For instance, a buffer with space for 64 data items suffices to ingest items continuously at 5GHz for over 100 years.
As such, we leave behavior for stretched and tilted curation past $2^{\colorS} - 2$ ingests to future work.

For convenience in exposition, note that we formally define and characterize the stretched and tilted algorithms only for $\colorT \in [0 \twodots 2^{\colorS - 1})$.
However, in practice, extension to $\colorT \in [0 \twodots 2^{\colorS} - 1)$ that respects established guarantees on curation quality is straightforward.
All algorithm psuedocode and reference implementations support this extended domain.

Restricting logical time $\colorT < 2^{\colorS - 1}$ bounds time epoch $\colort$ below
\begin{align*}
\colort &\leq \left\lfloor\log_2(2^{\colorS - 1} - 1)\right\rfloor - \colors + 1
\leq \colorS - \colors - 1
\end{align*}
assuming $\colorS \geq 4$.
The $\colorS - \colors$ relation can be understood as arising due to delay of epoch $\colort=1$ to time $\colorT = \colorS = 2^{\colors}$.
Supplementary Lemma \ref{thm:meta-epoch-bound} establishes the following upper bound on time meta-epoch $\colortau$,
\begin{align*}
\colortau
&\leq
\min\Big(
  \log_2(\colort + \colors),\;\;
  \log_2(\colort) + 1
\Big)
\text{ for } \colort \in [1 \twodots \colorS - \colors).
\end{align*}
Taking $\colort = \colorS - \colors - 1$, we can also bound $\colortau$ over the stretched and tilted algorithms' domains as $\colortau < \colors$.

\subsection{Miscellania}

Algorithm listings refer to a handful of utility helper functions (e.g., \textsc{BitCount}, \textsc{BitLength}, etc.).
Refer to Supplementary Section \ref{sec:pseudocode} for full definitions of these.

Let the binary floor of a value $x$ be denoted $\left\lfloor x \right\rfloor_\mathrm{bin} = 2^{\left\lfloor \log_2 x \right\rfloor}$.
For binary ceiling, let $\left\lceil x \right\rceil_\mathrm{bin} = 2^{\left\lceil \log_2 x \right\rceil}$.
In both cases, we correct $\left\lfloor x \right\rfloor_\mathrm{bin} = \left\lceil x \right\rceil_\mathrm{bin} = 0$.
As a final piece of minutiae, take $\{2^{\mathbb{N}}\}$ as shorthand for $\{2^n : n \in \mathbb{N} \}$.

\section{Software and Data Availability}
\label{sec:materials}

Supporting software and executable notebooks for this work are available via Zenodo at \url{https://doi.org/10.5281/zenodo.10779240} \citep{moreno2024hsurf}.
DStream algorithm implementations are also published on PyPI in the \texttt{downstream} Python package, where we plan to conduct longer-term, end-user-facing development and maintenance \citep{moreno2024downstream}.
All accompanying materials are provided open-source under the MIT License.

This project benefited significantly from open-source scientific software \citep{2020SciPy-NMeth,harris2020array,reback2020pandas,mckinney-proc-scipy-2010,waskom2021seaborn,hunter2007matplotlib,moreno2023teeplot}.

\section{Steady Algorithm} \label{sec:steady}

The steady criterion seeks to retain data items from time points evenly spread across observed history.
As given in Equation \ref{eqn:steady-cost} in Section \ref{sec:stream-curation-problem}, the steady criterion's cost function is the largest gap size between retained data items, $\mathsf{cost\_steady}(\colorT) = \max\{\colorG_{\colorT}(\colorTbar) : \colorTbar \in [0 \twodots \colorT)\}$.
For a buffer size $\colorS$ and time elapsed $\colorT$, largest gap size can be minimized no lower than
\begin{align}
\mathsf{cost\_steady}(\colorT)
&\geq
\left\lceil
\frac{\colorT - \colorS}{\colorS + 1}
\right\rceil
=
\left\lfloor
\frac{\colorT}{\colorS + 1}
\right\rfloor.
\label{eqn:steady-optimal-gap-size}
\end{align}
This section presents a stream curation algorithm designed to support the steady criterion, guaranteeing maximum gap size no worse than
\begin{align*}
\mathsf{cost\_steady}(\colorT)
&\leq 2 \left\lfloor \frac{\colorT}{\colorS} \right\rfloor_{\mathrm{bin}} - 1.
\end{align*}
Disparity from ideal arises because maintaining uniform gap spacing on an ongoing basis is impossible on account of data item discards merging neighboring gaps.

\subsection{Steady Algorithm Strategy}
\label{sec:steady-strategy}

Figure \ref{fig:hanoi-intuition-steady} overviews the proposed algorithm's core strategy, which revolves around prioritizing data item retention according to the \hv{} of the sequence indices, $\colorH(\colorTbar)$.
Specifically, we aim to keep data items with the largest hanoi values.

It turns out that with all data items $\colorH(\colorTbar) > m$ retained, gap size is at most $\colorg \leq 2^m - 1$.
To understand, imagine discarding items with $\colorH(\colorTbar) = 0$.
This action would drop every other item, and increase gap size from 0 to $\colorg \leq 1$.
Then, removing items with $\colorH(\colorTbar) = 1$ would again drop every other item, and increase gap size to $\colorg \leq 3$.
Continuing this pattern to prune successive hanoi values provides well-behaved transitions that gradually increase gap size while maintaining even spacing.

We thus set out to maintain, for a ratcheting threshold $n(\colorT)$, all items $\colorH(\colorTbar) > n(\colorT)$.
(The threshold $n(\colorT)$ must increase over time to ensure space for new high h.v. data items as we encounter them.)
Formally,
\begin{align*}
\mathsf{goal\_steady}
\coloneq \{
\colorTbar \in [0 \twodots \colorT)
: \colorH(\colorTbar) > n(\colorT)
\}.
\end{align*}
In practice, this requires repeatedly discarding all items with lowest \hv{} $\colorH(\colorTbar) = n(\colorT)$ as time elapses.
Supplementary Lemma \ref{thm:steady-hv-geq-epoch} shows that using a threshold of $n(\colorT) = \colort - 1$ fills available buffer space $\colorS$.%

\subsection{Steady Algorithm Mechanism}
\label{sec:steady-mechanism}

\begin{figure*}[htbp!]
  \flushleft
  \begin{subfigure}[b]{\linewidth}
    \includegraphics[width=\textwidth, trim={0cm 0cm 0cm 1cm}, clip]{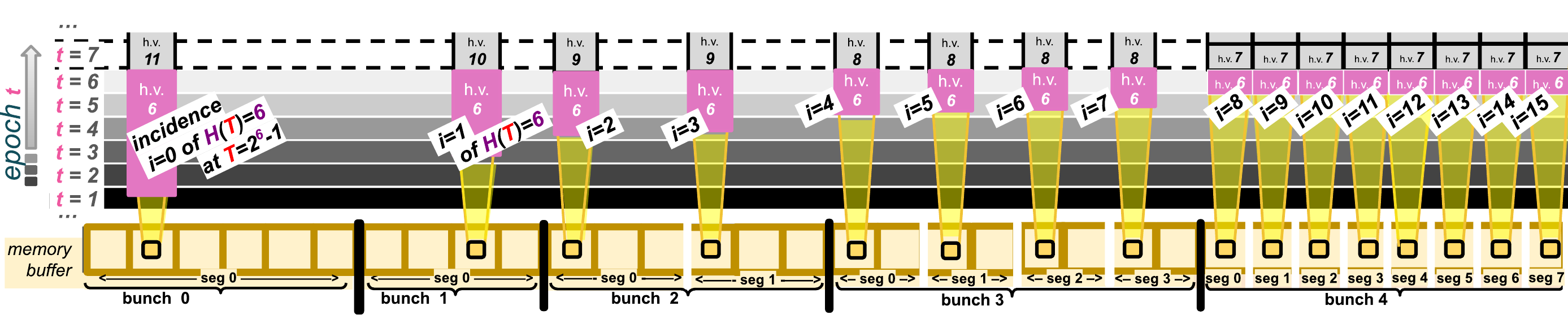}
    \vspace{-4ex}
    \caption{\footnotesize Site selection $\colorK(\colorT)$ for data items of one \hv{}, $\colorh = 6$.}
    \label{fig:hsurf-steady-intuition-diagram}
  \end{subfigure}
\vspace{-5ex}
  \flushright
  \begin{subfigure}[b]{0.98\linewidth}
    \includegraphics[width=0.011\textwidth, trim={0.2cm 2.8cm 31.8cm 2.8cm}, clip]{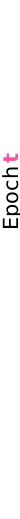}%
    \begin{tikzpicture}

    \node [
      above right,
      inner sep=0] (image) at (0,0) {\includegraphics[width=0.989\textwidth, trim={0.7cm 0cm 1.4cm 5.6cm}, clip]{binder/teeplots/12/reservation-mode=steady-full+surface-size=32+viz=site-reservation-at-ranks-heatmap+ext=.pdf}};

    \begin{scope}[
      x={($0.1*(image.south east)$)},
      y={($0.1*(image.north west)$)}]


\draw[stealth-, ultra thick, gray](0.37,1.1+2.68) -- ++(0, -0.5) node[font=\footnotesize] at (1.62, 2.38+1.1) {incidence $i=0$ of $\colorH(\colorT)=6$ at $\colorT=63$};
\draw[stealth-, ultra thick, gray](2.21+0.6,1.1+3.84) -- ++(0, -0.5) node[font=\footnotesize] at (2.45+0.6, 3.5+1.1) {$i=1$~~~of $\colorH(\colorT)=6$};
\draw[stealth-, ultra thick, gray](3.74-0.3,1.1+4.9) -- ++(0, -0.5) node[font=\footnotesize] at (3.53-0.3, 4.61+1.1) {$i=2$};
\draw[stealth-, ultra thick, gray](4.66-0.3,1.1+4.9) -- ++(0, -0.5) node[font=\footnotesize] at (4.45-0.3, 4.61+1.1) {$i=3$};
\draw[stealth-, ultra thick, gray](5.27,1.1+5.96) -- ++(0, -0.5) node[font=\footnotesize] at (5.06, 5.67+1.1) {$i=4$};
\draw[stealth-, ultra thick, gray](5.88,1.1+5.96) -- ++(0, -0.5) node[font=\footnotesize] at (5.67, 5.71+1.1) {$i=5$};
\draw[stealth-, ultra thick, gray](6.5,1.1+5.96) -- ++(0, -0.5) node[font=\footnotesize] at (6.29, 5.71+1.1) {$i=6$};
\draw[stealth-, ultra thick, gray](7.12,1.1+5.96) -- ++(0, -0.5) node[font=\footnotesize] at (6.91, 5.68+1.1) {$i=7$};
\draw[stealth-, ultra thick, gray](7.41+0.3,1.1+7.06) -- ++(0, -0.5) node[font=\footnotesize] at (7.30+0.3, 6.65+1.1) {$8$};
\draw[stealth-, ultra thick, gray](7.72+0.3,1.1+7.06) -- ++(0, -0.5) node[font=\footnotesize] at (7.61+0.3, 6.65+1.1) {$9$};
\draw[stealth-, ultra thick, gray](8.02+0.3,1.1+7.06) -- ++(0, -0.5) node[font=\footnotesize] at (7.91+0.3, 6.65+1.1) {$10$};
\draw[stealth-, ultra thick, gray](8.32+0.3,1.1+7.06) -- ++(0, -0.5) node[font=\footnotesize] at (8.21+0.3, 6.65+1.1) {$11$};
\draw[stealth-, ultra thick, gray](8.61+0.3,1.1+7.06) -- ++(0, -0.5) node[font=\footnotesize] at (8.5+0.3, 6.65+1.1) {$12$};
\draw[stealth-, ultra thick, gray](8.94+0.3,1.1+7.06) -- ++(0, -0.5) node[font=\footnotesize] at (8.83+0.3, 6.65+1.1) {$13$};
\draw[stealth-, ultra thick, gray](9.25+0.3,1.1+7.06) -- ++(0, -0.5) node[font=\footnotesize] at (9.14+0.3, 6.65+1.1) {$14$};
\draw[stealth-, ultra thick, gray](9.55+0.3,1.1+7.06) -- ++(0, -0.5) node[font=\footnotesize] at (9.44+0.3, 6.65+1.1) {$15$};
    \end{scope}

    \end{tikzpicture}
    \vspace{-5ex}
    \caption{\footnotesize Sites reserved for \hv{} $\colorHcal_{\colort}(\colork)$ over epochs $\colort=0$ to $\colort=7$. H.v.{} $\colorh=6$ annotated as an example.}
    \label{fig:hsurf-steady-intuition-heatmap}
  \end{subfigure}%
  \vspace{-2ex}
  \caption{
    \textbf{Steady algorithm strategy.}
    \footnotesize
    Top panel \ref{fig:hsurf-steady-intuition-diagram} shows sites selected for items with \hv{} $\colorh=6$ from their first occurrence during epoch $\colort=2$ to epoch $\colort=7$, when stored instances of that \hv{} are overwritten.
    Memory buffer sites are shown across the bottom of the schematic.
    Data items' vertical span stretches across time from the epoch when they are stored to the epoch when they are overwritten.
    The first data item with hanoi value $\colorH(\colorT) = \colorh$ is placed in bunch 0 during epoch $\colort=\colorh-4$.
    The next data item with \hv{} $\colorh$ is encountered in the following epoch, and it is placed in bunch 1.
    In epoch $\colort=\colorh-2$, two data items with \hv{} $\colorh$ are encountered and placed into segments within bunch 2.
    Epoch $\colort=\colorh-1$, encounters 4 data items with \hv{} $\colorh-1$ places them in bunch 3's segments.
    In epoch $\colort=\colorh$, eight \hv{} $\colorh$ data items (twice as many) are encountered.
    We place them in bunch 4's one-site segments.
    Finally, during epoch $\colort=\colorh+1$, all further ingested data items with \hv{} $\colorh$ are discarded and all existing stored \hv{} $\colorh$ items are overwritten.
    In this manner, data items with highest \hv{} are retained on a rolling basis to provide uniformly-spaced gaps --- as laid out in Figure \ref{fig:hanoi-intuition-steady}.
    Bottom panel \ref{fig:hsurf-steady-intuition-heatmap} shows \hv{} site reservations $\colorHcal_{\colort}(\colork)$ from epoch $\colort=0$ through $\colort=5$ with buffer size $\colorS=16$.
    Numbering/color coding corresponds to which \hv{} a site is reserved for.
    Black dividers separate bunches; white space divides segments within bunches.
    Annotations highlight the lifecycle of data items with \hv{} $\colorh=6$.
  }
  \label{fig:hsurf-steady-intuition}
\end{figure*}

Each epoch $\colort$, all items with $\colorH(\colorTbar) = \colort - 1$ must be overwritten to make space for new items with \hv{} $\colorh \geq \colort$.
Figure \ref{fig:hsurf-steady-intuition} overviews the layout procedure used to orchestrate replacement of data items with \hv{} $\colorh = \colort - 1$ each epoch.
We divide buffer space into $\colors$ ``bunches,'' themselves divided into ``segments.''
Bunch $i=0$ contains one segment of length $\colors + 1$ sites.
The layout of bunch $i=0$ is a special case, relative to subsequent bunches $i>0$.
For $i > 0$, bunch $i$ contains $2^{i-1}$ segments.
Although segment count increases across bunches $i > 0$, segment length decreases by 1 each bunch as $\colors - i$.
So, segments in the last bunch contain only one site.
With $\colors$ bunches, available buffer space $\colorS$ is filled by this reservation layout,
\begin{align*}
\colors + 1 + \sum_{i=0}^{\colors-1} (\colors - i - 1) \times 2^{i} = 2^{\colors} = \colorS.
\end{align*}

For each hanoi value $\colorh$, if we store one data item $\colorH(\colorTbar) = \colorh$ per segment, data items with a hanoi value $\colorh$ will touch all segments within exactly one bunch over the course of each epoch.
Bunch 0 will contain the first data item with \hv{} $\colorh$, which is encountered in epoch $\colort=\colorh - \colors$.
Bunch 1 contains the one data item with that \hv{} $\colorh$ from epoch $\colort=\colorh - \colors + 2$.
Bunch 2 contains the two data items with \hv{} $\colorh$ from epoch $\colort=\colorh - \colors + 3$.
In general, bunch $i>0$ will contain data items $\{ \colorTbar \in \lBrace \colort = \colorh - \colors + i + 1 \rBrace : \colorH(\colorTbar) = \colorh \}$.
Segment size (decreasing by one each bunch) is arranged so that one instance of all $\colors - i$ \hv's that have ``progressed'' to bunch $i$ can be stored within each segment in that bunch.

\begin{figure*}[htbp!]
  \centering

\begin{minipage}{\textwidth}
  \scriptsize
  \setlength{\tabcolsep}{2.5pt}
  \begin{tabularx}{\textwidth}{
    r
    Y|Y|Y|Y|Y|Y|Y|Y|
    Y|Y|Y|Y|Y Y Y|Y
    |Y|Y|Y|Y|Y|Y|Y
    |Y|Y|Y|Y Y
    }
     { Time $\colorT$} & \textbf{0} & \textbf{1} & \textbf{2} & \textbf{3} & \textbf{4} & \textbf{5} & \textbf{6} & \textbf{7}
    & \textbf{8} & \textbf{9} & \textbf{10} & \textbf{11} & \textbf{12} 
    &  \ldots
    & \textbf{28} & \textbf{29} & \textbf{30} & \textbf{31}
    & \textbf{32} & \textbf{33} & \textbf{34} & \textbf{35}
    & \textbf{36} & \textbf{37} & \textbf{38} & \textbf{39} & \textbf{40}
    & \ldots \\ \hline
     \rowcolor{lightgray!30}
   { Epoch $\colort$} & 0 & 0 & 0 & 0 & 0 & 0 & 0 & 0
    & 0 & 0 & 0 & 0 & 0 
    &  \ldots
    & 0 & 0 & 0 & 0
    & 1 & 1 & 1 & 1
    & 1 & 1 & 1 & 1 & 1
    & \ldots \\
    { \scriptsize$\colorH(\colorT)$} & 0 & 1 & 0 & 2 & 0 & 1 & 0 & 3
    & 0 & 1 & 0 & 2 & 0 
    &  \ldots
    & 0 & 1 & 0 & 5
    & 0 & 1 & 0 & 2
    & 0 & 1 & 0 & 3 & 0
    & \ldots \\
    \hline
     { \scriptsize $\colorK(\colorT)$} & \textbf{0} & \textbf{1} & \textbf{6} & \textbf{2} & \textbf{10} & \textbf{7} & \textbf{13} & \textbf{3}
     & \textbf{16} & \textbf{11} & \textbf{18} & \textbf{8} & \textbf{20} & \ldots
 & \textbf{30} & \textbf{23} & \textbf{31} & \textbf{5} & {\tiny \texttt{\textbf{null\hphantom{}}}}  
 & \textbf{24} & {\tiny \texttt{\textbf{null\hphantom{}}}} & \textbf{16}
 & {\tiny \texttt{\textbf{null\hphantom{}}}} & \textbf{25} & {\tiny \texttt{\textbf{null\hphantom{}}}} & \textbf{10} & {\tiny \texttt{\textbf{null\hphantom{}}}}  &\ldots
  \end{tabularx}
  \vspace{-2ex}
\end{minipage}
\begin{subfigure}{\textwidth}
\caption{\footnotesize Steady policy site selection $\colorK(\colorT)$ with buffer size $\colorS=32$. Ingests marked \nullval{} indicate item discarded without storing.}
\label{fig:hsurf-steady-implementation-site-selection}
\end{subfigure}
\vspace{-3ex}

\begin{subfigure}[b]{\linewidth}
\includegraphics[width=\linewidth]{
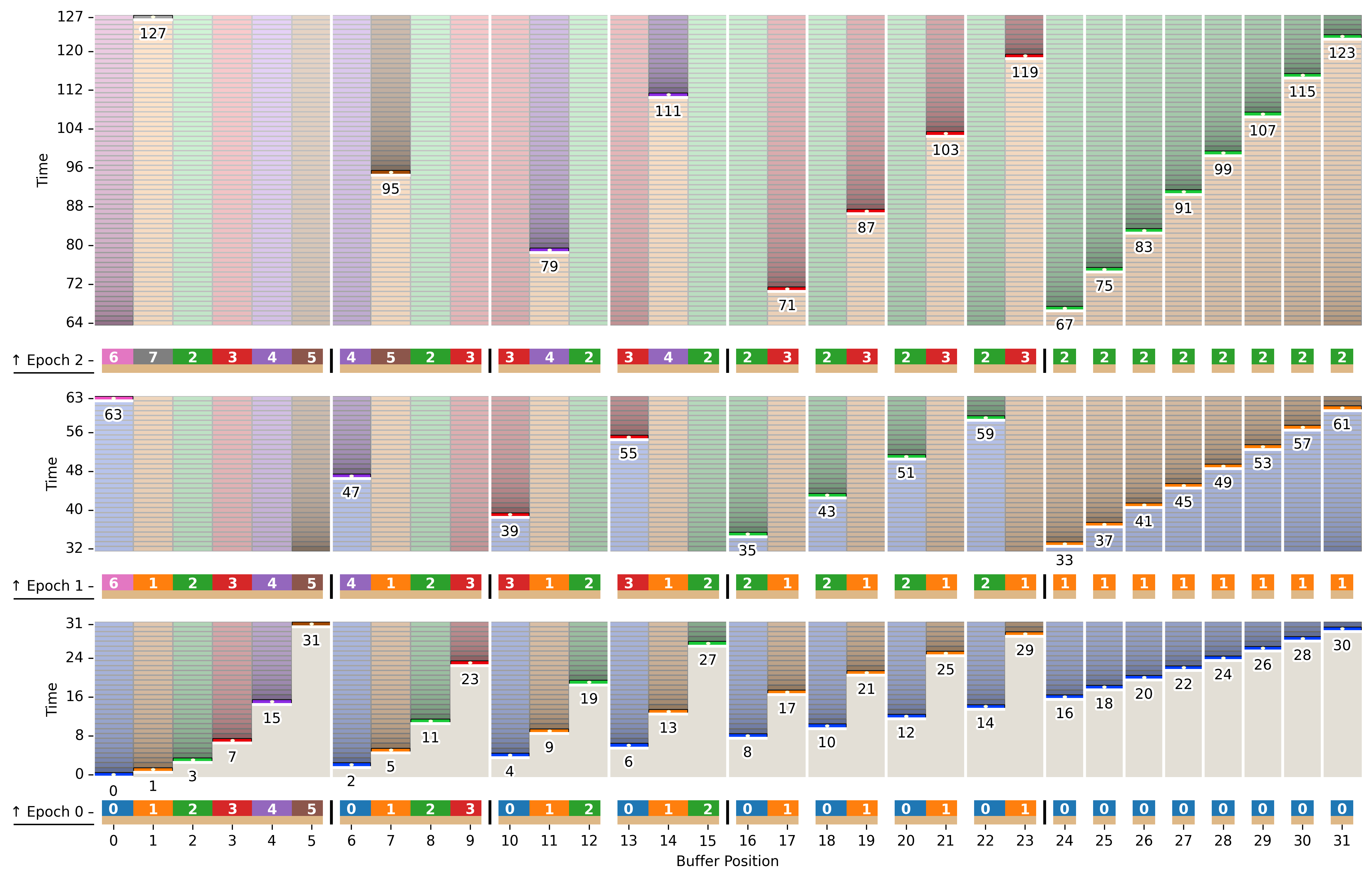}
\vspace{-4.5ex}\caption{\footnotesize
  Buffer composition across time, split by epoch with data items color-coded by hanoi value $\colorH(\colorTbar)$.
}
\label{fig:hsurf-steady-implementation-schematic}
\end{subfigure}

\vspace{0.5ex}
\begin{minipage}[]{\textwidth}
 \vspace{-2pt}
  \begin{subfigure}[t]{0.65\linewidth}
    \vspace{0pt}
    \centering
  \includegraphics[width=0.88\linewidth,clip]{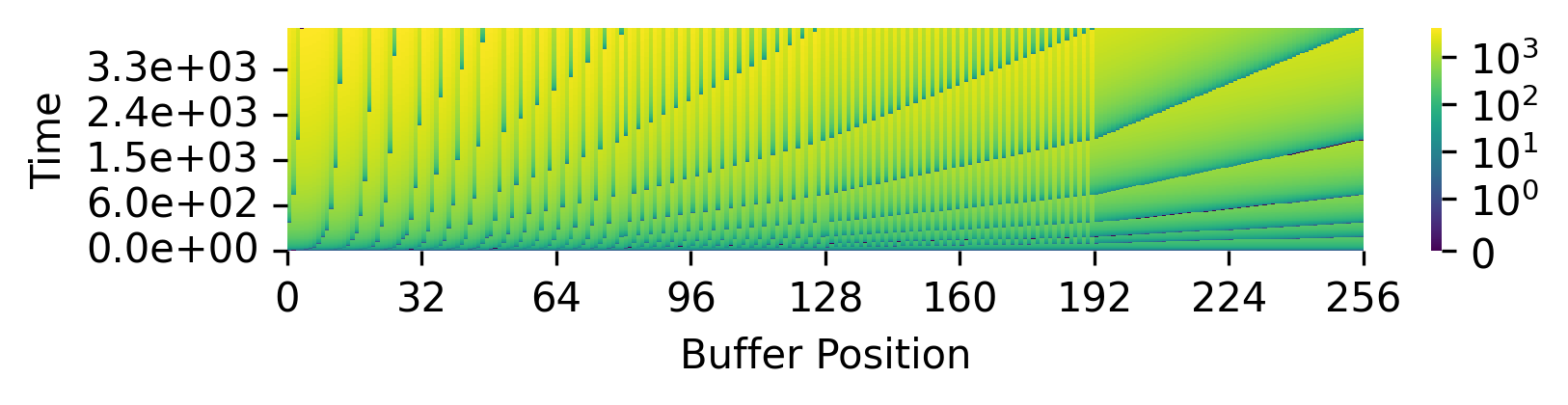}  
  \end{subfigure}%
  \begin{subfigure}[t]{0.35\linewidth}
  \vspace{-2pt}
  \caption{%
    \footnotesize
    Stored data item age across buffer sites for buffer size $\colorS=256$ from $\colorT=0$ to 4,096.
  }
  \label{fig:hsurf-steady-implementation-heatmap}
\end{subfigure}
\end{minipage}

  \vspace{-0.5ex}
   \begin{minipage}[]{\textwidth}
   \vspace{-2pt}
  \begin{subfigure}[t]{0.65\linewidth}
  \vspace{0pt}
    \centering
    \includegraphics[width=0.88\linewidth,clip]{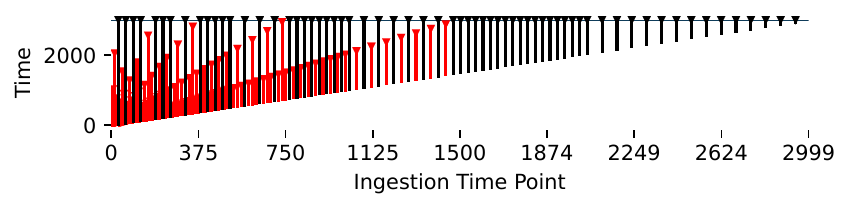}
  \end{subfigure}%
  \begin{subfigure}[t]{0.35\linewidth}
  \vspace{-2pt}
  \caption{%
    \footnotesize
    Data item retention time spans by ingestion time point for buffer size $\colorS=64$ from $\colorT=0$ to 3,000.
  }
  \label{fig:hsurf-steady-implementation-dripplot}
  \end{subfigure}
  \end{minipage}

  \vspace{-0.5ex}
 \begin{minipage}[]{\textwidth}
 \vspace{-2pt}
\begin{subfigure}[t]{0.65\linewidth}
\vspace{0pt}
  \centering
  \includegraphics[width=0.88\linewidth,clip]{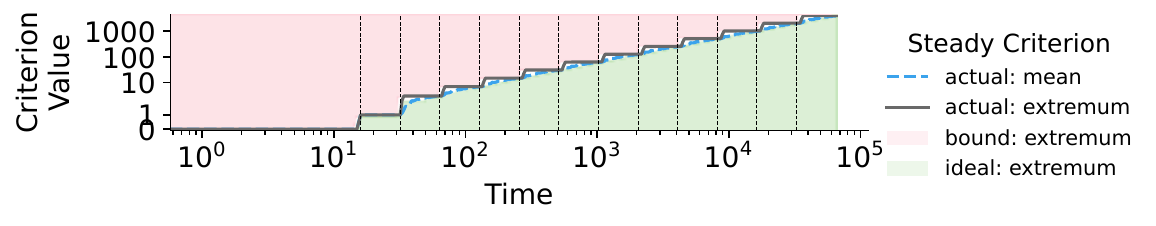}
\end{subfigure}%
\begin{subfigure}[t]{0.35\linewidth}
\vspace{-2pt}
\caption{%
  \footnotesize
  Steady criterion satisfaction across time points for buffer size $\colorS=16$.
}
\label{fig:hsurf-steady-implementation-satisfaction}
\end{subfigure}
\end{minipage}

\vspace{-2ex}\caption{%
  \textbf{Steady algorithm implementation.}
  \footnotesize
  Top panel \ref{fig:hsurf-steady-implementation-site-selection} enumerates initial steady policy site selection on a 32-site buffer.
  Panel \ref{fig:hsurf-steady-implementation-schematic} summarizes how data items are ingested and retained over time within a 32-site buffer, color-coded by data items' hanoi values $\colorH(\colorT)$.
  Between $\colorT=0$ and $\colorT=126$, time is segmented into epochs $\colort=0$, $\colort=1$, and $\colort=2$; strips before  each epoch show hanoi values assigned to each buffer site during that epoch.
  Time increases along the $y$ axis.
  Rectangles with small white ``$\blkhorzoval$'' symbol denote buffer site where the ingested data item from each timestep $\colorT$ is placed.
  Buffer space is split into ``reservation segments.''
  Reservation segments occur in five ``bunches'' --- (1) one 6-site segment, (2) one 4-site segment, (3) two 3-site segments, (4) four 2-site segments, and (5) eight 1-site segments.
  At each epoch, data items are filled into sites newly assigned for their ingestion-order hanoi value from left to right.
  In epoch $\colort=0$, all sites are filled with a first data item.
  During each subsequent epoch $\colort>0$, segments within bunch $i$ each accept one data item with h.v. $\colorh=\colort + \colors - 1 - i$.
  All newly-assigned sites were previously assigned to the overall now-lowest hanoi value $\colorh=\colort - 1$.
  In this way, all instances of the overall lowest hanoi value are overwritten each epoch.
  Heatmap panel \ref{fig:hsurf-steady-implementation-heatmap} shows the evolution of data item age at each site on a 256-bit field over the course of 4,096 time steps.
  Dripplot panel \ref{fig:hsurf-steady-implementation-dripplot} shows retention spans for 3,000 ingested time points.
  Vertical lines span durations between ingestion and elimination for data items from successive time points.
  Time points previously eliminated are marked in red.
  Lineplot panel \ref{fig:hsurf-steady-implementation-satisfaction} shows steady criterion satisfaction on a 16-bit surface over $2^{16}$ timepoints.
  Lower and upper shaded areas are best- and worst-case bounds, respectively.
  }
\label{fig:hsurf-steady-implementation}

\end{figure*}

The particulars of our layout become useful in managing elimination of data items with \hv{} $\colorh = \colort - 1$ during epoch $\colort$.
As noted above, \hv{} $\colorh = \colort + \colors$ will store exactly one data item in bunch 0 during epoch $\colort > 0$.
This is the same number of data items left by \hv{} $\colorh = \colort - 1$ in bunch 0 during earlier epoch $\colort - \colors - 2$.
The same correspondence holds in bunch 1, between \hv{} $\colorh = \colort + \colors - 2$ and \hv{} $\colorh = \colort - 1$.
Indeed, across all bunches $i>0$, the number of data items left by \hv{} $\colorh = \colort + \colors - i$ in bunch $i$ equals those left earlier by \hv{} $\colorh = \colort - 1$.

As shown in Figure \ref{fig:hsurf-steady-implementation}, we can take advantage of one-to-one correspondence between incoming data items and data items of \hv{} $\colorh=\colort-1$ to choreograph clean elimination of \hv{} $\colorh=\colort-1$ by overwrites each epoch.
In determining storage site $\colork$ for ingest $\colorTbar$, we map incoming data items with \hv{} $\colorh \geq \colort$ over items $\colorh = \colort - 1$ slated for elimination by placing them at segment positions $\colorh$ modulus segment size.
The number of \hv{} instances $\colorh = \colorH(\colorTbar)$ already seen, which can be calculated $\mathcal{O}(1)$, identifies the segment where data item $\colorTbar$ should be stored.
Supplementary Lemma \ref{thm:steady-hv-elimination} verifies the behavior of this procedure.

\begin{algorithm}[H]
\caption{Steady algorithm site selection $\colorK(\colorT)$.\\ \footnotesize Supplementary Algorithm \ref{alg:steady-time-lookup} gives steady algorithm site lookup $\colorL(\colorT)$. Supplementary Listings \cref{lst:steady_site_selection.py,lst:steady_time_lookup.py} provide reference Python code.}
\label{alg:steady-site-selection}
\begin{minipage}{0.5\textwidth}
    \hspace*{\algorithmicindent} \textbf{Input:} $\colorS \in \{2^{\mathbb{N}}\},\;\; \colorT \in \mathbb{N}$ \Comment{Buffer size and current logical time}\\
    \hspace*{\algorithmicindent} \textbf{Output:} $\colork \in [0 \twodots \colorS - 1) \cup \{\nullval\}$ \Comment{Selected site, if any}
    \begin{algorithmic}[1]
        \State $\texttt{uint\_t} ~ ~ \colors \gets \Call{BitLength}{\colorS} - 1$
        \State $\texttt{uint\_t} ~ ~ \colort \gets \Call{BitLength}{\colorT} - \colors $ \Comment{Current epoch (or negative)}
        \State $\texttt{uint\_t} ~ ~ \colorh \gets \Call{CountTrailingZeros}{\colorT + 1}$ \Comment{Current \hv{}}
        \If{$\colorh < \colort$} \Comment{If not a top $n(\colorT)$ \hv{}\;\ldots}
        \State \Return \nullval \Comment{\ldots discard without storing}
        \EndIf
        \State $\texttt{uint\_t} ~ ~ i \gets \Call{RightShift}{\colorT, \;\; \colorh + 1}$ \Comment{Hanoi value incidence (i.e., num seen)}
        \If{$i = 0$} \Comment{Special case the 0th bunch}
        \State $\texttt{uint\_t} ~ ~ \colork_b \gets 0$ \Comment{Bunch position}
        \State $\texttt{uint\_t} ~ ~ o \gets 0$ \Comment{Within-bunch offset}
        \State $\texttt{uint\_t} ~ ~ w \gets \colors + 1$ \Comment{Segment width}
        \Else
        \State $\texttt{uint\_t} ~ ~ j \gets \Call{BitFloor}{i} - 1$ \Comment{Num full-bunch segments}
        \State $\texttt{uint\_t} ~ ~ b \gets \Call{BitLength}{j}$ \Comment{Num full bunches}
        \State $\texttt{uint\_t} ~ ~ \colork_b \gets 2^{b}(\colors - b + 1)$ \Comment{Bunch position}
        \State $\texttt{uint\_t} ~ ~ w \gets \colorh - \colort + 1$ \Comment{Segment width}
        \State $\texttt{uint\_t} ~ ~ o \gets w (i - j - 1)$ \Comment{Within-bunch offset}
        \EndIf
        \State $\texttt{uint\_t} ~ ~ p \gets \colorh \bmod w$ \Comment{Within-segment offset}
        \State \Return $\colork_b + o + p$ \Comment{Calculate placement site}
    \end{algorithmic}
\end{minipage}
\end{algorithm}

Algorithm \ref{alg:steady-site-selection} provides a step-by-step listing of site selection calculation $\colorK(\colorT)$, which is $\mathcal{O}(1)$.
Site lookup $\colorL(\colorT)$ is provided in supplementary material, as Algorithm \ref{alg:steady-time-lookup}.
Reference Python implementations appear in Supplementary Listings \ref{lst:steady_site_selection.py} and \ref{lst:steady_time_lookup.py}, as well as accompanying unit tests.
Lookup of ingest time $\colorTbar$ for data item at $\colork$ at time $\colorT$ boils down to decoding its segment/bunch indices and checking whether (if slated) it has yet been replaced during the current epoch $\colort$.
Calculation of site lookup $\colorL(\colorTbar) = \colorTbar_{\colork=0},\;\; \colorTbar_{\colork=1},\;\; \ldots,\;\; \colorTbar_{\colork=\colorS-1}$ proceeds in $\mathcal{O}(\colorS)$ time.

\subsection{Steady Algorithm Criterion Satisfaction}
\label{sec:stready-satisfaction}

In this final subsection, we establish an upper bound on $\mathsf{cost\_steady}(\colorT)$ under the proposed steady curation algorithm.
Figure \ref{fig:hsurf-steady-implementation-satisfaction} plots an example of actual worst gap size over time under this algorithm.

\begin{theorem}[Steady algorithm gap size upper bound]
\label{thm:steady-gap-size}
Under the steady curation algorithm,
\begin{align*}
\mathsf{cost\_steady}(\colorT) \leq 2 \frac{\colorS + 1}{\colorS} \hat{\colorg} + 1,
\end{align*}
where $\hat{\colorg}$ is the optimal lower bound on $\mathsf{cost\_steady}(\colorT)$ given in Equation \ref{eqn:steady-optimal-gap-size}.
\end{theorem}
\begin{proof}
Recall that the time between instances of a data item with \hv{} $\colorH(\colorTbar) = \colorh$ is $2^{\colorh + 1}$ data items.
Recall also that the time elapsed between a \hv{} $\colorh$ and a data item with \hv{} greater than $\colorh$ is $2^{\colorh}$ data items.

Under the proposed algorithm, we retain all data items for hanoi values $\colorh \geq \colort$.
So, retained data items occur at most $2^{\colort}$ time steps apart.
This corresponds to gap size at most $2^{\colort} - 1$.
Finally, we test
\begin{align*}
2 \frac{\colorS + 1}{\colorS} \left\lceil \frac{\colorT - \colorS}{\colorS + 1} \right\rceil + 1
&\stackrel{?}{\geq}
2^{\colort} - 1\\
2 \frac{\colorS + 1}{\colorS} \left\lceil \frac{\colorT - \colorS}{\colorS + 1} \right\rceil
&\stackrel{?}{\geq}
2^{\left\lfloor \log_2(\colorT) \right\rfloor - \colors + 1} - 2 \tag{definition $\colort$, Equation \ref{eqn:epoch-defn}}\\
&\stackrel{?}{\geq}
2^{\left\lfloor \log_2(\colorT) - \log_2(\colorS) \right\rfloor + 1} - 2\\
&\stackrel{?}{\geq}
2\left\lfloor \frac{\colorT}{\colorS} \right\rfloor_{\mathrm{bin}} - 2\frac{\colorS}{\colorS}\\
\frac{\colorS + 1}{\colorS} \frac{\colorT - \colorS}{\colorS + 1}
&\stackrel{?}{\geq}
\frac{\colorT - \colorS}{\colorS}\\
\frac{\colorT - \colorS}{\colorS + 1}
&\stackrel{\checkmark}{\geq}
\frac{\colorT - \colorS}{\colorS + 1}.
\end{align*}
\end{proof}

\section{Stretched Algorithm} \label{sec:stretched}

The stretched criterion favors early data items, targeting a record with gap sizes proportional to data item ingest time $\colorTbar$.
As given in Equation \ref{eqn:stretched-cost} in Section \ref{sec:stream-curation-problem}, the stretched criterion's cost function is the largest ratio of gap size to ingest time,
\begin{align*}
\mathsf{cost\_stretched}(\colorT)
&=
\max\Big\{\frac{\colorG_{\colorT}(\colorTbar)}{\colorTbar} : \colorTbar \in [1 \twodots \colorT)\Big\}.
\end{align*}
For buffer size $\colorS$ and time elapsed $\colorT$, ideal retention would space retained items so that gap size grows proportionally to $\colorTbar$.
Under such a layout, spacing between data items would scale exponentially, and --- counting from zero --- the $n$th retained data item would have ingestion time $\colorT^{n/(\colorS - 1)}$.
Deriving an approximate bound without accounting for discretization effects, gap size ratio would be minimized at best,
\begin{align}
\label{eqn:approx-gap-bound}
\mathsf{cost\_stretched}(\colorT)
&\stackrel{\sim}{\geq}
\colorT^{1/\colorS} - 1.
\end{align}
Lemma \ref{thm:stretched-ideal-strict} works in discretization to prove a strict lower bound on gap size ratio,
\begin{align}
\mathsf{cost\_stretched}(\colorT)
\geq
\frac{
  1
}{
  1 + \colorS
  - \left\lfloor \colorS \log_{\colorT}\Big(
    (\colorT - \colorS)(\colorT^{1/\colorS} - 1) + 1
  \Big)\right\rfloor
}
\geq
\frac{
  1
}{
  1 + \colorS
}.
\label{eqn:stretched-best}
\end{align}
This section proposes a stream curation algorithm tailored to the stretched criterion, achieving gap size ratios no worse than
\begin{align}
\mathsf{cost\_stretched}(\colorT)
&\leq
\min\Big(
  \frac{2^{\colortau + 1}}{\colorS},\;\;
  \frac{2(\colort + \colors)}{\colorS},\;\;
  \frac{4\colort}{\colorS}
\Big)
\label{eqn:stretchednoworse}
\end{align}
over supported epochs $\colort \in [0\twodots\colorS - \colors)$.
This bound ensures gap size ratio ${\colorG_{\colorT}(\colorTbar)}/{\colorTbar} \leq 1$.
More generally, guarantees gap size ratio can be shown guaranteed within a factor of $(1 + 1/\colorS)\times\min(2\colort + 2\colors, \;\; 4\colort, \;\; 2^{\colortau + 1})$ times the optimal bound established in Equation \ref{eqn:stretched-best}.


\subsection{Stretched Algorithm Strategy}
\label{sec:stretched-strategy}

As with the steady algorithm, processing data items $\colorTbar$ based on their hanoi value $\colorH(\colorTbar)$ provides the backbone of our approach to stretched curation.
However, instead of keeping just the $m$ highest \hv{}'s encountered, we approximate a stretched distribution by keeping the first $n$ instances of all encountered \hv{}'s.
Figure \ref{fig:hanoi-intuition-stretched} shows how keeping the first $n$ instances of each \hv{} approximates stretched distribution.

To respect fixed buffer capacity, per-\hv{} capacity $n$ must degrade as we encounter new \hv{}'s.
We thus set out to maintain -- for a declining threshold $n(\colorT)$ --- the set of data items,
\begin{align*}
\mathsf{goal\_stretched}
&\coloneq
\bigcup_{\colorh \geq 0}
\{ \colorTbar = i2^{\colorh + 1} + 2^{\colorh} - 1 \text{ for } i \in [0 \twodots n(\colorT) - 1] : \colorTbar < \colorT \}.
\end{align*}
The set $\mathsf{goal\_stretched}$ is constructed as a union of the smallest $n(\colorT)$ instances of each \hv{}, excluding those not yet encountered at current time $\colorTbar$.
By construction, $\mathsf{goal\_stretched} \subseteq [0 \twodots \colorT)$.
Lemma \ref{thm:stretched-first-n-space} shows setting $n(\colorT) \coloneq 2^{\colors - 1 - \colortau}$ suffices to respect available buffer capacity $\colorS$.


\FloatBarrier  
\begin{figure*}[htbp!]
  \centering
  \begin{subfigure}{0.5\textwidth}
  \includegraphics[width=\textwidth, clip, trim={0 1.25cm 1.25cm 0}]{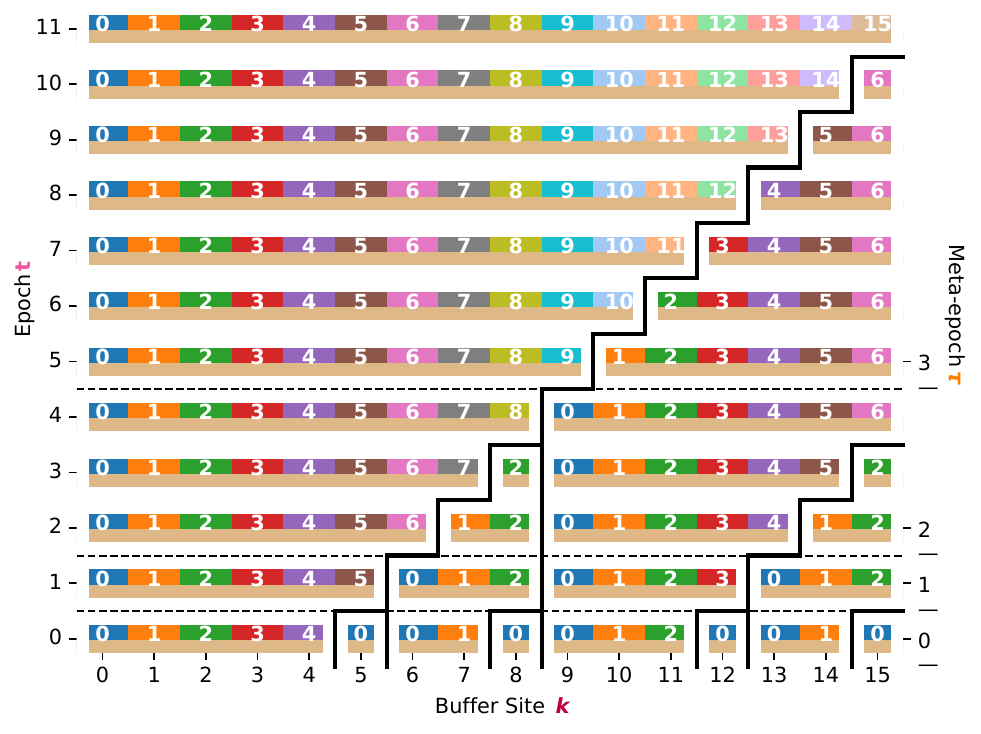}
  \end{subfigure}%
  \begin{subfigure}{0.5\textwidth}
  \includegraphics[width=\textwidth, clip, trim={1.25cm 1.25cm 0 0}]{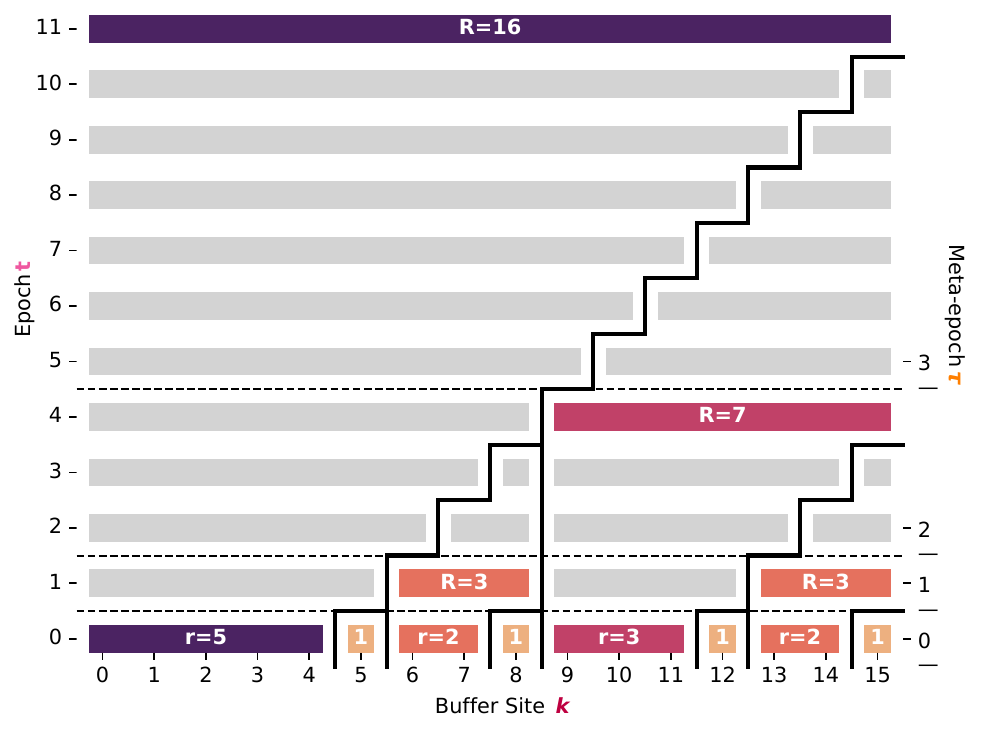}
  \end{subfigure}

\begin{subfigure}{\linewidth}
\begin{tikzpicture}[scale=\linewidth/50cm]

  \begin{scope}


\node[anchor=west, text=black!100] (texta) at (27.0, 3.5) {Bunch $r=1$};
\node[anchor=west, text=black!70] (textb) at (27.0, 2.5) {Bunch $r=2$};
\node[anchor=west, text=black!50] (textc) at (27.0, 1.5) {Bunch $r=3$};
\node[anchor=west, text=black!30] (textd) at (27.0, 0.5) {Bunch $r=5$};

\node[anchor=center,font={\footnotesize}] (textx) at (36.0, 4.5) {\textit{Num Segs.}};
\node[anchor=center,font={\footnotesize}] (textax) at (36.0, 3.5) {\textit{4}};
\node[anchor=center,font={\footnotesize}] (textbx) at (36.0, 2.5) {\textit{2}};
\node[anchor=center,font={\footnotesize}] (textcx) at (36.0, 1.5) {\textit{1}};
\node[anchor=center,font={\footnotesize}] (textdx) at (36.0, 0.5) {\textit{1}};

\node[anchor=center,font={\footnotesize}] (textay) at (40.0, 3.5) {\textit{filled last}};
\node[anchor=center,font={\footnotesize}] (textdy) at (40.0, 0.5) {\textit{filled first}};
\draw[-{stealth[scale=0.2]},shorten >=-1mm, shorten <=-0mm, line width=3pt,draw=MidnightBlue!40](textdy.north) -- (textay.south);

\node[anchor=center,font={\footnotesize}] (textay) at (44.5, 3.5) {\textit{invaded first}};
\node[anchor=center,font={\footnotesize}] (textdy) at (44.5, 0.5) {\textit{invaded last}};
\draw[-{stealth[scale=0.2]},shorten >=-1mm, shorten <=-0mm, line width=3pt,draw=MidnightBlue!45](textay.south) -- (textdy.north);

\filldraw[draw=white,ultra thick,fill=black!100] (8.5,5) -- ++(0.5,-1) -- ++(0.5,1) -- cycle;
\node[anchor=center,text=black!100] (textaa) at (9, 3.5) {S.{} 0};
\filldraw[draw=white,ultra thick,fill=black!100] (12.75,5) -- ++(0.5,-1) -- ++(0.5,1) -- cycle;
\node[anchor=center,text=black!100] (textab) at (13.25, 3.5) {S.{} 1};
\filldraw[draw=white,ultra thick,fill=black!100] (18.5,5) -- ++(0.5,-1) -- ++(0.5,1) -- cycle;
\node[anchor=center,text=black!100] (textac) at (19, 3.5) {S.{} 2};
\filldraw[draw=white,ultra thick,fill=black!100] (22.75,5) -- ++(0.5,-1) -- ++(0.5,1) -- cycle;
\node[anchor=center,text=black!100] (textad) at (23.25, 3.5) {S.{} 3};

\draw[-, line width=2pt, draw=black!100](textaa.east) -- (textab.west);
\draw[-, line width=2pt, draw=black!100](textab.east) -- (textac.west);
\draw[-, line width=2pt, draw=black!100](textac.east) -- (textad.west);
\draw[-, line width=2pt, draw=black!100](textad.east) -- (texta.west);

\filldraw[draw=white,ultra thick,fill=black!70] (10,5) -- ++(1.25,-2) -- ++(1.25,2) -- cycle;
\node[anchor=center,text=black!70] (textba) at (11.25, 2.5) {Seg.{} 0};
\filldraw[draw=white,ultra thick,fill=black!70] (20,5) -- ++(1.25,-2) -- ++(1.25,2) -- cycle;
\node[anchor=center,text=black!70] (textbb) at (21.25, 2.5) {Seg.{} 1};

\draw[-,line width=2pt,draw=black!70](textba.east) -- (textbb.west);
\draw[-,line width=2pt,draw=black!70](textbb.east) -- (textb.west);

\filldraw[draw=white,ultra thick,fill=black!50] (14,5) -- ++(2.15,-3) -- ++(2.15,3) -- cycle;
\node[anchor=center,text=black!50] (textca) at (16, 1.5) {Seg.{} 0};

\draw[-, line width=2pt,draw=black!50](textca.east) -- (textc.west);

\filldraw[draw=white,ultra thick,fill=black!30] (1.5,5) -- ++(3.25,-4) -- ++(3.25,4) -- cycle;
\node[anchor=center,text=black!30] (textda) at (4.75, 0.5) {Segment{} 0};

\draw[-, line width=2pt,draw=black!30](textda.east) -- (textd.west);

\node[anchor=west] (textbunch) at (-1, 4.2) {\phantom{Bunch}};


  \end{scope}
\end{tikzpicture}
\end{subfigure}

\begin{subfigure}{0.5\textwidth}
\includegraphics[width=\textwidth, clip, trim={0 0 1.25cm 11.2cm}]{binder/teeplots/20/surface-size=16+viz=site-reservation-at-ranks-heatmap+ext=.pdf}\vspace{-1ex}
\caption{\footnotesize Sites numbered by reserved hanoi value $\colorHcal_{\colort}(\colork)$ for epochs $\colort=0$ to $\colort=11$.}
\label{fig:hsurf-stretched-intuition-reservations}
\end{subfigure}%
\begin{subfigure}{0.5\textwidth}
\includegraphics[width=\textwidth, clip, trim={1.25cm 0 0 11.2cm}]{binder/teeplots/20/plotter=size+surface-size=16+viz=site-reservation-at-ranks-heatmap+ext=.pdf}\vspace{-1ex}
\caption{\footnotesize Initialized $r$ and mature $R$ reservation segment sizes.}
\label{fig:hsurf-stretched-intuition-reservations-size}
\end{subfigure}\vspace{-1ex}
  \caption{
    \textbf{Stretched algorithm strategy.}
    \footnotesize
    Left panel \ref{fig:hsurf-stretched-intuition-reservations} shows progression of \hv{} reservations $\colorHcal_{\colort}(\colork)$ on a buffer with size $\colorS=16$ across supported epochs $\colort \in [0\twodots\colorS - \colors)$.
    Epoch $\colort$ is indicated on the leftmost axis.
    The rightmost axis, in the right panel, indicates meta-epoch $\colortau$.
    Color coding reflects assigned \hv{}
    Observe, for instance, that four sites, colored dark blue, are reserved for \hv{} $\colorh=0$ during epoch $\colort=0$.
    As shown in the right panel \ref{fig:hsurf-stretched-intuition-reservations-size}, reservation segment bunches are nested recursively, with inner bunches having shorter segments.
    Reservation segments are separated by black lines in both diagrams.
    On the left, inverted triangles schematize the layout of segment bunches, which are nested and discontiguous.
    Bunches are indicated by color code in the right diagram, with segments having same initial size $r$ belonging to the same bunch.
    As epochs elapse, segments grow from initial size $r$ to mature size $R$ and are then invaded to elimination by their larger left neighbor.
    Note how recursive nesting ensures that the shortest segments are eliminated first.
    Note also how sites invaded during the same epoch all share the same reserved \hv{}, causing available sites for that \hv{} to instantaneously halve.
    To ensure it lasts longest, the first item with \hv{} $\colorH(\colorT) = 0$ is placed in the leftmost (and largest) segment $r=5$.
    Subsequent \hv{} instances are accommodated in segment $r=3$, the two $r=2$ segments, and then the four $r=1$ segments.
    Once available segment reservations are filled, subsequent \hv{} instances are discarded without storage.
    Because the segment sizes $r$ mirror the hanoi sequence, expansion of invading segments by one site per epoch $\colort$ ensures buffer space for instances of high \hv{} as they are encountered at later $\colorT$.
  In this manner, layout approximates the first-$n$ \hv{} strategy depicted in Figure \ref{fig:hanoi-intuition-stretched}, with $n$ progressively decreasing as segments are invaded and lost.
  }
  \label{fig:hsurf-stretched-intuition}
\end{figure*}

\subsection{Stretched Algorithm Mechanism}
\label{sec:stretched-mechanism}

\begin{figure*}[htbp!]
  \centering

\begin{minipage}{\textwidth}
  \scriptsize
  \setlength{\tabcolsep}{2.5pt}
  \begin{tabularx}{\textwidth}{
    r
    Y|Y|Y|Y|Y|Y|Y|Y|
    Y|Y|Y|Y|Y Y Y|Y
    |Y|Y|Y|Y|Y|Y|Y
    |Y|Y|Y|Y Y
    }
     { Time $\colorT$} & \textbf{0} & \textbf{1} & \textbf{2} & \textbf{3} & \textbf{4} & \textbf{5} & \textbf{6} & \textbf{7}
    & \textbf{8} & \textbf{9} & \textbf{10} & \textbf{11} & \textbf{12} 
    &  \ldots
    & \textbf{28} & \textbf{29} & \textbf{30} & \textbf{31}
    & \textbf{32} & \textbf{33} & \textbf{34} & \textbf{35}
    & \textbf{36} & \textbf{37} & \textbf{38} & \textbf{39} & \textbf{40}
    & \ldots \\ \hline
   { Epoch $\colort$} & 0 & 0 & 0 & 0 & 0 & 0 & 0 & 0
    & 0 & 0 & 0 & 0 & 0 
    &  \ldots
    & 0 & 0 & 0 & 0
    & 1 & 1 & 1 & 1
    & 1 & 1 & 1 & 1 & 1
    & \ldots \\
     \rowcolor{lightgray!30}
   { Meta-epoch $\colortau$} & 0 & 0 & 0 & 0 & 0 & 0 & 0 & 0
    & 0 & 0 & 0 & 0 & 0 
    &  \ldots
    & 0 & 0 & 0 & 0
    & 1 & 1 & 1 & 1
    & 1 & 1 & 1 & 1 & 1
    & \ldots \\
    { \scriptsize$\colorH(\colorT)$} & 0 & 1 & 0 & 2 & 0 & 1 & 0 & 3
    & 0 & 1 & 0 & 2 & 0 
    &  \ldots
    & 0 & 1 & 0 & 5
    & 0 & 1 & 0 & 2
    & 0 & 1 & 0 & 3 & 0
    & \ldots \\ \hline
     { \scriptsize $\colorK(\colorT)$} & \textbf{0} & \textbf{1} & \textbf{17} & \textbf{2} & \textbf{10} & \textbf{18} & \textbf{25} & \textbf{3}
     & \textbf{7} & \textbf{11} & \textbf{14} & \textbf{19} & \textbf{22} & \ldots
 & \textbf{28} & \textbf{30} & \textbf{31} & \textbf{5} & {\tiny \texttt{\textbf{null\hphantom{}}}}  
 & {\tiny \texttt{\textbf{null\hphantom{}}}} & {\tiny \texttt{\textbf{null\hphantom{}}}}  & \textbf{9} & {\tiny \texttt{\textbf{null\hphantom{}}}}
 & {\tiny \texttt{\textbf{null\hphantom{}}}} & {\tiny \texttt{\textbf{null\hphantom{}}}} & \textbf{13} & {\tiny \texttt{\textbf{null\hphantom{}}}}  &\ldots
  \end{tabularx}
  \vspace{-2ex}
\end{minipage}
\begin{subfigure}{\textwidth}
\caption{\footnotesize Stretched policy site selection $\colorK(\colorT)$ with buffer size $\colorS=32$. Ingests marked \nullval{} indicate item discarded without storing.}
\label{fig:hsurf-stretched-implementation-site-selection}
\end{subfigure}
\vspace{-3ex}

\begin{subfigure}[b]{\linewidth}
\includegraphics[width=\linewidth]{
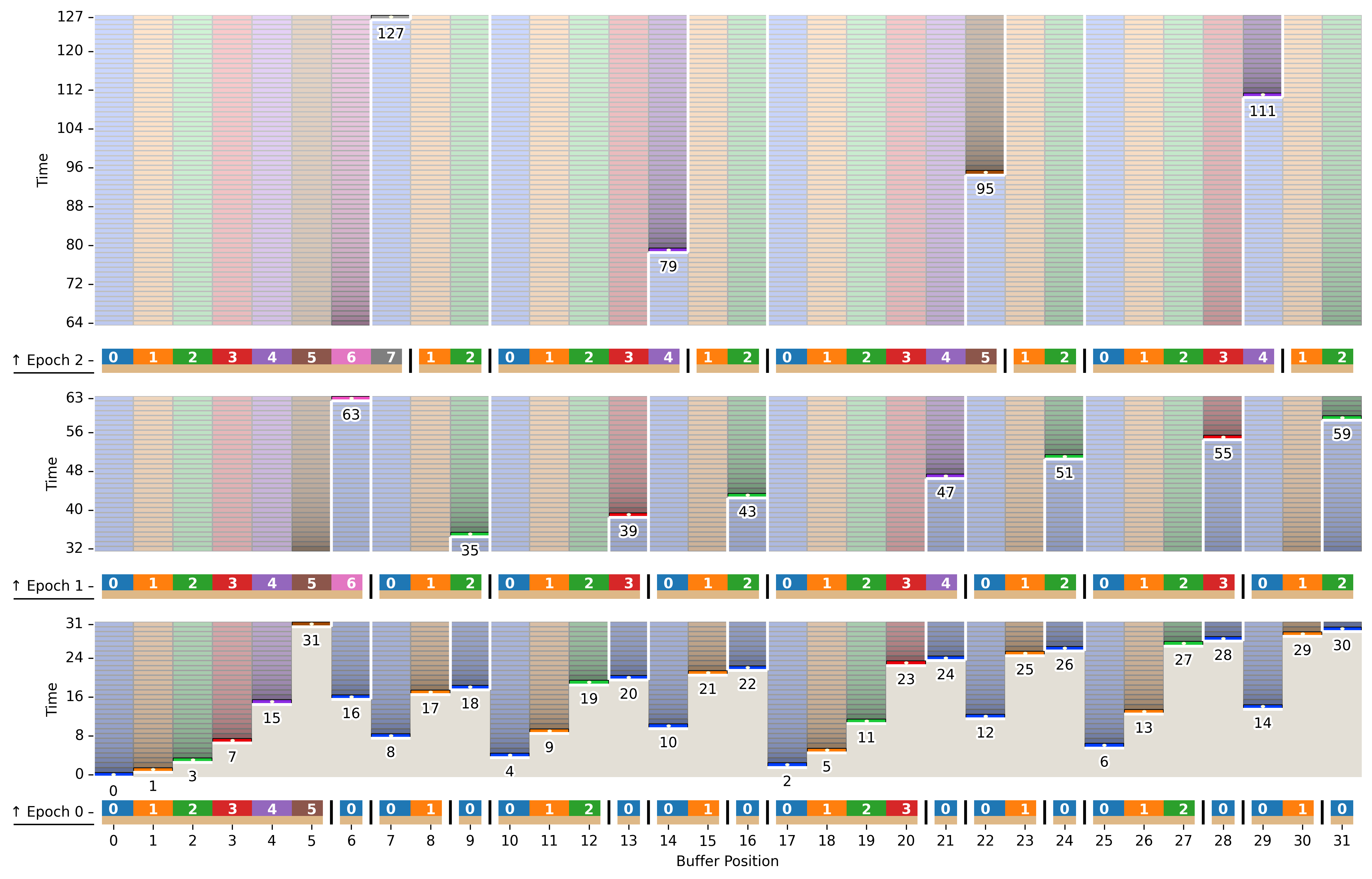}
\vspace{-4.5ex}\caption{\footnotesize
  Buffer composition across time, split by epoch with data items color-coded by hanoi value $\colorH(\colorTbar)$.
}
\label{fig:hsurf-stretched-implementation-schematic}
\end{subfigure}

\vspace{0.5ex}
\begin{minipage}[]{\textwidth}
 \vspace{-2pt}
  \begin{subfigure}[t]{0.65\linewidth}
    \vspace{0pt}
    \centering
  \includegraphics[width=0.88\linewidth,clip]{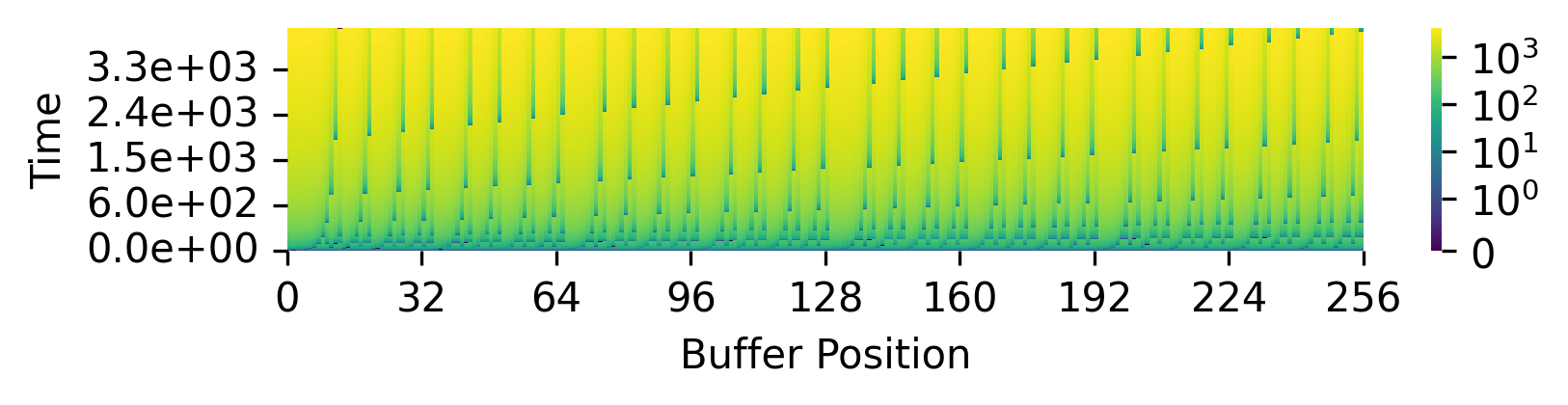}  
  \end{subfigure}%
  \begin{subfigure}[t]{0.35\linewidth}
  \vspace{-2pt}
  \caption{%
    \footnotesize
    Stored data item age across buffer sites for buffer size $\colorS=256$ from $\colorT=0$ to 4,096.
  }
  \label{fig:hsurf-stretched-implementation-heatmap}
\end{subfigure}
\end{minipage}

  \vspace{-0.5ex}
   \begin{minipage}[]{\textwidth}
   \vspace{-2pt}
  \begin{subfigure}[t]{0.65\linewidth}
  \vspace{0pt}
    \centering
    \includegraphics[width=0.88\linewidth,clip]{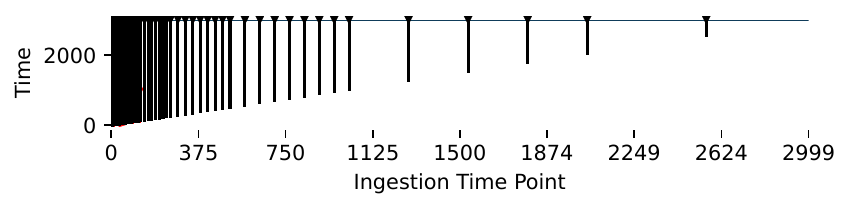}
  \end{subfigure}%
  \begin{subfigure}[t]{0.35\linewidth}
  \vspace{-2pt}
  \caption{%
    \footnotesize
    Data item retention time spans by ingestion time point for buffer size $\colorS=64$ from $\colorT=0$ to 3,000.
  }
  \label{fig:hsurf-stretched-implementation-dripplot}
  \end{subfigure}
  \end{minipage}

  \vspace{-0.5ex}
 \begin{minipage}[]{\textwidth}
 \vspace{-2pt}
\begin{subfigure}[t]{0.65\linewidth}
\vspace{0pt}
  \centering
  \includegraphics[width=0.88\linewidth,clip]{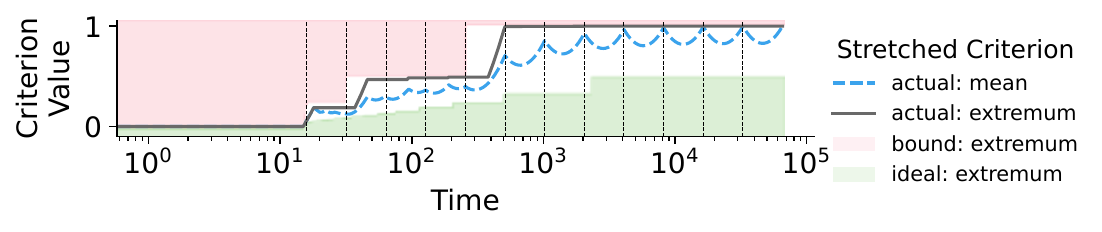}
\end{subfigure}%
\begin{subfigure}[t]{0.35\linewidth}
\vspace{-2pt}
\caption{%
  \footnotesize
  Stretched criterion satisfaction across time points for buffer size $\colorS=16$.
}
\label{fig:hsurf-stretched-implementation-satisfaction}
\end{subfigure}
\end{minipage}

\vspace{-2ex}\caption{%
  \textbf{Stretched algorithm implementation.}
  \footnotesize
  Top panel \ref{fig:hsurf-stretched-implementation-site-selection} enumerates initial stretched policy site selection on a 32-site buffer.
  Panel \ref{fig:hsurf-stretched-implementation-schematic} summarizes how data items are ingested and retained over time within a 32-site buffer, color-coded by data items' hanoi values $\colorH(\colorTbar)$.
  Between $\colorT=0$ and $\colorT=127$, time is segmented into epochs $\colort=0$, $\colort=1$, and $\colort=2$; strips before  each epoch show hanoi values assigned to each buffer site during that epoch.
  Time increases along the $y$ axis.
  Rectangles with small white ``$\blkhorzoval$'' symbol denote buffer site where the ingested data item from each timestep $\colorT$ is placed.
  Reservation segments occur in five recursively nested ``bunches'' --- (1) one 6-site reservation segment, (2) one 4-site reservation segment, (3) two 3-site segments, (4) four 2-site segments, and (5) eight 1-site segments.
  At each epoch, data items are filled into sites newly assigned for their ingestion-order hanoi value from left to right.
  In epoch 0, all sites are filled with a first data item.
  At subsequent epochs, the first site of all innermost-nested segments is ``invaded'' by new high \hv{} sites added to other segments.
  When data items are placed, they remain retained until invaded by a higher-\hv{} data item.
  This process continues until only one segment remains, as shown in Figure \ref{fig:hsurf-stretched-intuition-reservations}.
  Heatmap panel \ref{fig:hsurf-stretched-implementation-heatmap} shows the evolution of data item age at each site on a 256-bit field over the course of 4,096 time steps.
  Dripplot panel \ref{fig:hsurf-stretched-implementation-dripplot} shows retention spans for 3,000 ingested time points.
  Vertical lines span durations between ingestion and elimination for data items from successive time points.
  Time points previously eliminated are marked in red, although in this case they are largely obscured by crowding in small $\colorTbar$.
  Lineplot panel \ref{fig:hsurf-stretched-implementation-satisfaction} shows stretched criterion satisfaction on a 16-bit surface over $2^{16}$ timepoints.
  Lower and upper shaded areas are best- and worst-case bounds, respectively.
  }
\label{fig:hsurf-stretched-implementation}

\end{figure*}

Be reminded that our stretched retention plan is to guarantee space for the first $n(\colorT) =  2^{\colors - 1 - \colortau}$ instances of each hanoi value.
A naive layout might reserve a full $n(\colorT)$ sites for all $2^{\colors + \colort}$ \hv{}'s $\colorh$ that have been encountered by time $\colorT$.
However, such a naive approach would exceed available buffer capacity.
For example, at $\colortau=\colort=0$,
\begin{align*}
2^{\colors - 1 - \colortau} \times 2^{\colors + \colort}
&\geq
2^{2\colors - 1}\\
&\geq
\frac{\colorS^{2}}{2}\\
&> \colorS \text{ for } \colorS > 1.
\end{align*}
A more sophisticated approach will be needed, which we develop next.

\subsubsection{Stretched Algorithm Layout at $\colort,\colortau=0$}

In motivating a more apt stretched layout strategy, begin by restricting focus to epoch $\colort=\colortau=0$, where $\colorT < \colorS$.
Assume that we assign one site to each data item $0\leq \colorTbar < \colorS$ and arrange site assignments according to \hv{} $\colorh = \colorH(\colorTbar)$.
Suppose organization of reserved sites into contiguous segments, with no two items in the same segment allowed to share the same hanoi value $\colorh$.

Under this scheme, we will have at least $\colorS/2$ segments --- one per \hv{} $\colorh=0$ instance encountered.
In constructing segments, half of these $\colorh=0$ segments can be augmented with a site to house one of the $\colorS/4$ \hv{} $\colorh=1$ data items.
We can continue, and further augment $\colorS/8$ segments with \hv{} $\colorh=2$, etc.
Continuing this pattern to place all encountered \hv{} $\colorh\leq\colors$ yields segment sizes that turn out to recapitulate the hanoi sequence.
Special-casing the largest segment, constructed segment sizes can be enumerated as
\begin{align}
\colors + 1,\;\; \colorH(0) + 1,\;\; \colorH(1) + 1,\;\; \ldots,\;\; \colorH(\colorS/2 - 2) + 1.
\label{eqn:stretched-segment-sizes}
\end{align}
These segment sizes can be shown to exactly fill available buffer space $\colorS$,
\begin{align*}
\colors + 1
+  \sum_{\colorh = 0}^{\colors - 2}
2^{\colors - 2 - \colorh} \times (\colorh + 1)
&=
\colors + 1 +
2^{\colors} - \colors - 1
\stackrel{\checkmark}{=}
\colorS.
\end{align*}

Thus far, we have only considered segment sizes --- and not discussed the arrangement of segment order within buffer space $\colorS$.
One naive approach would simply order segments by length, as previously in Section \ref{sec:steady}.
However, as we will see shortly, it turns out that adopting the hanoi sequence's natural ordering (as done in Formula \ref{eqn:stretched-segment-sizes}) better serves our objectives.
The bottom row (``epoch 0'') of Figure \ref{fig:hsurf-stretched-intuition-reservations} shows application of this layout strategy to a 32-site buffer, with segments sized and arranged directly as enumerated in Formula \ref{eqn:stretched-segment-sizes}.

\subsubsection{Stretched Algorithm Layout at $\colort,\colortau\geq1$}

What about $\colorT \geq \colorS$ (i.e., $\colort \geq 1$)?
At epoch $\colort=\colortau=0$, we have successfully guaranteed $n(\colorT) = 2^{\colors - 1 - \colortau} = \colorS / 2$ reserved sites per hanoi value.
To satisfy $\mathsf{goal\_stretched}$ at $\colort=\colortau=1$, we only need to guarantee $n(\colorT) = \colorS/4$ reserved sites --- half as many as at $\colort=\colortau=0$.
So, half of our $S/2$ sites reserved to \hv{} $\colorh=0$ may be freed up.
One way to do this is by releasing all singleton segments containing \textit{only} \hv{} $\colorh=0$.

Because singleton segments intersperse all other segments, their elimination makes space for all remaining segments to ``invade'' by growing one site.
Sticking with our convention of at most one site with each \hv{} $\colorh$ per reservation segment, invading segments accrue space to host an additional high hanoi value data item.
For instance, the largest segment will grow a site reserved to \hv{} $\colorh=\colors +1$.
Two reservation sites will be added for \hv{} $\colorh=\colors - 1$, four for \hv{} $\colorh=\colors - 2$, etc. --- crucially, mirroring the incidence counts for these \hv{}'s during epoch $\colort=1$.

In subsequent epochs $\colort>1$, we can continue dissolving the smallest, innermost-nested reservation segments to grow capacity for new high-\hv{} data items.
Figure \ref{fig:hsurf-stretched-intuition-reservations} shows several steps through this ``invasion'' process on a 32-site buffer.
At final epoch $\colort=\colorS-\colors - 1$ (i.e., $\colorT \approx 2^{\colorS - 1}$), the proposed process of progressive, nested segment subsumption culminates to a single reservation segment containing one site for each \hv{} $0 \leq \colorh < \colorS$.

We will next show that meta-epochs $\colortau$, as defined earlier in Section \ref{sec:notation-metaepoch}, correspond precisely to the timing with which successive inner segments are subsumed.

\begin{lemma}[Meta-epochs $\colortau$ correspond to segment subsumption cycles]
\label{thm:stretched-meta-epoch}

The timing of meta-epoch $\colortau$, defined in Section \ref{sec:notation-metaepoch} as lasting $2^{\colortau} - 1$ epochs for $\colortau\geq1$, corresponds to the time window during which the reservation segments initialized with size $r=\colortau$ are removed through ``invasion.''
\end{lemma}

\begin{proof}

Recall that under the stretched algorithm's proposed layout strategy, buffer space is filled without any overwrites during epoch 0.
Then, during subsequent epochs, half of segments (designated ``invading'' segments) grow by addition of new high-\hv{} sites.
The other half of reservation segments are subsumed one site at a time, successively losing low-\hv{} sites to their invading neighbors.
Note that ``invaded'' segments are not allowed to add high-\hv{} sites --- during the invasion process, they are frozen while being eliminated.

By specification, ``invaded'' segments are always those of smallest remaining size.
Owing to the recursively nested structure of segment layout, smallest-remaining segments are always interspersed every second and always constitute half of active segments.

Because invading segments grow by exactly one buffer site per epoch, the number of epochs $\colort$ it takes for a reservation segment to be invaded to elimination corresponds exactly to the invaded segment's reservation size at invasion outset.
Our proof objective can thus be recast as determining the maximal ``mature''' size $R(r)$ reached by segments initialized size $r$ at epoch $\colort=0$ before frozen for elimination.

Recall from Section \ref{sec:notation-metaepoch} that the duration of meta-epoch $\tau$, $|\colort \in \colortausetoft|$, is $2^{\colortau} - 1$.
For reservation segments with $r=1$ (which are invaded in epoch $\colort=1$ and meta-epoch $\colortau = 1$), our goal is therefore to show $|\colort \in \colortausetoft| = 2^{\colortau} - 1$ matches $R(r)$ by showing $R(r) = 2^{r} - 1$.
As already mentioned, initialized-singleton $r=1$ segments are always invaded first, in epoch $\colort=1$.
Trivially, these segments also have $R(1) = 1$. on account of never having the opportunity to act as an invader.
Segments initialized at size $r=2$ are invaded next.
These segments acted as invader during epoch $\colort=1$, and so grew to size $R(2) = 3$.
Note that $R(1) \stackrel{\checkmark}{=} 2^1 - 1$ and $R(2) \stackrel{\checkmark}{=} 2^2 - 1$.

Subsequent segments $r>2$ grow exponentially --- having invaded segments that themselves already grew by invasion.
For instance, segments $r=3$ begin by invading their singleton neighbors $r=1$ during epoch $\colort=1$.
Then $r=3$ segments invade segments that began as $r=2$.
Thus, for $r=3$,
\begin{align*}
R(3)
&= 3 + R(1) + R(2)\\
&= 3 + 1 + 2 + 1\\
&\stackrel{\checkmark}{=} 2^3 - 1.
\end{align*}

This pattern generalizes across initialized segment sizes $r$ as
\begin{align*}
r + \sum_{j=1}^{r-1} j \times 2^{r-1-j}
&\stackrel{\checkmark}{=} 2^{r} - 1.
\end{align*}

\end{proof}




With relationship between segment subsumption and meta-epoch $\colortau$ thus established, Lemma \ref{thm:stretched-discarded-incidence-count} shows that our scheme maintains reservation layout sufficient to accommodate at least $n(\colorT) = 2^{\colors - 1 - \colortau}$ items of each hanoi value.

\subsubsection{Stretched Algorithm Implementation}
\label{sec:stretched-implementation}

Having determined reservation segment layout strategy, the remaining details of site selection can be addressed succinctly.

As we encounter data items with $\colorH(\colorTbar) = \colorh$, we fill reserved sites for that item's \hv{} in descending order of initialized segment size $r$.
Among same-size segments, we simply fill from left to right.
As invasion eliminates the smallest initialized segments first, this approach guarantees retention of the oldest data items with $\colorH(\colorTbar) = \colorh$.
We may thus reinterpret Lemma \ref{thm:stretched-discarded-incidence-count} as providing guarantees on the first $n$ instances of each \hv{} retained.
Once sites reserved to \hv{} $\colorh$ fill, it is necessary to discard further instances $\colorH(\colorTbar) = \colorh$ without storage.
Figure \ref{fig:hsurf-stretched-implementation-schematic} illustrates the resulting site selection process $\colorK(\colorT)$ over epochs $\colort \in \{0,1,2\}$ on an example buffer, size $\colorS=32$.
Algorithm \ref{alg:stretched-site-selection} provides a step-by-step listing of the stretched site selection procedure $\colorK(\colorT)$, which is $\mathcal{O}(1)$.

\begin{algorithm}[H]
\caption{Stretched algorithm site selection $\colorK(\colorT)$.\\ \footnotesize Supplementary Algorithm \ref{alg:stretched-time-lookup} gives stretched algorithm site lookup $\colorL(\colorT)$. Supplementary Listings \cref{lst:stretched_site_selection.py,lst:stretched_time_lookup.py} provide reference Python code.}
\label{alg:stretched-site-selection}
\begin{minipage}{0.53\textwidth}
    \hspace*{\algorithmicindent} \textbf{Input:} $\colorS \in \{2^{\mathbb{N}}\},\;\; \colorT \in \mathbb{N}$ \Comment{Buffer size and current logical time}\\
    \hspace*{\algorithmicindent} \textbf{Output:} $\colork \in [0 \twodots \colorS - 1) \cup \{\nullval\}$ \Comment{Selected site, if any}
    \begin{algorithmic}[1]
        \State $\texttt{uint\_t} ~ ~ \colors \gets \Call{BitLength}{\colorS} - 1$
        \State $\texttt{uint\_t} ~ ~ \colort \gets \max(0,\;\; \Call{BitLength}{\colorT} - \colors)$ \Comment{Current epoch}
        \State $\texttt{uint\_t} ~ ~ \colorh \gets \Call{CountTrailingZeros}{\colorT + 1}$ \Comment{Current \hv{}}
        \Statex
        \State $\texttt{uint\_t} ~ ~ i \gets \Call{RightShift}{\colorT, \;\; \colorh + 1}$ \Comment{Hanoi value incidence (i.e., num seen)}
        \State $\texttt{bool\_t} ~ ~ \epsilon_{\colortau} \gets \Call{BitFloorSafe}{2\colort} \;\; > \;\; \colort + \Call{BitLength}{\colort}$ \Comment{Correction factor}
        \State $\texttt{uint\_t} ~ ~ \colortau \gets  \Call{BitLength}{\colort} - \Call{I}{\epsilon_{\colortau}}$ \Comment{Current meta-epoch}
        \State $\texttt{uint\_t} ~ ~ B \gets \min(1,\;\; \Call{RightShift}{\colorS, \;\; \colortau + 1})$ \Comment{Num bunches available to \hv}
        \If{$i \geq B$} \Comment{If seen more than sites reserved to \hv{}\;\ldots}
            \State \Return \nullval \Comment{\ldots discard without storing}
        \EndIf
        \Statex
        \State $\texttt{uint\_t} ~ ~ b_l \gets i$ \Comment{Logical bunch index, in order filled \ldots}
        \Statex \Comment{\ldots i.e., increasing nestedness/decreasing init size $r$}
        \Statex
        \Statex \Comment{Need to calculate physical bunch index\ldots}
        \Statex \Comment{\ldots i.e., among bunches left-to-right in buffer space}
        \Statex
        \State $\texttt{uint\_t} ~ ~ v \gets \Call{BitLength}{b_l}$ \Comment{Nestedness depth level for physical bunch}
        \State $\texttt{uint\_t} ~ ~ w \gets \Call{RightShift}{\colorS, \;\; v} \;\; \times \;\;\Call{I}{v > 0}$ \Comment{Num bunches spaced between bunches in same nest level}
        \State $\texttt{uint\_t} ~ ~ o \gets 2w$  \Comment{Offset of nestedness level in physical bunch order}
        \State $\texttt{uint\_t} ~ ~ p \gets b_l - \Call{BitFloorSafe}{b_l}$ \Comment{Bunch position within nestedness level}
        \State $\texttt{uint\_t} ~ ~ b_p \gets o + wp$ \Comment{Physical bunch index\ldots}
        \Statex \Comment{\ldots i.e., in left-to-right buffer space ordering}
        \Statex
        \Statex \Comment{Need to calculate buffer position of $b_p$\textsuperscript{th} bunch}
        \Statex
        \State $\texttt{uint\_t} ~ ~ \epsilon_{\colork_b} = \Call{I}{b_l > 0}$  \Comment{Correction factor, 0\textsuperscript{th} bunch (i.e., bunch $r=\colors$ at site $\colork=0$)}
        \State $\texttt{uint\_t} ~ ~ \colork_b \gets \Call{BitCount}{2b_p +(2\colorS - b_p)} - 1 - \epsilon_{\colork_b}$  \Comment{Site index of bunch}
        \Statex
        \State \Return $\colork_b + \colorh$ \Comment{Calculate placement site, \hv{} $\colorh$ is offset within bunch}
    \end{algorithmic}
\end{minipage}
\end{algorithm}

Stretched site lookup $\colorL(\colorT)$ is provided in supplementary material, as Algorithm \ref{alg:stretched-time-lookup}.
Reference Python implementations appear in Supplementary Listings \ref{lst:stretched_site_selection.py} and \ref{lst:stretched_time_lookup.py}, as well as accompanying tests.
The data item $\colorTbar$ present at buffer site $\colork$ at time $\colorT$ can be determined by decoding that site's segment index and checking whether (if slated) it has yet been replaced during the current epoch $\colort$.
Both site selection and ingest time calculation can be accomplished through fast $\mathcal{O}(1)$ binary operations (e.g., bit mask, bit shift, count leading zeros, popcount).

\subsection{Stretched Algorithm Criterion Satisfaction}
\label{sec:stretched-satisfaction}

In this final subsection, we establish an upper bound on $\mathsf{cost\_stretched}(\colorT)$ for a buffer of size $\colorS$ at time $\colorT$ under the proposed stretched curation algorithm.

\begin{theorem}[Stretched algorithm gap size ratio upper bound]
\label{thm:stretched-gap-size}
Under the stretched curation algorithm, gap size ratio is bounded according to Equation \ref{eqn:stretchednoworse}.
\end{theorem}
\begin{proof}

Lemma \ref{thm:gap-size-ratio-stretched} establishes that gap size ratio is bounded below by $1/n$ if the first $n$ instances of each \hv{} $\colorh$ are retained.
Substituting expressions for the number of sites reserved per \hv{} derived in Lemma \ref{thm:stretched-discarded-incidence-count} and Corrolary \ref{thm:stretched-reservation-count} gives
\begin{align*}
  \mathsf{cost\_stretched}(\colorTbar)
  &\leq
  \Big[
    \max\Big(
      2^{\colors - 1 - \colortau},\;\;
      \frac{\colorS}{2(\colort + \colors)},\;\;
      \frac{\colorS}{4\colort}
    \Big)
  \Big]^{-1}.
\end{align*}

Simplifying resolves the result.

\end{proof}

During early epoch $\colort = 1$, $\mathsf{cost\_stretched}(\colorT) \leq 4/\colorS$.
Likewise, at the opposite extremum, $\mathsf{cost\_stretched}(\colorT) \leq 1$ during the last supported meta-epoch $\colortau = \colors - 1$.
Figure \ref{fig:hsurf-stretched-implementation-satisfaction} shows algorithm performance on the stretched criterion for buffer size $\colorS=16$, $\colorT \in [0\twodots 2^{\colorS} - 1)$.

\section{Tilted Algorithm} \label{sec:tilted}

The tilted criterion favors recent data items, mandating a record spaced proportionally to time elapsed since ingest, $\colorT - 1 - \colorTbar$.
This is opposite to the stretched criterion, which favors early data items.
As given in Equation \ref{eqn:tilted-cost} in Section \ref{sec:stream-curation-problem}, the tilted criterion's cost function is the largest ratio of gap size to ingest time,
\begin{align*}
\mathsf{cost\_tilted}(\colorT)
&=
\max\Big\{\frac{\colorG_{\colorT}(\colorTbar)}{\colorT - 1 - \colorTbar} : \colorTbar \in [0 \twodots \colorT-1)\Big\}.
\end{align*}

The approximate lower bound on best-case gap size ratio provided in Equation \ref{eqn:approx-gap-bound} for the stretched curation can also be applied to tilted curation, as can the strict bound on best-case gap size ratio accounting for discretization effects established in Theorem \ref{thm:stretched-ideal-strict}.
In this section, we present a stream curation algorithm tailored to the tilted criterion, achieving maximum gap size ratio no worse than
\begin{align}
  \mathsf{cost\_tilted}(\colorT)
  &\leq
  \frac{
    1
  }{
    \max\Big(
      \frac{\colorS}{2(\colort + \colors)},\;\;
      \frac{\colorS}{4\colort},\;\;
      \frac{\colorS}{2^{\colortau + 1}}
    \Big)
    - 1/2
  }
  \text{ for }
  \colorTbar < \colorT - 1
  \label{eqn:tilted-gap-size-bound}
\end{align}
over supported epochs $\colort \in [0\twodots\colorS - \colors)$.
Because $\min(2\colort + 2\colors,\;\; 4\colort,\;\; 2^{\colors - \colortau - 1}) \leq \colorS$, tilted gap size ratio is no greater than a factor of $2(1 + 1/\colorS)\times\min(2\colort + 2\colors, \;\; 4\colort, \;\; 2^{\colortau + 1})$ times the optimal bound established in Equation \ref{eqn:stretched-best}.
Additionally, gap size ratio is bounded $\mathsf{cost\_tilted}(\colorT) \leq 2$.


\subsection{Tilted Algorithm Strategy}
\label{sec:tilted-strategy}

\begin{figure*}[htbp!]
  \centering

  \hfill
\begin{subfigure}{0.43\textwidth}
  \includegraphics[width=\textwidth]{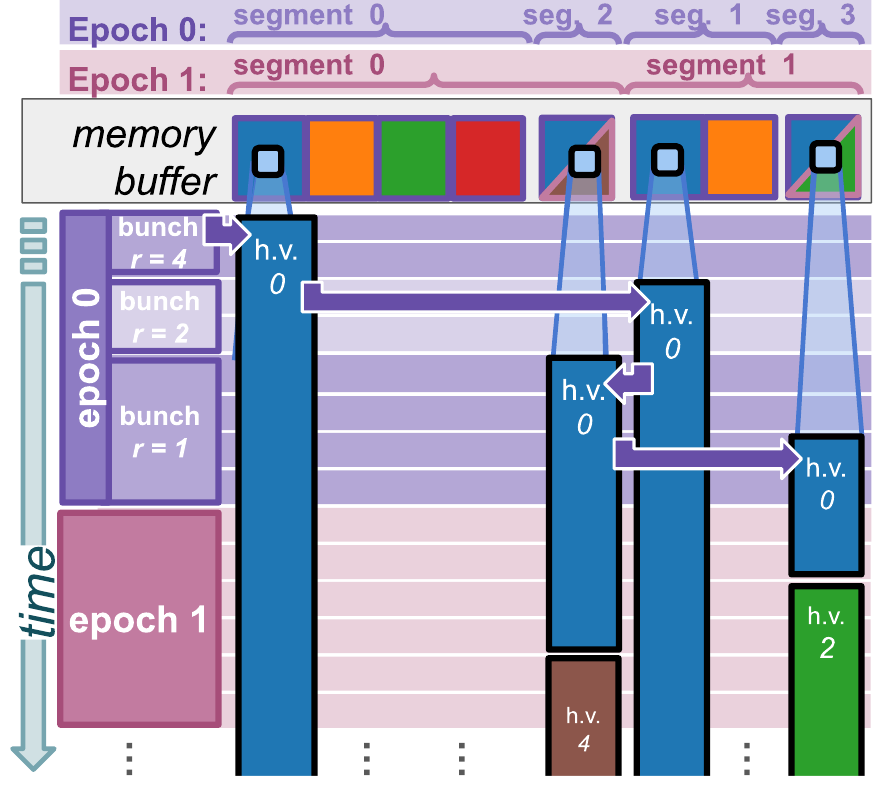}
  \caption{\footnotesize site selection for \hv{} $\colorh=0$ under \textit{\textbf{stretched}} algorithm}
  \label{fig:hsurf-tilted-intuition-site-selection}
\end{subfigure}
\hfill
\begin{subfigure}{0.43\textwidth}
  \includegraphics[width=\textwidth]{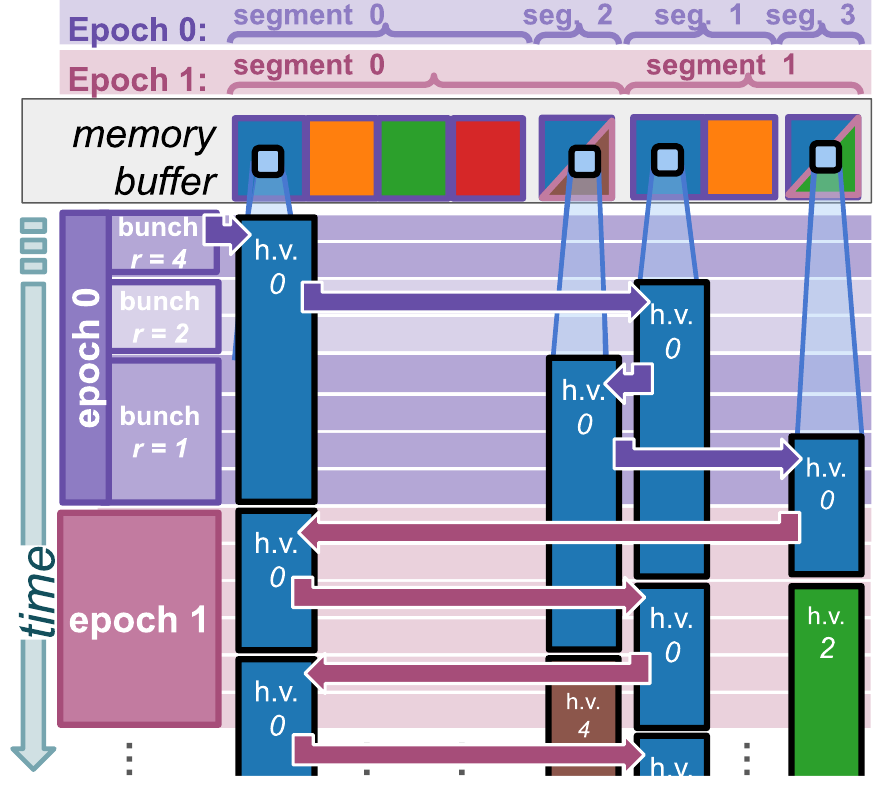}
  \caption{\footnotesize site selection for \hv{} $\colorh=0$ under \textit{\textbf{tilted}} algorithm}
  \label{fig:hsurf-tilted-intuition-tilted}
\end{subfigure}
\hfill

  \caption{
    \textbf{Tilted algorithm strategy.}
    \footnotesize
    Tilted algorithm strategy relates closely to stretched algorithm strategy.
    In particular, the tilted algorithm uses \hv{} reservation layout $\colorHcal_{\colort}(\colork)$ exactly identical to the stretched algorithm ( shown in Figure \ref{fig:hsurf-stretched-intuition}).
    As contrasted between left and right panels, the tilted and stretched algorithms differ in how they handle \hv{} instances after available reservation segments have been filled.
    Schematics show site selection strategy for items with \hv{} $\colorH(\colorTbar) = 0$ on a buffer of size $\colorS=8$.
    Whereas the stretched algorithm discards these items, the tilted algorithm treats reserved segments as a ring buffer by ``wrapping around'' and beginning again from the largest (and leftmost) segment $r=\colors$.
    In this way, the most recent $n$ (as opposed to the first $n$) data items corresponding to each hanoi value are kept, satisfying the tilted retention objective depicted in Figure \ref{fig:hanoi-intuition-tilted}.
    So, placements for a particular \hv{} cycle around available reservation sites, and then continue cycling around remaining sites after the \hv{} is invaded and half of reservation sites for that \hv{} are ceded.
  }
  \label{fig:hsurf-tilted-intuition}
\end{figure*}

The retention strategy for the tilted algorithm strongly resembles that of the stretched algorithm.
Recall that under the stretched algorithm the first $n(\colorT)$ data items of each \hv{} $\colorH(\colorTbar)$ are retained, with $n(\colorT)$ decreasing so as to shift from many copies of few encountered hanoi values to few copies of many encountered hanoi values.
Under the tilted algorithm, we instead keep the \textit{last} $n(\colorT)$ data items of each hanoi value.
Figure \ref{fig:hanoi-intuition-tilted} shows how keeping the last $n$ instances of each \hv{} approximates tilted distribution.

We thus set out to maintain --- for a declining threshold $n(\colorT)$ --- the set of data items,
\begin{align}
\\
\mathsf{goal\_tilted}
&\coloneq
\bigcup_{\colorh \geq 0}
\{ \colorTbar =
\eqnmarkbox[gray]{maxhanoi}{
  \left\lfloor
  \frac{\colorT - 2^{\colorh}}{2^{\colorh + 1}}
  \right\rfloor
  2^{\colorh + 1}
  + 2^{\colorh}
  - 1
}
- i2^{\colorh + 1} \text{ for } i \in [0 \twodots n(\colorT) - 1] : 0 \leq \colorTbar < \colorT \}.
\annotate[yshift=1em]{above,left}{maxhanoi}{$\max\{
  \colorTbar' \in [0 \twodots \colorT) : \colorH(\colorTbar') = \colorh
\}$}
\label{eqn:goal-tilted-set}
\end{align}

It can be shown analogously to the stretched algorithm's Lemma \ref{thm:stretched-first-n-space} that setting $n(\colorT) \coloneq 2^{\colors - 1 - \colortau}$ suffices to respect available buffer capacity $\colorS$ under the tilted algorithm.

\subsection{Tilted Algorithm Mechanism}
\label{sec:tilted-mechanism}

Because the tilted algorithm, like the stretched algorithm, also approximates an equal-$n$-per-\hv{} scheme, hanoi value reservation layout is maintained identically to the stretched algorithm's segment-based scheme.
Refer to Section \ref{sec:stretched-mechanism} for a detailed description of this \hv{} reservation layout, and how it unfolds across epochs $0 \leq \colort \leq \colorS - \colors$.

A pertinent result of the stretched layout is that at least $2^{\colors - 1 - \colortau}$ data item instances of each \hv{} are retained (Lemma \ref{thm:stretched-discarded-incidence-count}).
However, unlike the stretched algorithm, for the tilted algorithm we wish to keep the \textit{last} $n$ rather than the \textit{first} $n$ instances of each hanoi value.
We can do that by continuing to write data items for each \hv{} into buffer sites reserved for that \hv{} after they initially fill --- overwriting older instances of the \hv{} to keep a ``ring buffer'' of fresh \hv{} instances.

Supplemental materials prove several results related to the tilted algorithm's ring buffer mechanism, including that fill cycles align evenly to epoch and meta-epoch transitions (Lemma \ref{thm:tilted-last-touched}).
These results build to Lemma \ref{thm:tilted-most-recent-retained}, which confirms that our strategy always preserves the last $2^{\colors - 1 - \colortau}$ instances of each hanoi value.
We take particular care in considering transitions where the ``ring buffer'' of sites reserved to a \hv{} is halved by growth of invading segments.

\subsection{Tilted Algorithm Implementation}
\label{sec:tilted-implementation}

\begin{figure*}[htbp!]
  \centering

\begin{minipage}{\textwidth}
  \scriptsize
  \setlength{\tabcolsep}{2.5pt}
  \begin{tabularx}{\textwidth}{
    r
    Y|Y|Y|Y|Y|Y|Y|Y|
    Y|Y|Y|Y|Y Y Y|Y
    |Y|Y|Y|Y|Y|Y|Y
    |Y|Y|Y|Y Y
    }
     { Time $\colorT$} & \textbf{0} & \textbf{1} & \textbf{2} & \textbf{3} & \textbf{4} & \textbf{5} & \textbf{6} & \textbf{7}
    & \textbf{8} & \textbf{9} & \textbf{10} & \textbf{11} & \textbf{12} 
    &  \ldots
    & \textbf{28} & \textbf{29} & \textbf{30} & \textbf{31}
    & \textbf{32} & \textbf{33} & \textbf{34} & \textbf{35}
    & \textbf{36} & \textbf{37} & \textbf{38} & \textbf{39} & \textbf{40}
    & \ldots \\ \hline
   { Epoch $\colort$} & 0 & 0 & 0 & 0 & 0 & 0 & 0 & 0
    & 0 & 0 & 0 & 0 & 0 
    &  \ldots
    & 0 & 0 & 0 & 0
    & 1 & 1 & 1 & 1
    & 1 & 1 & 1 & 1 & 1
    & \ldots \\
     \rowcolor{lightgray!30}
   { Meta-epoch $\colortau$} & 0 & 0 & 0 & 0 & 0 & 0 & 0 & 0
    & 0 & 0 & 0 & 0 & 0 
    &  \ldots
    & 0 & 0 & 0 & 0
    & 1 & 1 & 1 & 1
    & 1 & 1 & 1 & 1 & 1
    & \ldots \\
    { \scriptsize$\colorH(\colorT)$} & 0 & 1 & 0 & 2 & 0 & 1 & 0 & 3
    & 0 & 1 & 0 & 2 & 0 
    &  \ldots
    & 0 & 1 & 0 & 5
    & 0 & 1 & 0 & 2
    & 0 & 1 & 0 & 3 & 0
    & \ldots \\ \hline
     { \scriptsize $\colorK(\colorT)$} & \textbf{0} & \textbf{1} & \textbf{17} & \textbf{2} & \textbf{10} & \textbf{18} & \textbf{25} & \textbf{3}
     & \textbf{7} & \textbf{11} & \textbf{14} & \textbf{19} & \textbf{22} & \ldots
 & \textbf{28} & \textbf{30} & \textbf{31} & \textbf{5} & \textbf{0}
 & \textbf{1} & \textbf{17} & \textbf{9} & \textbf{10}
 & \textbf{18} & \textbf{25} & \textbf{13} & \textbf{7}  &\ldots
  \end{tabularx}
\vspace{-2ex}
\end{minipage}
\begin{subfigure}{\textwidth}
\caption{\footnotesize Tilted policy site selection $\colorK(\colorT)$ with buffer size $\colorS=32$.}
\label{fig:hsurf-tilted-implementation-site-selection}
\end{subfigure}
\vspace{-3ex}

\begin{subfigure}[b]{\linewidth}
\includegraphics[width=\linewidth]{
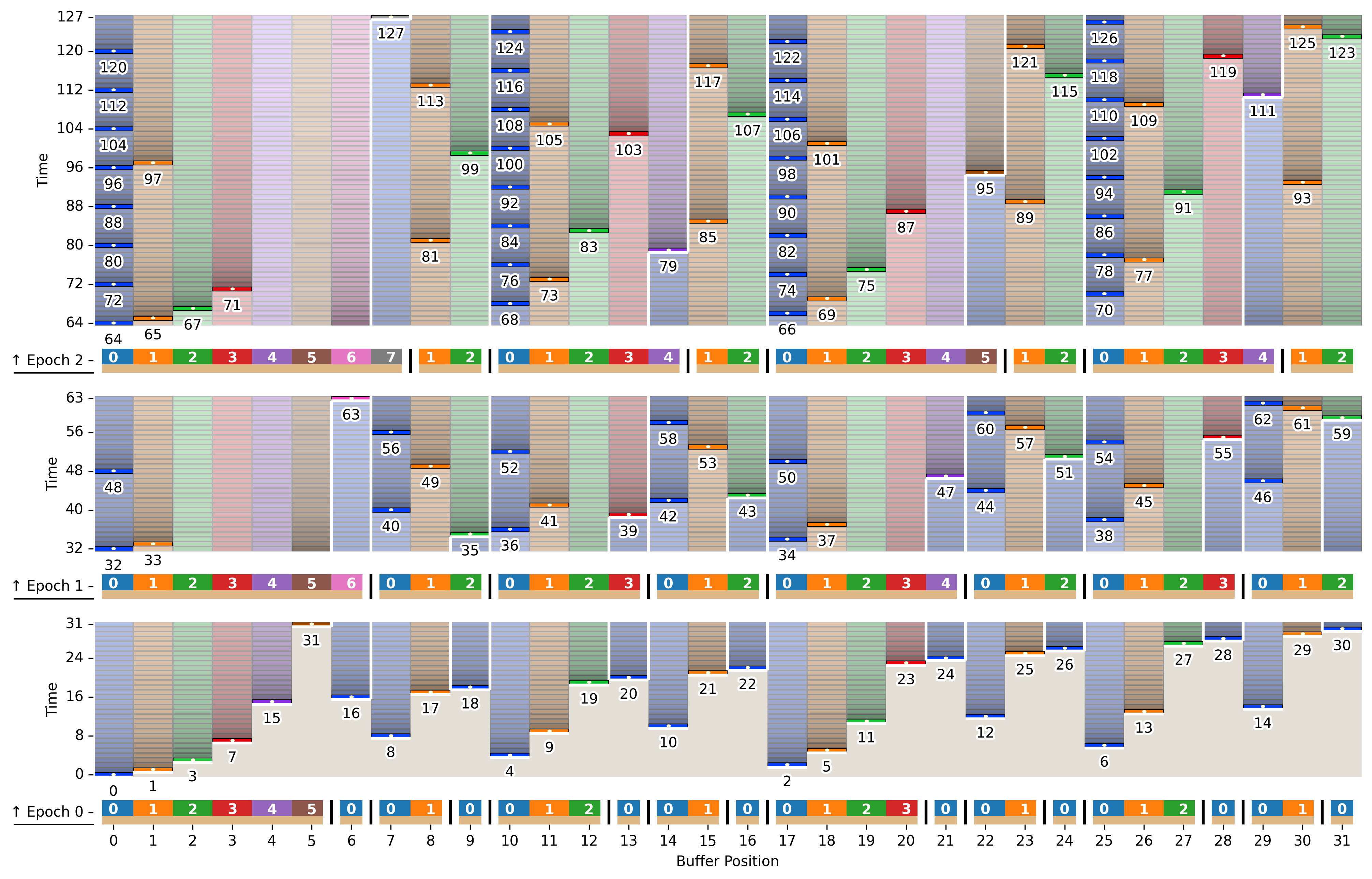}
\vspace{-4.5ex}\caption{\footnotesize
  Buffer composition across time, split by epoch with data items color-coded by hanoi value $\colorH(\colorTbar)$.
}
\label{fig:hsurf-tilted-implementation-schematic}
\end{subfigure}

\vspace{0.5ex}
\begin{minipage}[]{\textwidth}
 \vspace{-2pt}
  \begin{subfigure}[t]{0.65\linewidth}
    \vspace{0pt}
    \centering
  \includegraphics[width=0.88\linewidth,clip]{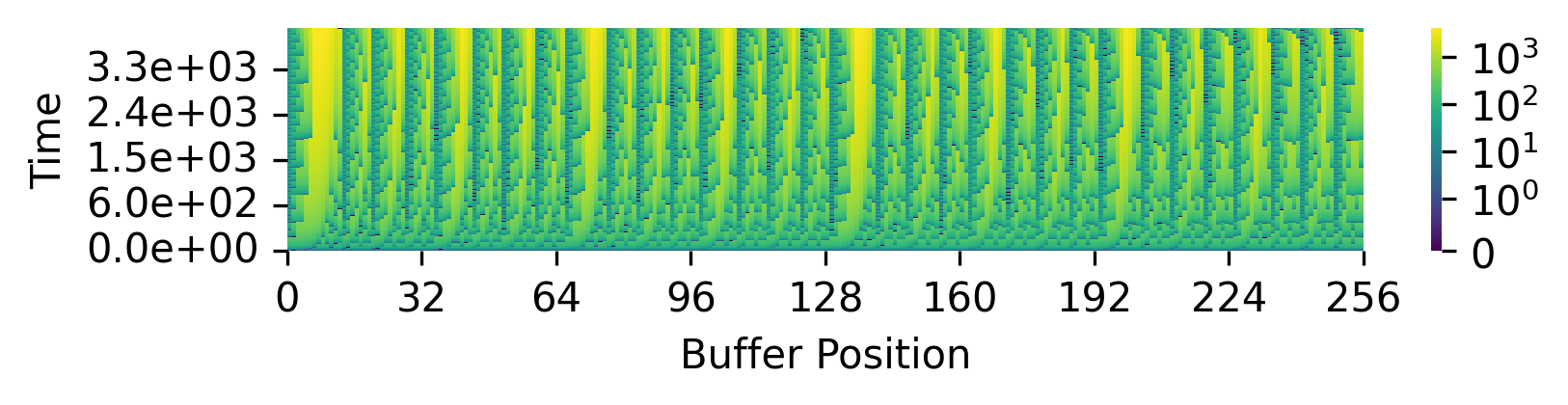}  
  \end{subfigure}%
  \begin{subfigure}[t]{0.35\linewidth}
  \vspace{-2pt}
  \caption{%
    \footnotesize
    Stored data item age across buffer sites for buffer size $\colorS=256$ from $\colorT=0$ to 4,096.
  }
  \label{fig:hsurf-tilted-implementation-heatmap}
\end{subfigure}
\end{minipage}

  \vspace{-0.5ex}
   \begin{minipage}[]{\textwidth}
   \vspace{-2pt}
  \begin{subfigure}[t]{0.65\linewidth}
  \vspace{0pt}
    \centering
    \includegraphics[width=0.88\linewidth,clip]{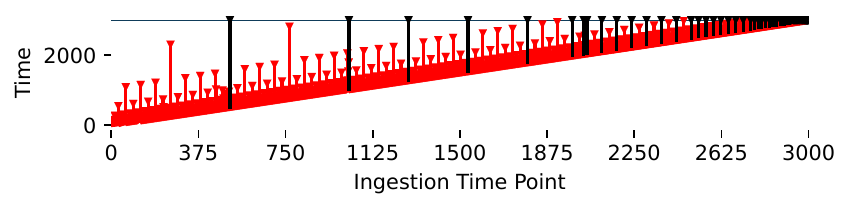}
  \end{subfigure}%
  \begin{subfigure}[t]{0.35\linewidth}
  \vspace{-2pt}
  \caption{%
    \footnotesize
    Data item retention time spans by ingestion time point for buffer size $\colorS=64$ from $\colorT=0$ to 3,000.
  }
  \label{fig:hsurf-tilted-implementation-dripplot}
  \end{subfigure}
  \end{minipage}

  \vspace{-0.5ex}
 \begin{minipage}[]{\textwidth}
 \vspace{-2pt}
\begin{subfigure}[t]{0.65\linewidth}
\vspace{0pt}
  \centering
  \includegraphics[width=0.88\linewidth,clip]{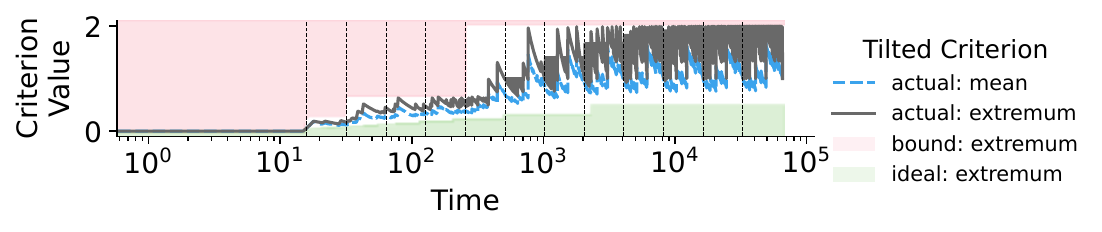}
\end{subfigure}%
\begin{subfigure}[t]{0.35\linewidth}
\vspace{-2pt}
\caption{%
  \footnotesize
  Tilted criterion satisfaction across time points for buffer size $\colorS=16$.
}
\label{fig:hsurf-tilted-implementation-satisfaction}
\end{subfigure}
\end{minipage}

\vspace{-2ex}\caption{%
  \textbf{Tilted algorithm implementation.}
  \footnotesize
  Top panel \ref{fig:hsurf-tilted-implementation-site-selection} enumerates initial tilted policy site selection on a 32-site buffer.
  Panel \ref{fig:hsurf-tilted-implementation-schematic} summarizes how data items are ingested and retained over time within a 32-site buffer, color-coded by data items' hanoi values $\colorH(\colorTbar)$.
  Between $\colorT=0$ and $\colorT=127$, time is segmented into epochs $\colort=0$, $\colort=1$, and $\colort=2$; strips before  each epoch show hanoi values assigned to each buffer site during that epoch.
  Time increases along the $y$ axis.
  Rectangles with small white ``$\blkhorzoval$'' symbol denote buffer site where the ingested data item from each timestep $\colorT$ is placed.
  At each epoch, data items are filled into sites newly assigned for their ingestion-order hanoi value from left to right.
  In epoch 0, all sites are filled with a first data item.
  At subsequent epochs, the first site of all innermost-nested segments is ``invaded'' by new high \hv{} sites added to other segments.
  Low \hv{} data items for which a newly-allocated reservation site is not available ``cycle'' within sites reserved for that \hv{}, ensuring the most recent data items corresponding to that hanoi value are retained.
  The invasion process continues over successive epochs until only one segment remains, as shown in Figure \ref{fig:hsurf-stretched-intuition-reservations}.
  Heatmap panel \ref{fig:hsurf-tilted-implementation-heatmap} shows evolution of data item age at buffer sites.
  Dripplot panel \ref{fig:hsurf-tilted-implementation-dripplot} shows retention spans for 3,000 ingested time points.
  Vertical lines span durations between ingestion and elimination for data items from successive time points.
  Time points previously eliminated are marked in red.
  Lineplot panel \ref{fig:hsurf-tilted-implementation-satisfaction} shows tilted criterion satisfaction on a 16-bit surface over $2^{16}$ timepoints.
  Lower and upper shaded areas are best- and worst-case bounds, respectively.
}
\label{fig:hsurf-tilted-implementation}

\end{figure*}

Site selection for data ingest proceeds similarly to the stretched algorithm, described in Section \ref{sec:stretched}.
However, instead of discarding data items after available sites reserved to that \hv{} fill, we simply cycle back and overwrite the first data items within that \hv{}'s reservations.
In practice, the target index among available sites reserved to a \hv{} can be calculated as the number of previous times a \hv{} has been encountered before time $\colorT$, modulus the number of sites reserved to that hanoi value.
Figure \ref{fig:hsurf-tilted-implementation-schematic} illustrates site selection over epochs $\colort \in \{0,1,2\}$ on buffer size $\colorS=32$.
Algorithm \ref{alg:tilted-site-selection} provides a step-by-step listing of the tilted site selection procedure, which is $\mathcal{O}(1)$.

\begin{algorithm}[H]
\caption{Tilted algorithm site selection $\colorK(\colorT)$.\\ \footnotesize Supplementary Algorithm \ref{alg:tilted-time-lookup} gives tilted algorithm site lookup $\colorL(\colorT)$. Supplementary Listings \cref{lst:tilted_site_selection.py,lst:tilted_time_lookup.py} provide reference Python code.}
\label{alg:tilted-site-selection}
\begin{minipage}{0.53\textwidth}
    \hspace*{\algorithmicindent} \textbf{Input:} $\colorS \in \{2^{\mathbb{N}}\},\;\; \colorT \in \mathbb{N}$ \Comment{Buffer size and current logical time}\\
    \hspace*{\algorithmicindent} \textbf{Output:} $\colork \in [0 \twodots \colorS - 1) \cup \{\nullval\}$ \Comment{Selected site, if any}
    \begin{algorithmic}[1]
        \State $\texttt{uint\_t} ~ ~ \colors \gets \Call{BitLength}{\colorS} - 1$
        \State $\texttt{uint\_t} ~ ~ \colort \gets \max(0,\;\; \Call{BitLength}{\colorT} - \colors)$ \Comment{Current epoch}
        \State $\texttt{uint\_t} ~ ~ \colorh \gets \Call{CountTrailingZeros}{\colorT + 1}$ \Comment{Current \hv{}}
        \Statex
        \State $\texttt{uint\_t} ~ ~ i \gets \Call{RightShift}{\colorT, \;\; \colorh + 1}$ \Comment{Hanoi value incidence (i.e., num seen)}
        \State $\texttt{bool\_t} ~ ~ \epsilon_{\colortau} \gets \Call{BitFloorSafe}{2\colort} \;\; > \;\; \colort + \Call{BitLength}{\colort}$ \Comment{Correction factor}
        \State $\texttt{uint\_t} ~ ~ \colortau \gets  \Call{BitLength}{\colort} - \Call{I}{\epsilon_{\colortau}}$ \Comment{Current meta-epoch}
        \State $\texttt{uint\_t} ~ ~ \colort_0 \gets 2^{\colortau} - \colortau$ \Comment{First epoch of meta-epoch}
        \State $\texttt{uint\_t} ~ ~ \colort_1 \gets 2^{\colortau + 1} - (\colortau + 1)$ \Comment{First epoch of next meta-epoch}
        \State $\texttt{uint\_t} ~ ~ \epsilon_B \gets \Call{I}{\colort \;\; < \;\; \colorh + \colort_0 \;\; < \;\; \colort_1}$ \Comment{Uninvaded correction factor}
        \State $\texttt{uint\_t} ~ ~ B \gets \max(1,\;\; \Call{RightShift}{\colorS, \;\; \colortau + 1 - \epsilon_B})$ \Comment{Num bunches available to \hv}
        \Statex
        \State $\texttt{uint\_t} ~ ~ b_l \gets \Call{ModPow2}{i, \;\; B}$ \Comment{Logical bunch index, in order filled \ldots}
        \Statex \Comment{\ldots i.e., increasing nestedness/decreasing init size $r$}
        \Statex
        \Statex \Comment{Need to calculate physical bunch index\ldots}
        \Statex \Comment{\ldots i.e., among bunches left-to-right in buffer space}
        \Statex
        \State $\texttt{uint\_t} ~ ~ v \gets \Call{BitLength}{b_l}$ \Comment{Nestedness depth level for physical bunch}
        \State $\texttt{uint\_t} ~ ~ w \gets \Call{RightShift}{\colorS, \;\; v} \;\; \times \;\;\Call{I}{v > 0}$ \Comment{Num bunches spaced between bunches in same nest level}
        \State $\texttt{uint\_t} ~ ~ o \gets 2w$  \Comment{Offset of nestedness level in physical bunch order}
        \State $\texttt{uint\_t} ~ ~ p \gets b_l - \Call{BitFloorSafe}{b_l}$ \Comment{Bunch position within nestedness level}
        \State $\texttt{uint\_t} ~ ~ b_p \gets o + wp$ \Comment{Physical bunch index\ldots}
        \Statex \Comment{\ldots i.e., in left-to-right buffer space ordering}
        \Statex
        \Statex \Comment{Need to calculate buffer position of $b_p$\textsuperscript{th} bunch}
        \Statex
        \State $\texttt{uint\_t} ~ ~ \epsilon_{\colork_b} \gets \Call{I}{b_l > 0}$  \Comment{Correction factor, 0\textsuperscript{th} bunch (i.e., bunch $r=\colors$ at site $\colork=0$)}
        \State $\texttt{uint\_t} ~ ~ \colork_b \gets \Call{BitCount}{2b_p +(2\colorS - b_p)} - 1 - \epsilon_{\colork_b}$  \Comment{Site index of bunch}
        \Statex
        \State \Return $\colork_b + \colorh$ \Comment{Calculate placement site, \hv{} $\colorh$ is offset within bunch}
    \end{algorithmic}
\end{minipage}
\end{algorithm}

The data item $\colorTbar$ present at buffer site $\colork$ at time $\colorT$ can be determined by decoding that site's segment index and checking whether (if slated) it has yet been replaced during the current epoch $\colort$.
Both site selection $\colorK$ and lookup $\colorL$ can be accomplished through fast $\mathcal{O}(1)$ binary operations (e.g., bit mask, bit shift, count leading zeros, popcount).
Tilted site lookup is provided in supplementary material, as Algorithm \ref{alg:tilted-time-lookup}.
Reference Python implementations appear in Supplementary Listings \ref{lst:tilted_site_selection.py} and \ref{lst:tilted_time_lookup.py}, as well as accompanying tests.

\subsection{Tilted Algorithm Criterion Satisfaction}
\label{sec:tilted-satisfaction}

In this final subsection, we establish an upper bound on $\mathsf{cost\_tilted}(\colorT)$ for a buffer of size $\colorS$ at time $\colorT$ under the proposed tilted curation algorithm.

\begin{theorem}[Tilted algorithm gap size ratio upper bound]
\label{thm:tilted-gap-size}
Under the tilted curation algorithm, gap size ratio is bounded per Equation \ref{eqn:tilted-gap-size-bound}.
\end{theorem}
\begin{proof}

From Supplementary Lemma \ref{thm:gap-size-ratio-tilted}, we have that if the first $n$ instances of each \hv{} $\colorh$ are retained, $\mathsf{cost\_tilted}(\colorT)$ is bounded below by $1/(n - 1/2)$.
Substituting expressions for the number of sites reserved per \hv{} $n$ from Supplementary Lemma \ref{thm:stretched-discarded-incidence-count} and Supplementary Corollary \ref{thm:stretched-reservation-count} gives the result.
\end{proof}

During early epoch $\colort = 1$, $\mathsf{cost\_tilted}(\colorT) \leq 4/\colorS$.
Likewise, during the last supported meta-epoch $\colortau = \colors - 1$, $\mathsf{cost\_tilted}(\colorT) \leq 2$.
Figure \ref{fig:hsurf-tilted-implementation-satisfaction} shows algorithm performance on the tilted criterion for buffer size $\colorS=16$, $\colorT \in [0\twodots 2^{\colorS} - 1)$.

\section{Conclusions and Further Directions} \label{sec:conclusion}

In closing, we will briefly review the principal objectives, major results, and impact of our presented work.
We finish by laying out future work --- in yet-incomplete aspects of the presented work, as well as opportunities for extension and elaboration.
We also outline steps to build out broad availability of developed algorithms as an off-the-shelf, plug-and-play software tool.

\subsection{Summary and Discussion}

In this work, we have introduced new ``DStream'' algorithms for fast and space-efficient data stream curation --- subsampling from a rolling sequence of data items to dynamically maintain a representative cross-sample across observed time points.
Our approach, in particular, targets use cases that are fixed-capacity and resource-constrained.

As a simplifying assumption, we have reduced data ingestion to a sole update operation: ``site selection,'' picking a buffer index for the $n$th received data item --- overwriting any existing data item at that location.
In the interest of concision and efficiency, we forgo any explicit metadata storage or data structure overhead (e.g., pointers).
Instead, we require site selection for the $n$th ingested item to be computable \textit{a priori}.
Interpreting stored data, therefore, additionally requires support for ``inverse'' decoding of ingest time based solely on an item's buffer index $\colork$ and current time $\colorT$.

Ultimately, the purpose of stream curation is to dictate what data to keep, and for how long.
As objectives in this regard differ by use case, we have explored a suite of three possible retention strategies.

The first is \textit{steady} curation, which calls for retention of evenly-spaced samples across data stream history.
Our proposed algorithm guarantees worst-case even coverage within a constant factor of the optimum.

The next two curation objectives explored (\textit{stretched} and \textit{tilted} criteria) bias retention to favor earlier or more recent data items, respectively.
Proposed algorithms for these two criteria relate closely in structure, differing only in that the former freezes the first encountered data items in place, while the latter uses a ring buffer approach to maintain the most recently encountered data items.
Unlike the proposed steady curation algorithm, which handles indefinitely many data item ingestions, we leave behavior for time $\colorT \geq 2^{\colorS} - 2$ unspecified in defining the proposed stretched and tilted algorithms.
As noted earlier, we expect support for $2^{\colorS} - 2$ ingests to suffice for most use cases.

For all three DStream algorithms, we explain buffer layout procedure and show how site selection proceeds on this basis.
As implemented, all algorithms provide $\mathcal{O}(1)$ site selection operations and are $\mathcal{O}(\colorS)$ to decode ingest times at all $\colorS$ buffer sites.
Each algorithm also provides strict worst-case upper bounds on curation quality across elapsed stream history.

\subsection{Future Algorithm Development}

As mentioned above, the core limitation of this work is the restriction of stretched and tilted algorithms to $2^{\colorS} - 2$ data item ingests.
As such, work remains to design behavior past this point.
One possibility is switching over at $\colorT = 2^{\colorS} - 1$ to apply steady curation on logical time hanoi value $\colorH(\colorT)$ (i.e., rather than on logical time $\colorT$ itself, as originally formulated).

Another enhancement would be random-access lookup calculation.
Current implementations assume an $\mathcal{O}(\colorS)$ pass over all all stored data items.

Several interesting openings exist for extension of additional operations on curated data.
Notably, fast retrieval of the retained data item closest to a query $\colorTbar$ would be useful, as would fast ingest-order iteration over buffer sites $\colork \in [0 \twodots \colorS)$.

A final unexplored direction is fast comparison between curated collections --- which is critical for applications that rely on identifying discrepancies between stream histories, such as hereditary stratigraphy.
These use cases would benefit from fast operations to identify the retained data items $\colorTbar$ shared in common between two time points $\colorT_1$ and $\colorT_2$.
The stable buffer position of data items, once stored, raises the possibility of applying vectorized operations for record-to-record comparison (e.g., masked bitwise equality tests).


\subsection{Algorithm Implementation}

Our foremost motivation for this work is application-driven: We hope to see DStream algorithms put into production to help address real-world challenges in resource-constrained data management.

Indeed, a key driver of this work has been development of ``hereditary stratigraphy'' tooling to support distributed lineage tracking in large-scale digital evolution experiments.
In this use case, stream curation downsamples randomly-generated lineage ``checkpoints'' that accrue as generations elapse, allowing divergence between lineages to be identified via mismatching checkpoints \citep{moreno2022hereditary}.
Prototype implementations of presented algorithms have already seen successful deployment in lineage tracking over massively distributed, agent-based evolution experiments conducted on the 850,000 core Cerebras Wafer-Scale Engine (WSE) device \citep{moreno2024trackable}.
Promisingly, empirical microbenchmark experiments reported in that work corroborate order-of-magnitude efficiency gains from the algorithms presented here, compared to existing approaches used for hereditary stratigraphy.

However, we also anticipate broader use cases beyond hereditary stratigraphy.
This possibility warrants standalone software implementations of algorithms proposed herein, independent of infrastructure developed to support hereditary stratigraphy \citep{moreno2022hstrat}.
As described in Section \ref{sec:materials}, we have organized stream-curation-specific components --- including all three algorithms presented here --- as the standalone software library \citep{moreno2024downstream}.
Going forward, we intend for stream curation algorithms to support lineage tracking implementation as a public-facing external dependency rather than as an opaque internal utility.

One challenge in supporting end-users is cross-language interoperation.
Partial implementations are currently available in Python, Zig, and the closely related Cerebras Software Language (CSL) \citep{moreno2024hsurf,moreno2024downstream,moreno2024wse}.
For our own purposes, we plan to establish ports of stream curation algorithms for Rust and C++.

We would be highly interested in collaborations in assembling DStream implementations in other languages as needed --- whether incorporating new implementations into the \textit{downstream} software repository or linking to outside repositories from the \textit{downstream} documentation.
In either case, care will be needed for consistency across implementations, as the semantics of stored data depend subtly upon exactly how site selection unfolded.
One possible approach to this issue would be to simply designate a canonical implementation and provide language-agnostic tests to validate other implementations against it.
Alternatively, effort could be invested in preparing and maintaining an explicit standard or specification.

\section*{Acknowledgment}

This material is based upon work supported by the Eric and Wendy Schmidt AI in Science Postdoctoral Fellowship, a Schmidt Sciences program.
Thank you to Ryan Moreno and Connor Yang for providing valuable feedback on manuscript drafts.
This work also benefited from the thoughtful suggestions of several anonymous reviewers.

\putbib

\end{bibunit}

\clearpage
\newpage

\begin{bibunit}

\appendix
\section*{\Huge Supplemental Material}

\setcounter{section}{0}

\makeatletter
\def\@seccntformat#1{\@ifundefined{#1@cntformat}%
   {\csname the#1\endcsname\space}
   {\csname #1@cntformat\endcsname}}
\newcommand\section@cntformat{\thesection.\space} 
\makeatother
\renewcommand{\thesection}{S\arabic{section}}
\counterwithin{equation}{section}
\counterwithin{figure}{section}
\counterwithin{table}{section}
\counterwithin{theorem}{section}
\counterwithin{algorithm}{section}
\counterwithin{lstlisting}{section}

\section{Pseudocode Helper Functions}
\label{sec:pseudocode}

\begin{minipage}{0.5\textwidth}
\begin{algorithmic}
  \Function{BitCount}{$x$}
    \LComment{Equivalent {\normalfont $\texttt{std::popcount}(x)$}}
    \State \Return $|\{ n \in \mathbb{N} : x \bmod 2^n = 0 \}|$
  \EndFunction
  \Statex
  \Function{BitFloor}{$x$}
    \State \Return $\Call{LeftShift}{1, \;\; \Call{BitLength}{x} - 1}$
  \EndFunction
  \Statex
  \Function{BitFloorSafe}{$x$}
    \LComment{Equivalent {\normalfont $\texttt{std::bit\_floor}(x)$}}
    \State \Return \Call{BitFloor}{$x$} \textbf{if} $x > 0$ \textbf{else} 0
  \EndFunction
  \Statex
  \Function{BitLength}{$x$}
    \LComment{Equivalent {\normalfont $\texttt{std::bit\_width}(x)$}}
    \State \Return $\left\lfloor\log_{2}(x)\right\rfloor + 1$  \textbf{if} $x$ \textbf{else} $0$
  \EndFunction
  \Statex
  \Function{I}{$x$} \Comment{``Indicator'' function $I$}
    \LComment{Equivalent {\normalfont $\texttt{static\_cast<unsigned int>}(x)$}}
    \State \Return $1$ \textbf{if} $x > 0$ \textbf{else} $0$
  \EndFunction
  \Statex
  \Function{CountTrailingZeros}{$x$}
    \LComment{Equivalent {{\normalfont $\texttt{std::countr\_zero}(x)$}}}
    \State \Return $\max \{ n \in \mathbb{N} : x \bmod 2^n = 0 \}$
  \EndFunction
  \Statex
  \Function{ElvisOp}{$x, \;\; y$}
    \LComment{Equivalent {\normalfont $x \;\; \texttt{?:} \;\; y$}}
    \State \Return $x$ \textbf{if} $x \neq 0$ \textbf{else} $y$
  \EndFunction
  \Statex
  \Function{LeftShift}{$x, \;\; n$}
    \LComment{Equivalent {\normalfont $x \;\; \texttt{<{}<} \;\; y$}}
    \State \Return $2^n x$
  \EndFunction
  \Statex
  \Function{ModPow2}{$x, \;\; n$}
    \LComment{Equivalent {\normalfont $x \;\; \texttt{\&} \;\; (n - 1)$}, requiring $n \in \{2^{\mathbb{N}}\}$}
    \State \Return $x \bmod n$
  \EndFunction
  \Statex
  \Function{RightShift}{$x, \;\;n$}
    \LComment{Equivalent {\normalfont $x \;\; \texttt{>{}>} \;\; y$}}
    \State \Return $\left\lfloor x/2^n \right\rfloor$
  \EndFunction
  \Statex
\end{algorithmic}
\end{minipage}

\section{Site Lookup Algorithms}

\begin{algorithm}[H]
\caption{Steady algorithm ingest time lookup $\colorL(\colorT)$.\\
\footnotesize Supplementary Listing \ref{lst:steady_time_lookup.py} provides reference Python code.}
\label{alg:steady-time-lookup}
\begin{minipage}{0.5\textwidth}
    \hspace*{\algorithmicindent} \textbf{Input:} $\colorS \in \{2^{\mathbb{N}}\},\;\; \colorT \in \mathbb{N}$ \Comment{Buffer size and current logical time}\\
    \hspace*{\algorithmicindent} \textbf{Output:} $\colorTbar \in [0 \twodots \colorT) \cup \{\nullval\}$ \Comment{Ingestion time of stored data item, if any}
    \begin{algorithmic}[1]
    \If{$\colorT < \colorS$} \Comment{If buffer not yet filled\ldots}
        \ForAll{$v \in \Call{$\colorL'$}{\colorS,\;\; \colorS}}$
            \If{$v < \colorT$} \Comment{\ldots filter out not-yet-encountered values}
                \State \textbf{yield} $v$
            \Else
                \State \textbf{yield} $\nullval$
            \EndIf
        \EndFor
    \Else \Comment{No filter needed once buffer is filled}
        \ForAll{$v \in \Call{$\colorL'$}{\colorS,\;\; \colorT}$}
            \State \textbf{yield} $v$
        \EndFor
    \EndIf
    \Statex
    \Function{$\colorL'$}{$\colorS, \;\; \colorT$}\\
        \hspace*{\algorithmicindent} \textbf{Input:} $\colorS \in \{2^{\mathbb{N}}\},\;\; \colorT \in [\colorS \twodots)$ \Comment{Buffer size and current logical time}\\
        \hspace*{\algorithmicindent} \textbf{Output:} $\colorTbar \in [0 \twodots \colorT)$ \Comment{Ingestion time of stored data item, if any}
        \State $\texttt{uint\_t} ~ ~ \colors \gets \Call{BitLength}{\colorS} - 1$
        \State $\texttt{uint\_t} ~ ~ \colort \gets \Call{BitLength}{\colorT} - \colors $ \Comment{Current epoch}
        \State $\texttt{uint\_t} ~ ~ b \gets 0$ \Comment{Bunch logical/physical index (ordered left to right)}
        \State $\texttt{uint\_t} ~ ~ m^{\scriptscriptstyle\shortdownarrow}_b \gets 1$ \Comment{Countdown on segments traversed within bunch}
        \State $\texttt{bool\_t} ~ ~ b^{*} \gets \texttt{True}$ \Comment{Flag if have traversed all segments in bunch?}
        \State $\texttt{uint\_t} ~ ~ {\colork}^{\scriptscriptstyle\shortdownarrow}_m \gets \colors + 1$ \Comment{Countdown on sites traversed within segment}
        \ForAll{$\colork \in [0\twodots \colorS)$} \Comment{Iterate over buffer sites}
        \Statex
        \Statex \Comment{Calculate info about current segment\ldots}
        \State $\texttt{uint\_t} ~ ~ \epsilon_{w} \gets \Call{I}{b == 0}$ \Comment{Correction on seg width for first bunch}
        \State $\texttt{uint\_t} ~ ~ w \gets \colors - b + \epsilon_{w}$ \Comment{Number of sites in current segment (i.e., segment size)}
        \State $\texttt{uint\_t} ~ ~ m_p \gets 2^b -m^{\scriptscriptstyle\shortdownarrow}_b$ \Comment{Calc left-to-right index of current segment}
        \State $\texttt{uint\_t} ~ ~ \colorh_{\max} \gets \colort + w - 1$ \Comment{Max possible \hv{} in segment during current epoch $\colort$}
        \State $\texttt{uint\_t} ~ ~ \colorh' \gets \colorh_{\max} - ((\colorh_{\max} + {\colork}^{\scriptscriptstyle\shortdownarrow}_m) \bmod w)$ \Comment{Candidate hanoi value}
        \Statex
        \Statex \Comment{Decode ingest $\colorTbar$ from physical segment index $m_p$\ldots}
        \Statex \Comment{\dots which tells instance of reserved \hv{} (i.e., how many seen)\ldots}
        \State $\texttt{uint\_t} ~ ~ \colorTbar' \gets 2^{\colorh'}(2m_p + 1) - 1$ \Comment{Guess ingest time of data item at current site}
        \State $\texttt{uint\_t} ~ ~ \epsilon_{\colorh} \gets \Call{I}{\colorTbar \geq \colorT} \;\; \times \;\; w$ \Comment{Correction on \hv{} if assigned instance not yet seen (i.e., $\colorTbar \geq \colorT$)}
        \State $\texttt{uint\_t} ~ ~ \colorh \gets \colorh' - \epsilon_{\colorh}$ \Comment{Corrected true resident \hv{} at site}
        \State $\texttt{uint\_t} ~ ~ \colorTbar_{\colork} \gets 2^{\colorh}(2m_p + 1) - 1$ \Comment{True ingest time}
        \State \textbf{yield} $\colorTbar_{\colork}$
        \Statex
        \Statex \Comment{Update state for next site iterated over\ldots}
        \State ${\colork}^{\scriptscriptstyle\shortdownarrow}_m \gets \Call{ElvisOp}{{\colork}^{\scriptscriptstyle\shortdownarrow}_m, \;\; w} - 1$ \Comment{Bump to next site in segment, or reset for new segment}
        \State $m^{\scriptscriptstyle\shortdownarrow}_b \gets m^{\scriptscriptstyle\shortdownarrow}_b - \Call{I}{{\colork}^{\scriptscriptstyle\shortdownarrow}_m  = 0}$ \Comment{Bump to next segment within bunch}
        \State $b^{*} \gets (m^{\scriptscriptstyle\shortdownarrow}_b = 0) \wedge ({\colork}^{\scriptscriptstyle\shortdownarrow}_m = 0)$ \Comment{Should bump to next bunch?}
        \State $b \gets b + \Call{I}{b^{*}}$ \Comment{Do bump to next bunch, if should}
        \State $m^{\scriptscriptstyle\shortdownarrow}_b \gets \Call{ElvisOp}{m^{\scriptscriptstyle\shortdownarrow}_b, \;\; 2^{b - 1}}$ \Comment{Set within-bunch segment countdown if bumping to next bunch}
        \EndFor
    \EndFunction
    \end{algorithmic}
\end{minipage}
\end{algorithm}

\begin{algorithm}[H]
\caption{Stretched algorithm ingest time lookup $\colorL(\colorT)$.\\
\footnotesize Supplementary Listing \ref{lst:stretched_time_lookup.py} provides reference Python code.}
\label{alg:stretched-time-lookup}
\begin{minipage}{0.55\textwidth}
    \hspace*{\algorithmicindent} \textbf{Input:} $\colorS \in \{2^{\mathbb{N}}\},\;\; \colorT \in [0 \twodots 2^{\colorS} - 1)$ \Comment{Buffer size and current logical time}\\
    \hspace*{\algorithmicindent} \textbf{Output:} $\colorTbar \in [0 \twodots \colorT) \cup \{\nullval\}$ \Comment{Ingestion time of stored data item, if any}
    \begin{algorithmic}[1]
    \If{$\colorT < \colorS - 1$} \Comment{If buffer not yet filled\ldots}
        \ForAll{$v \in \Call{$\colorL'$}{\colorS, \;\; \colorS}}$
            \If{$v < \colorT$} \Comment{\ldots filter out not-yet-encountered values}
                \State \textbf{yield} $v$
            \Else
                \State \textbf{yield} $\nullval$
            \EndIf
        \EndFor
    \Else \Comment{No filter needed once buffer is filled}
        \ForAll{$v \in \Call{$\colorL'$}{\colorS, \;\; \colorT}$}
            \State \textbf{yield} $v$
        \EndFor
    \EndIf
    \Statex
    \Function{$\colorL'$}{$\colorS, \;\; \colorT$}\\
        \hspace*{\algorithmicindent} \textbf{Input:} $\colorS \in \{2^{\mathbb{N}}\},\;\; \colorT \in [\colorS - 1 \twodots 2^{\colorS} - 1)$ \Comment{Buffer size and current logical time}\\
        \hspace*{\algorithmicindent} \textbf{Output:} $\colorTbar \in [0 \twodots \colorT)$ \Comment{Ingestion time of stored data item, if any}
        \State $\texttt{uint\_t} ~ ~ \colors \gets \Call{BitLength}{\colorS} - 1$
        \State $\texttt{uint\_t} ~ ~ \colort \gets \Call{BitLength}{\colorT} - \colors $ \Comment{Current epoch}
        \State $\texttt{bool\_t} ~ ~ \epsilon_{\colortau} \gets \Call{BitFloorSafe}{2\colort} \;\; > \;\; \colort + \Call{BitLength}{\colort}$ \Comment{Correction factor for calculating meta-epoch $\colortau$}
        \State $\texttt{uint\_t} ~ ~ \colortau_{\mkern-5mu\scriptscriptstyle 0} \gets \Call{BitLength}{\colort} - \Call{I}{\epsilon_{\colortau}}$ \Comment{Current meta-epoch}
\State $\texttt{uint\_t} ~ ~ \colortau_{\mkern-5mu\scriptscriptstyle 1} \gets \colortau_{\mkern-5mu\scriptscriptstyle 0} + 1$ \Comment{Next meta-epoch}
        \State $\texttt{uint\_t} ~ ~ M \gets \Call{ElvisOp}{\Call{RightShift}{\colorS, \;\; \colortau_{\mkern-5mu\scriptscriptstyle 1}}, \;\; 1}$ \Comment{Number of invading segments present at current epoch}
        \State $\texttt{uint\_t} ~ ~ w_0 \gets 2^{\colortau_{\mkern-5mu\scriptscriptstyle 0}} - 1$ \Comment{Smallest segment size at outset of meta-epoch $\colortau_{\mkern-5mu\scriptscriptstyle 0}$}
        \State $\texttt{uint\_t} ~ ~ w_1 \gets 2^{\colortau_{\mkern-5mu\scriptscriptstyle 1}} - 1$ \Comment{Smallest segment size at outset of meta-epoch $\colortau_{\mkern-5mu\scriptscriptstyle 1}$}
        \Statex
        \State $\texttt{uint\_t} ~ ~ \colorh' \gets 0$ \Comment{Reserved hanoi value at site $\colork=0$}
        \State $\texttt{uint\_t} ~ ~ m_p \gets 0$ \Comment{Physical segment index at site $\colork=0$ (i.e., left-to-right position)}
        \ForAll{$\colork \in [0\twodots \colorS)$} \Comment{Iterate overall buffer sites}
        \State $\texttt{uint\_t} ~ ~ b_l \gets \Call{CountTrailingZeros}{M + m_p}$ \Comment{Logical bunch index in reverse fill order\ldots}
        \State \Comment{\dots (i.e., decreasing nestedness/increasing initial size $r$)}
        \State $\texttt{uint\_t} ~ ~ \epsilon_w \gets \Call{I}{m_p = 0}$ \Comment{Correction factor for segment size}
        \State $\texttt{uint\_t} ~ ~ w \gets w_1 + b_l + \epsilon_w$ \Comment{Number of sites in current segment (i.e., segment size)}
        \Statex \Comment{Calc corrections for not-yet-seen data items $\colorTbar \geq \colorT$}
        \Statex $\texttt{uint\_t} ~ ~ i' \gets \Call{RightShift}{M + m_p, \;\; b_l + 1}$ \Comment{Guess \hv{} incidence (i.e., num seen)}
        \State $\texttt{uint\_t} ~ ~ \colorTbar_{\colork}' \gets 2^{\colorh}(2i' + 1) - 1$ \Comment{Guess ingest time}
        \State $\texttt{uint\_t} ~ ~ \epsilon_{\colorh} \gets \Call{I}{\colorTbar_{\colork}' \geq \colorT} \;\; \times \;\; (w - w_0)$ \Comment{Correction factor, reserved \hv{} $\colorh$}
        \State $\texttt{uint\_t} ~ ~ \epsilon_{i} \gets \Call{I}{\colorTbar_{\colork}' \geq \colorT} \;\; \times \;\; (m_p + M - i')$ \Comment{Correction factor, \hv{} instance $i$ (i.e., num seen)}
        \Statex
        \Statex \Comment{Decode ingest time of reserved hanoi value \ldots}
        \State $\texttt{uint\_t} ~ ~ \colorh \gets \colorh' - \epsilon_{\colorh}$ \Comment{True reserved \hv{}}
        \State $\texttt{uint\_t} ~ ~ i \gets i' + \epsilon_{i}$ \Comment{True \hv{} incidence}
        \State $\texttt{uint\_t} ~ ~ \colorTbar_{\colork} \gets 2^{\colorh}(2i + 1) - 1$ \Comment{True ingest time}
        \State \textbf{yield} $\colorTbar_{\colork}$
        \Statex
        \Statex \Comment{Update state for next site $\colork+1$ to iterate over\ldots}
        \State $\colorh' \gets \colorh' + 1$ \Comment{Increment next site's \hv{} guess}
        \State $m_p \gets m_p + \Call{I}{\colorh' = w}$ \Comment{Bump to next segment if current is filled}
        \State $\colorh' \gets \colorh' \;\; \times \;\; \Call{I}{\colorh' \neq w}$ \Comment{Reset \hv{} to zero if segment is filled to start new segment}
        \EndFor
    \EndFunction
    \end{algorithmic}
\end{minipage}
\end{algorithm}

\begin{algorithm}[H]
\caption{Tilted algorithm ingest time lookup $\colorL(\colorT)$.\\
\footnotesize Supplementary Listing \ref{lst:tilted_time_lookup.py} provides reference Python code.}
\label{alg:tilted-time-lookup}
\begin{minipage}{0.55\textwidth}
    \hspace*{\algorithmicindent} \textbf{Input:} $\colorS \in \{2^{\mathbb{N}}\},\;\; \colorT \in [0 \twodots 2^{\colorS} - 1)$ \Comment{Buffer size and current logical time}\\
    \hspace*{\algorithmicindent} \textbf{Output:} $\colorTbar \in [0 \twodots \colorT) \cup \{\nullval\}$ \Comment{Ingestion time of stored data item, if any}
    \begin{algorithmic}[1]
    \If{$\colorT < \colorS - 1$} \Comment{If buffer not yet filled\ldots}
        \ForAll{$v \in \Call{$\colorL'$}{\colorS, \;\; \colorS}}$
            \If{$v < \colorT$} \textbf{yield} $v$ \textbf{else} \textbf{yield} $\nullval$ \Comment{\ldots filter out not-yet-encountered values}
            \EndIf
        \EndFor
    \Else ~ \textbf{yield from} $\Call{$\colorL'$}{\colorS, \;\;\colorT}$\Comment{No filter needed once buffer is filled}
    \EndIf
    \Statex
    \Function{$\colorL'$}{$\colorS, \;\; \colorT$} \Comment{Assume buffer has been filled} \\
        \hspace*{\algorithmicindent} \textbf{Input:} $\colorS \in \{2^{\mathbb{N}}\},\;\; \colorT \in [\colorS - 1 \twodots 2^{\colorS} - 1)$ \Comment{Buffer size and current logical time}\\
        \hspace*{\algorithmicindent} \textbf{Output:} $\colorTbar \in [0 \twodots \colorT)$ \Comment{Ingestion time of stored data item, if any}
        \State $\texttt{uint\_t} ~ ~ \colors \gets \Call{BitLength}{\colorS} - 1$
        \State $\texttt{uint\_t} ~ ~ \colort \gets \Call{BitLength}{\colorT} - \colors $ \Comment{Current epoch}
        \State $\texttt{bool\_t} ~ ~ \epsilon_{\colortau} \gets \Call{BitFloorSafe}{2\colort} \;\; > \;\; \colort + \Call{BitLength}{\colort}$ \Comment{Correction factor for calculating meta-epoch $\colortau$}
        \State $\texttt{uint\_t} ~ ~ \colortau_{\mkern-5mu\scriptscriptstyle 0} \gets \Call{BitLength}{\colort} - \Call{I}{\epsilon_{\colortau}}$ \Comment{Current meta-epoch}
        \State $\texttt{uint\_t} ~ ~ \colortau_{\mkern-5mu\scriptscriptstyle 1} \gets \colortau_{\mkern-5mu\scriptscriptstyle 0} + 1$ \Comment{Next meta-epoch}
        \State $\texttt{uint\_t} ~ ~ \colort_0 \gets 2^{\colortau_{\mkern-5mu\scriptscriptstyle 0}} - \colortau_{\mkern-5mu\scriptscriptstyle 0}$ \Comment{Opening epoch of current meta-epoch}
        \State $\texttt{uint\_t} ~ ~ \colorT_0 \gets 2^{\colort + \colors - 1}$ \Comment{Opening time of current epoch}
        \State $\texttt{uint\_t} ~ ~ M' \gets \Call{ElvisOp}{\Call{RightShift}{\colorS, \;\; \colortau_{\mkern-5mu\scriptscriptstyle 1}}, \;\; 1}$ \Comment{Number of invading segments present at current epoch}
        \State $\texttt{uint\_t} ~ ~ w_0 \gets 2^{\colortau_{\mkern-5mu\scriptscriptstyle 0}} - 1$ \Comment{Smallest segment size at outset of meta-epoch $\colortau_{\mkern-5mu\scriptscriptstyle 0}$}
        \State $\texttt{uint\_t} ~ ~ w_1 \gets 2^{\colortau_{\mkern-5mu\scriptscriptstyle 1}} - 1$ \Comment{Smallest segment size at outset of meta-epoch $\colortau_{\mkern-5mu\scriptscriptstyle 1}$}
        \State $\texttt{uint\_t} ~ ~ \colorh' \gets 0$ \Comment{Reserved \hv{} for site $\colork=0$}
        \State $\texttt{uint\_t} ~ ~ m_p \gets 0$ \Comment{Physical segment index at site $\colork=0$ (i.e., left-to-right position)}
        \ForAll{$\colork \in [0\twodots \colorS)$} \Comment{Iterate over all buffer sites}
        \State $\texttt{uint\_t} ~ ~ b_l \gets \Call{CountTrailingZeros}{M + m_p}$ \Comment{Logical bunch index in reverse fill order\ldots}
        \Statex \Comment{\dots (i.e., decreasing nestedness/increasing initial size $r$)}
        \State $\texttt{uint\_t} ~ ~ \epsilon_w \gets \Call{I}{m_p = 0}$ \Comment{Correction factor for segment size $w$}
        \State $\texttt{uint\_t} ~ ~ w \gets w_1 + b_l + \epsilon_w$ \Comment{Number of sites in current segment (i.e., segment size)}
        \State $\texttt{uint\_t} ~ ~ m_l' \gets \Call{RightShift}{M + m_p, \;\; b_l + 1}$ \Comment{Guess logical (fill order) segment index}
        \Statex \Comment{\textbf{Scenario A}: site in invaded segment, \hv{} ring buffer intact}
        \State $\texttt{bool\_t} ~ ~  \rchi_{\mkern-5mu\mathrm{A}} \gets \colorh' - (\colort - \colort_0) > w - w_0$ \Comment{Will \hv{} ring buffer be invaded in future epoch $(\colort' > \colort) \in \lBrace \colortau_{\mkern-5mu\scriptscriptstyle 0} \rBrace$?}
        \State $\texttt{uint\_t} ~ ~ \colorT_i \gets 2^{\colorh'}(2m_l' + 1) - 1$ \Comment{When will current site $\colork$ overwritten by invader?}
        \State $\texttt{bool\_t} ~ ~ \rchi_{\mkern-5mu\mathrm{\widehat{A}}} \gets (\colorh' - (\colort - \colort_0) = w - w_0) \wedge (\colorT_i \geq \colorT)$ \Comment{Current site $\colork$ to be overwritten by invader later in $\colort$?}
        \Statex \Comment{\textbf{Scenario B}: site in invading segment, \hv{} ring buffer intact}
        \State $\texttt{bool\_t} ~ ~ \rchi_{\mkern-5mu\mathrm{B}} \gets (\colort - \colort_0 < \colorh' < w_0) \wedge (\colort < \colorS - \colors)$ \Comment{Will \hv{} ring buffer be invaded in future epoch $(\colort' > \colort) \in \lBrace \colortau_{\mkern-5mu\scriptscriptstyle 0} \rBrace$?}
        \State $\texttt{uint\_t} ~ ~ \colorT_r \gets \colorT_0 + \colorT_i$ \Comment{When will current site $\colork$ be refilled after ring buffer halves?}
        \State $\texttt{bool\_t} ~ ~ \rchi_{\mkern-5mu\mathrm{\widehat{B}}} \gets (\colorh' = \colort - \colort_0) \wedge (\colorT_r \geq \colorT) \wedge (\colort < \colorS - \colors) $ \Comment{Current site $\colork$ to be refilled after invasion later in $\colort$?}
        \Statex \texttt{/*} \textit{Note:} $\Call{I}{\rchi_{\mkern-5mu\mathrm{A}}} + \Call{I}{\rchi_{\mkern-5mu\mathrm{\widehat{A}}}} + \Call{I}{\rchi_{\mkern-5mu\mathrm{B}}} + \Call{I}{\rchi_{\mkern-5mu\mathrm{\widehat{B}}}} \in \{0,1\}$ \texttt{*/} \Comment{Apply corrections for complicating scenarios\ldots}
        \State $\texttt{uint\_t} ~ ~ \epsilon_M \gets \Call{I}{\rchi_{\mkern-5mu\mathrm{A}} \vee \rchi_{\mkern-5mu\mathrm{\widehat{A}}} \vee \rchi_{\mkern-5mu\mathrm{B}} \vee \rchi_{\mkern-5mu\mathrm{\widehat{B}}}} \;\; \times \;\; M'$ \Comment{Correction on guessed segment count $M'$}
        \State $\texttt{uint\_t} ~ ~ \epsilon_{\colorh} \gets \Call{I}{\rchi_{\mkern-5mu\mathrm{A}} \vee \rchi_{\mkern-5mu\mathrm{\widehat{A}}}} \;\; \times \;\; (w - w_0)$ \Comment{Correction on guessed reserved \hv{} $\colorh'$}
        \State $\texttt{uint\_t} ~ ~ \epsilon_{\colorT} \gets \Call{I}{\rchi_{\mkern-5mu\mathrm{\widehat{A}}} \vee \rchi_{\mkern-5mu\mathrm{\widehat{B}}}} \;\; \times \;\; (\colorT - \colorT_0)$ \Comment{Need to snap back to outset of current epoch $\colort$?}
        \State $\texttt{uint\_t} ~ ~ M \gets M' + \epsilon_M$ \Comment{Corrected number of segments in \hv{} ring buffer $M$}
        \State $\texttt{uint\_t} ~ ~ \colorh \gets \colorh' - \epsilon_{\colorh}$ \Comment{Corrected resident \hv{} $\colorh$}
        \State $\texttt{uint\_t} ~ ~ \colorT_c \gets \colorT - \epsilon_{\colorT}$ \Comment{Corrected lookup time $\colorT$}
        \State $\texttt{uint\_t} ~ ~ m_l \gets \Call{ElvisOp}{\Call{I}{\rchi_{\mkern-5mu\mathrm{A}} \vee \rchi_{\mkern-5mu\mathrm{\widehat{A}}}} \;\; \times \;\; (M' + m_p), \;\; m_l'}$ \Comment{Corrected logical segment index $m_l$}
        \State \Comment{Decode what \hv{} instance fell on site $\colork$\ldots}
        \State $\texttt{uint\_t} ~ ~ j \gets \Call{RightShift}{\colorT_c + 2^{\colorh}, \;\; \colorh + 1} - 1$ \Comment{Most recent instance of \hv, zero-indexed (i.e., num seen less 1)}
        \State $\texttt{uint\_t} ~ ~ i \gets j - \Call{ModPow2}{j - m_l + M, \;\; M}$ \Comment{Hanoi value incidence resident at site $\colork$}
        \State \textbf{yield} $2^{\colorh}(2i + 1) - 1$ \Comment{Decode ingest time $\colorTbar_{\colork}$ of assigned \hv{}}
        \Statex \Comment{Update state for next site $\colork+1$ to iterate over\ldots}
        \State $\colorh' \gets \colorh' + 1$ \Comment{Increment next site's \hv{} guess}
        \State $m_p \gets m_p + \Call{I}{\colorh' = w}$ \Comment{Bump to next segment if current segment is filled}
        \State $\colorh' \gets \colorh' \;\; \times \;\; \Call{I}{\colorh' \neq w}$ \Comment{Reset \hv{} to zero if segment filled to start new segment}
        \EndFor
    \EndFunction
    \end{algorithmic}
\end{minipage}
\end{algorithm}

\section{Meta-epoch Bound} \label{sec:meta-epoch-bound}

\begin{lemma}[Current meta-epoch upper bounds]
\label{thm:meta-epoch-bound}
The current meta-epoch at epoch $\colort$ is bounded,
\begin{align*}
\colortau \leq
\min\Big(
  \log_2(\colort + \colors),\;\;
  \log_2(\colort) + 1
\Big)
\text{ for } \colort \in [1 \twodots \colorS - \colors).
\end{align*}
\end{lemma}
\begin{proof}

By definition,
\begin{align*}
\colort \geq 2^{\colortau} - \colortau \quad \forall \colort \in \colortausetoft.
\end{align*}
Given $\colort$, it is not possible to derive an analytical expression $f(\colort) = \colortau$ such that $\colort \in \colortausetoft$.
However, we can show an expression $f(\colort)$ as an inclusive upper bound on $\colortau$ with $\colort \in \colortausetoft$ by demonstrating,
\begin{align*}
\colort \leq 2^{f(\colort)} - f(\colort).
\end{align*}
The following demonstrates two such expressions $f(\colort)$ --- one that provides a tighter upper bound on $n$ for small $\colort$ and the other as a tighter bound for large $\colort$.
The result comprises Formulas \ref{eqn:meta-epoch-bound-small} and \ref{eqn:meta-epoch-bound-large}, using the $\min$ operator to apply the tighter of these bounds at each epoch $\colort$.

\begin{proofpart}
First, we show
\begin{align}
\colortau \leq \log_2(\colort) + 1 \text{ for } \colort \geq 1.
\label{eqn:meta-epoch-bound-small}
\end{align}
This bound follows from,
\begin{align*}
2^{\log_{2}(\colort) + 1} - \log_2(\colort) - 1
&= 2\colort - \log_2(\colort) - 1 \\
&\stackrel{\checkmark}{\geq} \colort \text{ for } \colort \in \mathbb{N}^{+}.
\end{align*}
\end{proofpart}

\begin{proofpart}
The upper bound
\begin{align}
\colortau \leq \log_2(\colort + \colors - 1) \label{eqn:meta-epoch-bound-large}
\end{align}
can also be established for $\colort \leq \colorS - \colors$.
Consider,
\begin{align*}
2^{\log_{2}(\colort + \colors)} - \log_2(\colort + \colors)
&\stackrel{?}{\geq} \colort\\
\colort + \colors - \log_2(\colort + \colors)
&\stackrel{?}{\geq} \colort\\
\colors - \log_2(\colort + \colors)
&\stackrel{?}{\geq} 0\\
\log_2(2^{\colors}) - \log_2(\colort + \colors)
&\stackrel{?}{\geq} 0\\
\log_2\frac{2^{\colors}}{\colort + \colors}
&\stackrel{?}{\geq} 0\\
\frac{2^{\colors}}{\colort + \colors}
&\stackrel{?}{\geq} 1\\
\frac{\colorS}{\colort + \colors}
&\stackrel{?}{\geq} 1\\
\colorS
&\stackrel{?}{\geq} \colort + \colors\\
\colorS - \colors &\stackrel{?}{\geq} \colort.
\end{align*}
Stretched and tilted algorithms do not define ingestion for $\colort \geq \colorS - \colors$ (i.e., $\colorT \geq 2^{\colorS - 1}$).
Restricting $\colort$,
\begin{align*}
2^{\log_{2}(\colort + \colors)} - \log_2(\colort + \colors)
&\stackrel{\checkmark}{\geq} \colort \text{ for } \colort \in [1 \twodots \colorS - \colors).
\end{align*}
\end{proofpart}
\end{proof}

\section{Steady Algorithm}

\begin{lemma}[Space required to store $\mathsf{goal\_steady}$] \label{thm:steady-hv-geq-epoch}

At any time $\colorT$ in epoch $\colort$, sufficient buffer space exists to store all data items with \hv{} $\colorh > \colort - 1$.
That is, $\left| \{\colorTbar \in [0 \twodots \colorT) : \colorH(\colorTbar) \geq \colort \} \right| \leq \colorS$.
\end{lemma}

\begin{proof}
It is sufficient to consider epochs' last time point, $\max(\colorT \in \colortsetofT) = 2^{\colors + \colort} - 2$, when storage demand is highest.
Recall that \hv{} $\colorh$ is encountered for the first time at time $\colorT = 2^{\colorh} - 1$.
Summing data item counts for \hv{}'s $\colorh \in [\colort \twodots \colors + \colort]$,
\begin{align*}
\left| \{\colorTbar \in [0 \twodots 2^{\colors + \colort} - 1] : \colorH(\colorTbar) \geq \colort \} \right|\\
&= \sum_{\colorh=\colort}^{\colors + \colort} \left\lceil2^{(\colors + \colort) - \colorh - 1}\right\rceil
= 1 + \sum_{i=1}^{\colors} 2^{i - 1} \\
&= 2^{\colors}\\
&\stackrel{\checkmark}{\leq} \colorS.
\end{align*}
\end{proof}

\begin{lemma}[Placements overwrite \hv{} $\colorh = \colort - 1$]
Placing data items $\colorTbar$ within segments at position $\colorH(\colorTbar)$ modulo segment length ensures elimination of \hv{} $\colorh = \colort - 1$ from each segment.
\end{lemma} \label{thm:steady-hv-elimination}
\begin{proof}
Recall that \hv{} $\colorh = \colort + \colors - i - 1$ is placed in the $i$th bunch during epoch $\colort$ for $i>0$.
By construction, segments in the $i$th bunch have $\colors - i$ sites for $i>0$.
We must verify,
\begin{align*}
\colort - 1
&\stackrel{?}{=}
\mathsf{invading\_h.v.} - \mathsf{segment\_length}\\
&\stackrel{?}{=}
(\colort + \colors - i - 1) - (\colors - i)
 \\
&\stackrel{\checkmark}{=} \colort - 1.
\end{align*}
An identical result can be shown for the bunch $i=0$ segment, which has $\colors+1$ sites.
\end{proof}


\section{Stretched Algorithm}

\begin{lemma}[Best-possible stretched criterion satisfaction]
\label{thm:stretched-ideal-strict}
The stretched criterion (i.e., largest gap size ratio) for a buffer of size $\colorS$ at time $\colorT$ can be minimized no lower than,
\begin{align*}
\mathsf{cost\_stretched}(\colorT)
&\geq
\frac{
  1
}{
  1 + \colorS
  - \left\lfloor \colorS \log_{\colorT}\Big(
    (\colorT - \colorS)(\colorT^{1/\colorS} - 1) + 1
  \Big)\right\rfloor
}.
\end{align*}
\end{lemma}

\begin{proof}
At time $\colorT > \colorS$, we have discarded at least $\colorT - \colorS$ data items.
Hence, total gap space is $\sum \colorg \geq \colorT - \colorS$.
For optimal minimization of gap size ratio, we may assume
\begin{align}
\mathsf{gap\_space} = \colorT - \colorS.
\label{eqn:gap-space-a}
\end{align}

Due to discretization, the smallest possible gap size is 1 data item.
Optimal retention grows successive gap sizes by a factor of $\colorT^{1/\colorS}$.
Calculating total gap space as a sum of gap sizes,
\begin{align}
\mathsf{gap\_space}
&=
\sum_{i = 0}^{\mathsf{num\_gaps}} \colorT^{i/\colorS} 
\nonumber \\
&=
\frac{
  \colorT^{(\mathsf{num\_gaps} + 1)/\colorS} - 1
}{
  \colorT^{1/\colorS} - 1
}.
\label{eqn:gap-space-b}
\end{align}

Equating \ref{eqn:gap-space-a} and \ref{eqn:gap-space-b} and solving for the number of discrete gaps instantiated,
\begin{align*}
\\
\mathsf{num\_gaps}
&\eqnmarkbox[gray]{numgapsgeq}{\geq}
\left\lfloor
\colorS \log_{\colorT}\Big(
  (\colorT - \colorS)(\colorT^{1/\colorS} - 1) + 1
\Big) - 1
\right\rfloor.
\annotate[yshift=1em]{above,left}{numgapsgeq}{%
Integer floor ensures lower bound on $\mathsf{num\_gaps}$.
}
\end{align*}

Counting discarded time steps and retained ``fence posts,'' the smallest gap (of at least size 1) will be located $\mathsf{num\_gaps} + \mathsf{gap\_space}$ time steps back from the most recent observed time $\colorT$.
Note that the $\mathsf{num\_gaps}$ term accounts for the time steps occupied by retained data between gaps (i.e., ``fence posts'').
So, the first gap will occur at time $\colorTbar = \colorT - \mathsf{num\_gaps} - \mathsf{gap\_space}$ and the gap size ratio will be at least
\begin{align*}
\frac{\colorG_{\colorT}(\colorTbar)}{\colorTbar}
&\geq
\frac{1}{
\colorT
- \left\lfloor
\colorS \log_{\colorT}\Big(
  (\colorT - \colorS)(\colorT^{1/\colorS} - 1) + 1
\Big) - 1
\right\rfloor - (\colorT - \colorS)
}
\end{align*}
for $\colorTbar > 0$.

Simplifying terms gives the result.
\end{proof}

\begin{lemma}[Space required to store $\mathsf{goal\_stretched}$]
\label{thm:stretched-first-n-space}

Buffer space $\colorS$ suffices to store set $\mathsf{goal\_stretched}$.
That is,
\begin{align*}
|\bigcup_{\colorh \geq 0}
\{ \colorTbar = i2^{\colorh + 1} + 2^{\colorh} - 1 \text{ for } i \in [0 \twodots n(\colorT) - 1] : \colorTbar < \colorT \}| \leq \colorS
\end{align*}
for $n(\colorT) = 2^{\colors - 1 - \colortau}$.
\end{lemma}
\begin{proof}
As defined over supported $\colorT < 2^{\colorS - 1}$, all meta-epochs $\colortau<\colors$.
Counting data items required by $\mathsf{goal\_stretched}$,
\begin{align*}
\\
|\mathsf{goal\_stretched}|
&=
\sum_{\colorh} \min\Big(
\eqnmarkbox[gray]{sonetau}{2^{\colors - 1 - \colortau}},\;\; |\{ \colorTbar < \colorT : \colorH(\colorTbar) = \colorh \}|\Big)\\
&=
1 + \sum_{\colorh=0}^{\colors + \colort - 1} \min\Big(2^{\colors - 1 - \colortau},\;\; 2^{\colors + \colort - 1 - \colorh}\Big)\\
&=
1 + \sum_{\colorh=0}^{\colors + \colort - 1} \min\Big(2^{\colors - 1 - \colortau},\;\; 2^{\colors + \colort - 1 - \colorh}\Big).
\annotate[yshift=1em]{above,left}{sonetau}{%
The set $\mathsf{goal\_stretched}$ only requires $2^{\colors - 1 - \colortau}$ of each \hv.
}
\end{align*}

Splitting where $\colors - 1 - \colortau = \colors + \colort - 1 - \colorh$ (i.e., $\colorh = \colort + \colortau$),
\begin{align*}
|\mathsf{goal\_stretched}|
&=
1 + \sum_{\colorh=0}^{\colort + \colortau} 2^{\colors - 1 - \colortau} + \sum_{\colorh=\colort + \colortau + 1}^{\colors + \colort - 1} 2^{\colors + \colort - 1 - \colorh}\\
&=
1 + (\colort + \colortau + 1) 2^{\colors - 1 - \colortau} + 2^{\colors - 1 - \colortau} - 1 \tag{via summation identities}\\
&=
(\colort + \colortau + 2) 2^{\colors - 1 - \colortau}\\
&\leq
(2^{\colortau + 1} - (\colortau + 1) - 1 + \colortau + 2) 2^{\colors - 1 - \colortau} \tag{Equation \ref{eqn:meta-epoch-defn} for $\min(\colort \in \lBrace \colortau + 1 \rBrace)$} \\
&\leq
(2^{\colortau + 1}) 2^{\colors - 1 - \colortau}\\
&\leq
2^{\colors}\\
&\stackrel{\checkmark}\leq
\colorS.
\end{align*}




\end{proof}

\begin{lemma}[Minimum retained data items per \hv{}]
\label{thm:stretched-discarded-incidence-count}
No data item $\colorTbar'$ is discarded unless more than $2^{\colors - 1 - \colortau}$ items with \hv{} $\colorH(\colorTbar')$ have been encountered.
That is,
\begin{align*}
|\{
\colorHcal_{\colort}(\colork) = \colorH(\colorTbar')
: \colork \in [0\twodots\colorS)
\}|
&\geq
\min\Big(
|\{
\colorTbar \in [0 \twodots \colorT)
: \colorH(\colorTbar) = \colorH(\colorTbar')
\}|,\;\;
2^{\colors - 1 - \colortau}
\Big).
\end{align*}
\end{lemma}

\begin{proof}
By layout design, this proposition is trivially true for hanoi values with at least $2^{\colors-1-\colortau}$ sites.
However, we must consider \hv{}'s with fewer than $2^{\colors-1-\colortau}$ reserved sites more closely.
For these under-reserved \hv{}'s $\colorh$, we must show that no more items $\colorH(\colorTbar) = \colorh$ are encountered than sites reserved to \hv{} $\colorh$.

\begin{proofpart}[How many hanoi values $\colorh$ have $2^{\colors - 1 - \colortau}$ reserved sites?]

At the outset of each meta-epoch $\colortau$, there remain $2^{\colors - 1 - \colortau}$ uninvaded segments.
Recall that at any epoch $\colort>0$, the smallest invading segment will be slated next for invasion after the current invasion's $R$ epochs.
Thus, the smallest uninvaded segment's size at the outset of meta-epoch $\colortau$ can be calculated by subtracting growth during current meta-epoch $\colortau$ from site at next meta-epoch $\colortau - 1$,
\begin{align*}
R(\colortau + 1) - R(\colortau)
&= (2^{\colortau + 1} - 1) - (2^{\colortau} - 1) \tag{by Lemma \ref{thm:stretched-meta-epoch}}\\
&= 2^{\colortau + 1} - 2^{\colortau}\\
&= 2^{\colortau}.
\end{align*}
With one site contributed for each \hv{} per uninvaded segment, all \hv{} $\colorh < 2^{\colortau}$ thus have reserved at least $2^{\colors - 1 - \colortau}$ sites.
We thus can restrict consideration to $\colorh \geq 2^{\colortau}$.
\end{proofpart}

\begin{proofpart}[Hanoi values without $2^{\colors-1-\colortau}$ reserved sites]
Recall that at the conclusion of epoch $\colort$, we have encountered one of the highest-value \hv{} $\colorh$, one of the second highest-value \hv{} $\colorh-1$, two of the third-highest \hv{} $\colorh-2$, etc.
Also be reminded that the highest-value encountered \hv{} $\colorh$ increases by one per epoch $\colort$.

Initial reservation segments are laid out with sizes drawn from the hanoi sequence (Formula \ref{eqn:stretched-segment-sizes}).
By construction, retained reservations grow exactly one site per epoch.
Because reservations are eliminated in increasing order of their initialized size $r$, we will always (over supported domain $\colorTbar < 2^{\colors}$) have the largest reservation segment $r=\colors$ to provide a site for the lone instances of our two highest hanoi values $\colorh=\colort+\colors$ and $\colorh=\colort+\colors-1$.
Along these lines, we can store the two instances of the next-smallest \hv{} $\colorh=\colort+\colors-2$ in the largest and second-largest reservations $r=\colors$ and $r=\colors-2$.
Proceeding into deeper uninvaded segment layers, reserved site count doubles --- exactly in step with \hv{} instance counts.

With segments $r \geq \colors - \colortau$ active, we can safely store all encountered \hv{} $\colorH(\colorTbar) = \colorh$ instances for the largest $\colors - \colortau$ encountered \hv's.
During epoch $\colort$, the highest-encountered \hv{} is $\colorh=\colors + \colort$.
So, we can safely store all encountered instances for \hv{}'s
\begin{align*}
\colorh
&\geq
\colors + \colort - (\colors - \colortau)\\
&\geq
\colort + \colortau
\end{align*}
over the entirety of meta-epoch $\colortau$.
With $(\colort = 2^{\colortau} - \colortau) \in \colortausetoft$, we can thus further restrict our consideration to $\colorh < 2^{\colortau}$.
\end{proofpart}

\begin{proofpart}[Have we accounted for all hanoi values?]
Combining the above, the question of covering all encountered \hv's $0\leq\colorh\leq\colors+\colort$ becomes whether $\exists \colorh \in \mathbb{N}$ such that $\colorh < 2^{\colortau}$ and $\colorh \geq 2^{\colortau}$.
No such $\colorh$ exists, so we have accounted for all \hv{} in satisfying our requirements.
\end{proofpart}

\end{proof}

\begin{corollary}[Minimum retained items per \hv{}, bound approximations]
\label{thm:stretched-reservation-count}
Under the stretched curation algorithm, space for at least
\begin{align*}
n &\geq
\max\Big(
  \frac{\colorS}{2(\colort + \colors)},\;\;
  \frac{\colorS}{4\colort}
\Big)
\end{align*}
encountered data items of each \hv{} $\colorh$ is provided.
\end{corollary}
\begin{proof}

By Lemma \ref{thm:stretched-discarded-incidence-count}, space for at least $2^{\colors - 1 - \colortau}$ encountered items of each \hv{} is provided.
Applying Supplementary Lemma \ref{thm:meta-epoch-bound} completes the result.
At any epoch $1 \leq \colort \leq \colorS - \colors$,
\begin{align*}
n
&= 2^{\colors - 1 - \colortau}\\
&\geq 2^{\colors - 1 - \colortau}\\
&\geq2^{\colors - 1 - \min\Big(
  \log_2(\colort + \colors),\;\;
  \log_2(\colort) + 1
\Big)}\\
&\geq \max\Big(
  2^{\colors - 1 - \log_2(\colort + \colors)},\;\;
  2^{\colors - 1 - \log_2(\colort) - 1}
\Big)\\
&\stackrel{\checkmark}{\geq} \max\Big(
  \frac{\colorS}{2(\colort + \colors)},\;\;
  \frac{\colorS}{4\colort}
\Big).
\end{align*}

\end{proof}

\section{Stretched Algorithm Gap Size Ratio} \label{sec:gap-size-ratio-stretched}

\begin{lemma}[Stretched algorithm retained data items]
\label{thm:retained-equivalence-stretched}
If the first $n$ data items $\colorH(\colorTbar) = \colorh$ for each \hv{} $\colorh$ are retained, then we are guaranteed to have retained
\begin{align*}
\colorTbar
&\in
\{
  j'2^{\colorh'} - 1
  :
  j' \in [1 \twodots 2n]
  \text{ and }
  \colorh' \in \mathbb{N}
\}.
\end{align*}
Note that, although this formulation nominally includes $\colorTbar > \colorT$, an extension filtering $\colorTbar \in [0 \twodots \colorT)$ follows trivially.
\end{lemma}
\begin{proof}

Recall that the $j$th instance of hanoi value $\colorh$ appears at ingest time
\begin{align*}
\colorTbar
&= j2^{\colorh + 1} + 2^{\colorh} - 1,
\end{align*}
indexed from $j=0$.

The set of retained data items can be denoted
\begin{align*}
\mathsf{have\_retained} \coloneq
\{
  j2^{\colorh + 1} + 2^{\colorh} - 1
  :
  j \in [0 \twodots n-1]
  \text{ and }
  \colorh \in \mathbb{N}
\}.
\end{align*}

We will show $\mathsf{have\_retained}$ equivalent to,
\begin{align*}
\mathsf{want\_retained} \coloneq
\{
  j'2^{\colorh'} - 1
  :
  j' \in [1 \twodots 2n]
  \text{ and }
  \colorh' \in \mathbb{N}
\}.
\end{align*}

\begin{proofpart}[$\mathsf{have\_retained} \subseteq \mathsf{want\_retained}$]
Suppose $\colorTbar \in \mathsf{have\_retained}$.
Then $\exists j \in [0 \twodots n-1]$ and $\colorh \in \mathbb{N}$ such that
\begin{align*}
\colorTbar
&= j2^{\colorh + 1} + 2^{\colorh} - 1\\
&= (2j + 1)2^{\colorh} - 1.
\end{align*}
Noting $2j + 1 \in [1 \twodots 2n]$ gives $\mathsf{have\_retained} \stackrel{\checkmark}{\subseteq} \mathsf{want\_retained}$.
\end{proofpart}

\begin{proofpart}[$\mathsf{want\_retained} \subseteq \mathsf{have\_retained}$]
Suppose $\colorTbar \in \mathsf{want\_retained}$.
Then $\exists j'\in [1 \twodots 2n]$ and $\colorh' \in \mathbb{N}$ such that
\begin{align*}
\colorTbar
&= j'2^{\colorh'} - 1.
\end{align*}

First, where $j' \in [1,3,5,\;\;\ldots,2n-1]$,
\begin{align*}
\colorTbar
&= \frac{j'-1}{2} 2^{\colorh' + 1} + 2^{\colorh'} - 1.
\end{align*}
Because $\frac{j'-1}{2} \in [0 \twodots n-1]$ here, $\mathsf{want\_retained} \stackrel{\checkmark}{\subseteq} \mathsf{have\_retained}$ in this case.

In the case that $j' \in [0,2,4,\;\;\ldots,2n]$,
\begin{align*}
\colorTbar
&= j'2^{\colorh'} - 1\\
&=
\eqnmarkbox[gray]{inpositiven}{
  \frac{j'}{2}
}2^{\colorh' + 1} - 1.\\
\annotate[yshift=0em]{below,right}{inpositiven}{$\in [1\twodots n]$}
\end{align*}

Recalling that $\colorH(j'/2 - 1) = \log_2 \max\{ i \in \{2^{\mathbb{N}}\} : j'/2 \bmod i = 0 \}$,
\begin{align*}
\\
\colorTbar
&=
\eqnmarkbox[gray]{inoddintegers}{
  \frac{j'/2}{2^{\colorH(j'/2 - 1)}}
}
2^{\colorh'} - 1\\
&=
\eqnmarkbox[gray]{inevenintegers}{
  \Big(\frac{j'/2}{2^{\colorH(j'/2 - 1)}} - 1\Big)
}
2^{\colorh'} + 2^{\colorh'} - 1.\\
\annotate[yshift=1em]{above,right}{inoddintegers}{$\in \{x \in [1,3,5,\;\;\ldots] : x \leq n\}$}
\annotate[yshift=0em]{below,right}{inevenintegers}{$\in \{x \in [0,2,4,\;\;\ldots] : x \leq n - 1\}$}
\end{align*}
Pulling out a factor of 2 from the first coefficient,
\begin{align*}
  \colorTbar
  &=
  \eqnmarkbox[Orchid]{}{
  \frac{
    \frac{j'/2}{2^{\colorH(j'/2 - 1)}} - 1
  }{2}
  }
2^{\colorh' + 1} + 2^{\colorh'} - 1.
\end{align*}
With
\begin{align*}
  \eqnmarkbox[Orchid]{}{
  \frac{
    \frac{j'/2}{2^{\colorH(j'/2 - 1)}} - 1
  }{2}
  }
  &\in \{x \in \mathbb{N}: x \leq (n-1)/2\}\\
  &\stackrel{\checkmark}{\in} [0 \twodots n-1],
\end{align*}
we have $\mathsf{want\_retained} \stackrel{\checkmark}{\subseteq} \mathsf{have\_retained}$ in this case, too.
\end{proofpart}
\end{proof}

\begin{lemma}[Stretched gap size ratio given first $n$ items per \hv{}]
\label{thm:gap-size-ratio-stretched}
If the first $n$ data items $\colorH(\colorTbar) = \colorh$ for each \hv{} $\colorh$ are retained at time $\colorT$, then gap size ratio is bounded,
\begin{align*}
\mathsf{cost\_stretched}(\colorT)
&\leq
\frac{1}{n}.
\end{align*}
\end{lemma}
\begin{proof}

From Lemma \ref{thm:retained-equivalence-stretched}, we have retained data items
\begin{align*}
\mathsf{want\_retained} =
\{
  j'2^{\colorh'} - 1
  :
  j' \in [1 \twodots 2n]
  \text{ and }
  \colorh' \in \mathbb{N}
\}.
\end{align*}

What is the smallest $m \in \{2^{\mathbb{N}}\}$ such that $m \times (2n - 1) \geq \colorTbar$?
\begin{align*}
m \times 2n
&\geq \colorTbar + 1\\
m
&\geq \frac{\colorTbar + 1}{2n}\\
m
&= \left\lceil \frac{\colorTbar + 1 }{2n} \right\rceil_{\mathrm{bin}}.
\end{align*}

So, $\colorG_{\colorT}(\colorTbar) \leq \left\lceil \frac{\colorTbar + 1 }{2n} \right\rceil_{\mathrm{bin}} - 1$.
Thus, for $\colorTbar > 0$, gap size ratio $\mathsf{cost\_stretched}(\colorT)$ can be bounded
\begin{align*}
\frac{\colorG_{\colorT}(\colorTbar)}{\colorTbar}
&\leq
\frac{
\left\lceil \frac{\colorTbar + 1 }{2n} \right\rceil_{\mathrm{bin}} - 1
}{
\colorTbar
}\\
&\leq
\frac{
2\frac{\colorTbar + 1}{2n} - 1
}{
\colorTbar
}\\
&\leq
\frac{1}{n} + \frac{1 - n}{n\colorTbar}\\
&\leq
\frac{1}{n} - \frac{n - 1}{n\colorTbar}\\
&\leq
\frac{1}{n} - \eqnmarkbox[gray]{geq}{\frac{1 - 1/n}{\colorTbar}}\\
&\stackrel{\checkmark}{\leq}
1 / n.
\annotate[yshift=0em]{below,right}{geq}{$\geq 0$}
\end{align*}

\end{proof}

\section{Tilted Algorithm}

\begin{lemma}[Last instance of a \hv{} within epoch $\colort$]
\label{thm:tilted-last-touched}
The final instance of each \hv{} encountered during an epoch is placed in the rightmost site reserved for that hanoi value.
That is, during any epoch $\colort$
\begin{align*}
\colorK\Big(
  \max\{\colorT \in \colortsetofT : \colorH(\colorT) = \colorh\}
\Big)
=
\max\{\colork \in [0\twodots\colorS) : \colorHcal_{\colort}(\colork) = \colorh \}
\end{align*}
for all $\colorh \in \{\colorHcal_{\colort}(\colork) : \colork \in [0\twodots\colorS)\}$.
\end{lemma}

\begin{proof}
Hanoi value instances do not cycle back to leftmost reservation $r=\colors$ until the number of encountered \hv{} instances $2^{\colort + \colors - \colorh}$ exceeds space for $2^{\colors - \colortau - 1}$ items guaranteed by Lemma \ref{thm:stretched-discarded-incidence-count}.
Before this point, $\colorh \geq \colort + \colortau + 1$.
We will consider this case separately from $\colorh < \colort + \colortau + 1$.

\begin{proofpart}[$\colorh \geq \colort + \colortau + 1$]
\label{thm:hgeqttau}
Order segment bunches by descending initial size $r$.
Observe that bunches $0,\;\;1,\;\;\ldots,\;\;i$ contain a total of $2^i$ segments.
Placing into segments belonging to bunch $i = \colort + \colors - \colorh$ sychronizes accrued reservations with net encountered \hv{} instances, $2^{\colort + \colors - \colorh}$.
Filling a new, smaller $r$ bunch layer each epoch ensures the rightmost reserved site is filled last each epoch.
\end{proofpart}

\begin{proofpart}[$\colorh < \colort + \colortau + 1$]
In this case, the number of sites reserved to a \hv{} $\colorh$ will be $2^{\colors - 1 - \colortau}$ or, if the current meta-epoch's pending invasion has not yet reached \hv{} $\colorh$, $2 \times 2^{\colors - 1 - \colortau}$.
If, given the latter, reserved sites equal encountered \hv{} instances $2^{\colort + \colors - \colorh} = 2 \times 2^{\colors - 1 - \colortau}$, simply proceed to fill $i$th bunch $i=\colort + \colors - \colorh$ as in Part \ref{thm:hgeqttau} above.

Otherwise, to ensure completion of exactly full cycles around our ``ring buffer'' of sites reserved to \hv{} $\colorh$, we must show that the number of sites $\colork$ reserved to \hv{} $\colorh$ evenly divides the number of \hv{} $\colorh$ instances encountered during epoch $\colort$.
That is, we wish to show
\begin{align*}
\\
\eqnmarkbox[gray]{numitemsh}{
  \mathsf{num\_items}_{\colorh}
}
\bmod
\eqnmarkbox[gray]{numsitesh}{
  \mathsf{num\_sites}_{\colorh}
}
&= 0.
\annotate[yshift=1em]{above,left}{numitemsh}{$|\{
  \colorT \in \colortsetofT : \colorH(\colort) = \colorh
\}|$}
\annotate[yshift=1em]{above,right}{numsitesh}{$|\{
  \colork \in [0\twodots\colorS) : \colorHcal_{\colort}(\colork) = \colorh
\}|$}
\end{align*}

How many instances of a hanoi value $\colorh$ are encountered during epoch $\colort$?
This is
\begin{align*}
\mathsf{num\_items}_{\colorh}
= 2^{\colort + \colors - \colorh} - \left\lfloor 2^{\colort + \colors - \colorh - 1}\right\rfloor.
\end{align*}

How many sites are reserved to a hanoi value $\colorh$ during epoch $\colort$?
As established above, we only concern
\begin{align*}
\mathsf{num\_sites}_{\colorh}
&<
2^{\colort + \colors - \colorh}\\
&\leq
2^{\colort + \colors - \colorh - 1}.
\end{align*}
Because both $\mathsf{num\_items}_{\colorh}$ and $\mathsf{num\_sites}_{\colorh}$ are $\in \{2^{\mathbb{N}}\}$, all that remains is to show
\begin{align*}
\mathsf{num\_items}_{\colorh}
&\stackrel{?}{\geq}
\mathsf{num\_sites}_{\colorh}\\
2^{\colort + \colors - \colorh} - \left\lfloor 2^{\colort + \colors - \colorh - 1}\right\rfloor
&\stackrel{\checkmark}{\geq}
2^{\colort + \colors - \colorh - 1}.
\end{align*}
\end{proofpart}
\end{proof}

\begin{lemma}[Leftmost invaded site is overwritten last within epoch $\colort$]
\label{thm:tilted-last-overwritten}
Among the invaded data items $\colork$ overwritten at epoch $\colort > 0$, the leftmost data item is overwritten last.
That is,
\begin{align*}
\min\{ \colorK(\colorT) \text{ for } \colorT \in \colortsetofT : \colorHcal_{\colort-1}(\colorK(\colorT)) \neq \colorH(\colorT) \}
&=
\colorK\Big(\max( \colorT \in \colortsetofT )\Big).
\end{align*}
\end{lemma}
\begin{proof}
By design, the leftmost site invaded during an epoch $\colort > 0$ is $\colork = \colort + \colors$, invaded by \hv{} $\colorh = \colort + \colors$.
Hanoi value $\colorh = \colort + \colors$ occurs first at ingest time $\colorT = 2^{\colort + \colors} - 1$.
Epoch $\colort + 1$ begins at time $\colorT = 2^{\colort + \colors}$, so epoch $\colort$ (which begins at time $\colorT = 2^{\colort + \colors - 1}$) ends at ingest time $2^{\colort + \colors - 1}$, giving the result.
\end{proof}

\begin{lemma}[Monotonicity of \hv{} reservation $\colorHcal_{\colort}(\colork)$ for buffer site $\colork$]
\label{thm:tilted-invader-minus-invaded}
A site's assigned hanoi value reservation never decreases.
Where it increases, it does so by at least 2.
Formally, where $\colorHcal_{\colort + 1}(\colork) \neq \colorHcal_{\colort}(\colork)$,
\begin{align*}
\colorHcal_{\colort + 1}(\colork) - \colorHcal_{\colort}(\colork) \geq 2.
\end{align*}
\end{lemma}

\begin{proof}
By design, invasion of any segment begins at the segment's leftmost site $\colork$, always assigned $\colorHcal_{\colort}(\colork) = 0$.
Because singleton $r=0$ reservation segments never invade, the invader of this leftmost $\colorh = 0$ site will always stem from a segment $r \geq 1$ and have \hv{} $\colorh > 1$.
The delta $\colorHcal_{\colort + 1}(\colork) - \colorHcal_{\colort}(\colork) \geq 2$ remains constant over subsequent invasion steps because invader and invaded \hv{}'s both increment by exactly 1 each epoch (until complete elimination of the invaded segment).
\end{proof}

\begin{lemma}[Invasion overwrite order within epoch $\colort$]
\label{thm:tilted-invading-overwrite-order}
Except for the leftmost invaded site in segment $r=\colors$, invaded sites are overwritten left-to-right. For $\colort > 0$, pick
\begin{align*}
\\
\colork',\colork''
\in
\eqnmarkbox[gray]{invadedsites}{
  \mathsf{invaded\_sites}_{\colort}
}
: \colors < \colork' < \colork'' < \colorS.
\annotate[yshift=1em]{above,left}{invadedsites}{$\{
  \colork \in [0\twodots\colorS)
  : \colorHcal_{\colort - 1}(\colork) \neq \colorHcal_{\colort}(\colork)
\}$}
\end{align*}
Then,
\begin{align*}
\min\{
  \colorT \in \colortsetofT
  : \colorK(\colorT) = \colork'
\}
<
\min\{
  \colorT \in \colortsetofT
  : \colorK(\colorT) = \colork''
\}.
\end{align*}

\end{lemma}
\begin{proof}
Recall from Lemma \ref{thm:stretched-meta-epoch} that $\colortau$ tells how many reservation segment subsumption cycles have elapsed.
Recall also that $\min(\colort \in \colortausetoft) = 2^{\colortau} - \colortau$.

There are $2^{\colors - 1 - \colortau}$ uninvaded reservation segments at meta-epoch $\colortau$.
From left to right across buffer space, invader \hv{}'s during epoch $\colort \geq 1$ are
\begin{align*}
\colort + \colors, \;\; \colorH(0) + \colort + \colortau, \;\; \colorH(1) + \colort + \colortau, \;\;\ldots, \;\; \colorH(2^{\colors - 1 - \colortau} - 2) + \colort + \colortau.
\end{align*}

Rewritten based on properties of the hanoi sequence and excluding the leftmost invader (characterized separately in Lemma \ref{thm:tilted-last-overwritten}),
\begin{align*}
\\
\colorH\Big(
\eqnmarkbox[gray]{hanoitime1}{1}
\times 2^{\colort + \colortau} + 2^{\colort + \colors - 1} - 1\Big),\;\;
\colorH\Big(
\eqnmarkbox[gray]{hanoitime2}{2} \times 2^{\colort + \colortau} + 2^{\colort + \colors - 1} - 1
\Big),
\;\;\ldots,\;\;
\colorH\Big(
\eqnmarkbox[gray]{hanoitime3}{
(2^{\colors - 1 - \colortau} - 1)
}
\times 2^{\colort + \colortau}
+ 2^{\colort + \colors - 1} - 1
\Big).
\end{align*}
\annotate[yshift=1em]{above,right}{hanoitime1,hanoitime2,hanoitime3}{$[1 \twodots 2^{\colors - 1 - \colortau} - 1]$}

To reach our proof objective, we will analyze the sequence of timepoints $\mathsf{invader\_hv\_times}_{\colort}$ mapping to invader \hv{}'s,
\begin{align*}
\colorT \in i \times 2^{\colort + \colortau} +
2^{\colort + \colors - 1} - 1
:
i \in [1 \twodots 2^{\colors - 1 - \colortau} - 1].
\end{align*}
First, we should confirm $\mathsf{invader\_hv\_times}_{\colort} \subseteq \colort$.
With $\colort + \colortau \geq 2$ for all invasions because both $\colort > 0$ and $\colortau > 0$, we readily have
\begin{align*}
\min(\mathsf{invader\_hv\_times}_{\colort})
&\stackrel{?}{\geq}
\min(\colorT \in \colortsetofT)\\
2^{\colort + \colortau} + 2^{\colort + \colors - 1} - 1
&\stackrel{\checkmark}{\geq}
2^{\colort + \colors - 1}.
\end{align*}
Testing $\mathsf{invader\_hv\_times}_{\colort}$ against the upper bound of $\colort$,
\begin{align*}
\max(\mathsf{invader\_hv\_times}_{\colort})
&\stackrel{?}{\leq}
\max(\colorT \in \colortsetofT)
\\
(2^{\colors - 1 - \colortau}  - 1)
2^{\colort + \colortau} + 2^{\colort + \colors - 1}
&\stackrel{?}{\leq}
2^{\colort + \colors} - 1
\\
2^{\colors - 1 - \colortau} \times 2^{\colort + \colortau} - 2^{\colort + \colortau}
&\stackrel{?}{\leq}
2^{\colort + \colors - 1} - 1
\\
2^{\colort + \colors - 1}
&\stackrel{\checkmark}{\leq}
2^{\colort + \colors - 1} + 3.
\tag{as above, $\colortau + \colort \geq 2$}
\end{align*}

Our final step is to establish that $\mathsf{invader\_hv\_times}_{\colort}$ captures all $\colorT \in \colortsetofT$ here $\colorh = \colorH(\colorT)$ is an invading hanoi value --- that is, $\colorH(\colorT) \geq \colort + \colortau$.
Remark that $\colorT$ such that $\colorH(\colorT) \geq \colort + \colortau$ occur are spaced $2^{\colort + \colortau}$ items apart.
Because $\mathsf{invader\_hv\_times}_{\colort}$ have an identical cadence, our question boils down to whether
\begin{align*}
\min(\mathsf{invader\_hv\_times}_{\colort})
- 2^{\colort + \colortau}
&\stackrel{?}{<}
\min(\colorT \in \colortsetofT)\\
2^{\colort + \colors - 1} - 1
&\stackrel{\checkmark}{<}
2^{\colort + \colors - 1},
\end{align*}
and
\begin{align*}
\max(\mathsf{invader\_hv\_times}_{\colort})
+ 2^{\colort + \colortau}
&\stackrel{?}{>}
\max(\colorT \in \colortsetofT)\\
2^{\colors - 1 - \colortau} 2^{\colort + \colortau}
&\stackrel{?}{>}
2^{\colort + \colors} - 1\\
2^{\colort + \colors - 1}
+ 2^{\colort + \colors - 1}
&\stackrel{?}{>}
2^{\colort + \colors} - 1\\
2^{\colort + \colors}
&\stackrel{\checkmark}{>}
2^{\colort + \colors} - 1.
\end{align*}

The result follows from algorithm specification, with each \hv{} filling its own newly assigned reservation sites from left to right.
\end{proof}

\begin{lemma}[Minimum recent items retained per \hv{}]
\label{thm:tilted-most-recent-retained}
At least the most recent $2^{\colors - 1 - \colortau}$ encountered instances of every \hv{} $\colorh$ are retained under tilted curation.
Concretely, we wish to show $\mathsf{goal\_tilted} \subseteq \colorB_{\colorT}$, with $\mathsf{goal\_tilted}$ defined per Equation \ref{eqn:goal-tilted-set}.
\end{lemma}
\begin{proof}
From Lemma \ref{thm:stretched-discarded-incidence-count}, we have reservations available to store at least the first $2^{\colors - 1 - \colortau}$ instances of each hanoi value.
After this point, data item placement cycles around sites reserved to a \hv as a ring buffer --- keeping most recent $2^{\colors - 1 - \colortau}$ instances.
However, we must validate behavior at the transition points where this ring buffer shrinks due to invasion.

In the case of invasion, the number of reserved sites drops from $2^{\colors - \colortau}$ to $2^{\colors - 1 - \colortau}$.
Recall from Lemma \ref{thm:tilted-last-touched} that the final instance of each \hv{} each epoch is placed into the rightmost reservation segment.
We therefore know that the final $2^{\colors - 1 - \colortau}$ instances of a \hv{} encountered during an epoch were laid out left to right in each of the smallest-size remaining segments, $r = \colortau$ (with the last instance occupying the rightmost reservation segment).

So, at the outset of epoch $\colort$, reassigned sites $\{\colork \in [0\twodots\colorS) : \colorHcal_{\colort - 1}(\colork) \neq \colorHcal_{\colort}(\colork)\}$ always contain the most recent $2^{\colors - 1 - \colortau}$ instances of \hv{} $\colorh = \colorHcal_{\colort - 1}(\colork)$, arranged left to right.
If data items in these reassigned sites were lost instantaneously at time $\min(\colorT \in \colortsetofT)$, we would not meet our proof objectives.
At that point, we would have none of the most recent $2^{\colors - 1 - \colortau}$ \hv{} $\colorh$ data items retained.
However, data items are not lost instantaneously when a site is reassigned.
Instead, data items in reassigned sites $\colork$ linger until they are \textit{actually} overwritten by incoming data items $\colorT \in \colortsetofT$ with $\colorK(\colorT) = \colork$.

From Lemma \ref{thm:tilted-invading-overwrite-order}, we have that, over the course of an epoch, invaded data items are overwritten left to right --- except the leftmost reservation, which is overwritten last.
Ensuring retention of the most recent $2^{\colors - 1 - \colortau}$ data items for a \hv{} during invasion therefore requires two desiderata:
\begin{enumerate}
\item at least two instances of invaded \hv{} $\colorh$ occur before the first invading overwrite, and
\item the cadence of overwrites proceeds no faster than fresh instances of invaded \hv{} $\colorh$ accrue.
\end{enumerate}

\begin{mybox}
\textbf{Intuition.}
Imagine the chain of $2^{\colors - 1 - \colortau}$ recent instances of \hv{} $\colorh$ as the protagonist of the classic video game ``snake'' \citep{de2016complexity}.
In that game, the titular snake slithers by growing at its head and shrinking at its tail.
Analogously, our sequence of most recent \hv{} instances adds new items at the front and has tail items overwritten.
When an invasion occurs and half of ring buffer sites are reassigned, the snake's body of $2^{\colors - 1 - \colortau}$ sites is stretched across the reassigned half of the ring buffer.
In other words, our snake is laid out entirely within the \textit{danger zone}!

At the point when an invasion epoch $\colort$ begins, our snake containing $2^{\colors - 1 - \colortau}$ items will be chased into the preserved half of the ring buffer as overwrites enchroach at its rear.
The two desiderata described above ensure that the snake (1) pulls ahead and (2) stays ahead of invading overwrites to keep $2^{\colors - 1 - \colortau}$ body segments intact.
Mixing metaphors, the snake slithers head then tail to safety as the rickety bridge of reassigned but not-yet-overwritten sites it had been occupying collapses behind it.
After escaping the reassigned $2^{\colors - 1 - \colortau}$ ring buffer sites, the snake of recent \hv{} instances happily crawls in circles around its $2^{\colors - 1 - \colortau}$ reserved sites --- at least, until invaded again.
\end{mybox}

\begin{proofpart}[Two instances of invaded \hv{} before first invading overwrite]
Let $\colorT' = \min(\colorT \in \colortsetofTone)$.
The fractal properties of the hanoi sequence provide the following equivalence for hanoi values encountered during epoch $\colort + 1$:
\begin{align*}
\colorH(\colorT'),\;\; \colorH(\colorT'+1), \;\;\ldots, \colorH(2\colorT' - 1) = \colorH(0),\;\; \colorH(1), \;\;\ldots,\;\; \colorH(\colorT' - 1).
\end{align*}
Recall that $2\colorT' - 1 = \max(\colorT \in \colortsetofTone)$.

By Lemma \ref{thm:tilted-invader-minus-invaded}, for \hv{} $\colorHcal_{\colort}(\colork)$ invaded by \hv{} $\colorHcal_{\colort + 1}(\colork)$ (i.e., $\colorHcal_{\colort}(\colork) \neq \colorHcal_{\colort + 1}(\colork)$), we have $\colorHcal_{\colort + 1}(\colork) \geq \colorHcal_{\colort}(\colork) + 2$.
Hanoi value $\colorh$ occurs first at ingest time $\colorT = 2^{\colorh} - 1$ and then recurs at $\colorT = 2^{\colorh + 1} + 2^{\colorh} - 1 = 3 \times 2^{\colorh} - 1$.
Hence,
\begin{align*}
|\{\colorT \in [0 \twodots 3 \times 2^{\colorh} - 1] : \colorH(\colorT) = \colorh\}| = 2.
\end{align*}
\end{proofpart}
With $3 \times 2^{\colorh} - 1 < 2^{\colorh + 2} - 1 = \min\{\colorT : \colorH(\colorT) = \colorh + 2\}$, we have our result.

\begin{proofpart}[Overwrite cadence slower than invaded \hv{} cadence]
The cadence of a \hv{} $\colorh$, after first occuring at time $\colorT=2^{\colorh} - 1$ is to recur every $2^{\colorh + 1}$ ingests, where $\colorT \bmod 2^{\colorh + 1} = 2^{\colorh} - 1$.
Ingests with \hv{} $\colorH(\colorT) \geq \colorh$ occur twice as frequently, where $\colorT \bmod 2^{\colorh} = 2^{\colorh} - 1$.

Again, by Lemma \ref{thm:tilted-invader-minus-invaded}, for \hv{} $\colorHcal_{\colort}(\colork)$ invaded by \hv{} $\colorHcal_{\colort + 1}(\colork)$ (i.e., $\colorHcal_{\colort}(\colork) \neq \colorHcal_{\colort + 1}(\colork)$), we have $\colorHcal_{\colort + 1}(\colork) \geq \colorHcal_{\colort}(\colork) + 2$.
New incidences of invaded \hv{} $\colorh = \colorHcal_{\colort}(\colork)$ accrue faster than they are overwritten by ingests with $\colorH(\colorT) \geq \colorh + 2$ because
\begin{align*}
2^{\colorh + 1} \stackrel{\checkmark}{<} 2^{\colorh + 2}.
\end{align*}
\end{proofpart}

\end{proof}

\section{Tilted Algorithm Gap Size Ratio} \label{sec:gap-size-ratio-tilted}

\begin{lemma}[Tilted algorithm retained data items]
\label{thm:retained-equivalence-tilted}
If the most recent $n$ data items $\colorH(\colorTbar) = \colorh$ for each \hv{} $\colorh$ are guaranteed retained, then we are guaranteed to have all
\begin{align*}
\colorTbar
&\in
\{
  2^{\colorh'}\Big(\left\lfloor \frac{\colorT}{2^{\colorh'}} \right\rfloor - j'\Big) - 1
  :
  j' \in [0 \twodots 2n - 1]
  \text{ and }
  \colorh' \in \mathbb{N}
\}.
\end{align*}
Note that, although this formulation nominally includes $\colorTbar < 0$, an extension filtering $\colorTbar \in [0 \twodots \colorT)$ follows trivially.
\end{lemma}
\begin{proof}

Recall that the $j$th instance of hanoi value $\colorh$ appears at ingest time
\begin{align*}
\colorTbar
&= j2^{\colorh + 1} + 2^{\colorh} - 1
\end{align*}
with $j$ indexed from zero.

The set of retained data items can be denoted
\begin{align*}
\mathsf{have\_retained} \coloneq
\{
  2^{\colorh + 1}\Big( \left\lfloor\frac{\colorT - 2^{\colorh}}{2^{\colorh + 1}} \right\rfloor - j\Big) + 2^{\colorh} - 1
  :
  j \in [0 \twodots n-1]
  \text{ and }
  \colorh \in \mathbb{N}
\}.
\end{align*}
We will show $\mathsf{have\_retained}$ equivalent to,
\begin{align*}
\mathsf{want\_retained} \coloneq
\{
  2^{\colorh'}\Big(\left\lfloor \frac{\colorT}{2^{\colorh'}} \right\rfloor - j'\Big) - 1
  :
  j' \in [0 \twodots 2n-1]
  \text{ and }
  \colorh' \in \mathbb{N}
\}.
\end{align*}

Sublemma \ref{thm:tilted-rsubset} shows $\mathsf{have\_retained} \subseteq \mathsf{want\_retained}$.
From Sublemma \ref{thm:tilted-subsetr}, $\mathsf{want\_retained} \subseteq \mathsf{have\_retained}$.
Hence, $\mathsf{want\_retained} = \mathsf{have\_retained}$.

\end{proof}

\begin{sublemma}[$\mathsf{have\_retained} \subseteq \mathsf{want\_retained}$]
\label{thm:tilted-rsubset}
Set $\mathsf{have\_retained}$ subsets set $\mathsf{want\_retained}$,
\begin{align*}
&\{
  2^{\colorh + 1}\Big( \left\lfloor\frac{\colorT - 2^{\colorh}}{2^{\colorh + 1}} \right\rfloor - j\Big) + 2^{\colorh} - 1
  :
  j \in [0 \twodots n-1]
  \text{ and }
  \colorh \in \mathbb{N}
\}\\
&\subseteq
\{
  2^{\colorh'}\Big(\left\lfloor \frac{\colorT}{2^{\colorh'}} \right\rfloor - j\Big) - 1
  :
  j' \in [0 \twodots 2n-1]
  \text{ and }
  \colorh' \in \mathbb{N}
\}.
\end{align*}
\end{sublemma}
\begin{proof}
Suppose $\colorTbar \in \mathsf{have\_retained}$.
Then $\exists j \in [0 \twodots n-1]$ and $\colorh \in \mathbb{N}$ such that
\begin{align*}
\colorTbar
&= 2^{\colorh \eqnmarkbox[yellow]{}{+ 1}}\Big(\left\lfloor \frac{\colorT - 2^{\colorh}}{2^{\colorh + 1}} \right\rfloor - j \Big) \eqnmarkbox[orange]{}{+ 2^{\colorh}} - 1\\
&= 2^{\colorh} \Big(\eqnmarkbox[yellow]{}{2}\left\lfloor \frac{\colorT - 2^{\colorh}}{2^{\colorh + 1}} \right\rfloor  - \eqnmarkbox[yellow]{}{2}j \eqnmarkbox[orange]{}{+ 1} \Big) - 1.
\end{align*}
Denoting $\epsilon \in \{0, 1\}$ as a continuity correction for the integer floor,
\begin{align*}
\colorTbar
&= 2^{\colorh} \Big( \left\lfloor \eqnmarkbox[YellowGreen]{}{2} \frac{\colorT \eqnmarkbox[Orchid]{}{- 2^{\colorh}}}{2^{\colorh+1}} \right\rfloor \eqnmarkbox[YellowGreen]{}{- \epsilon}  - 2j + 1 \Big) - 1\\
&= 2^{\colorh} \Big( \left\lfloor \frac{\colorT}{2^{\colorh}} \right\rfloor \eqnmarkbox[Orchid]{}{- 1} - \epsilon - 2j + 1 \Big) - 1\\
&= 2^{\colorh} \Big(\left\lfloor \frac{\colorT}{2^{\colorh}}\right\rfloor - (2j + \epsilon) \Big) - 1.
\end{align*}
Note that $(2j + \epsilon) \in [0 \twodots 2n-1]$ for $j \in [0 \twodots n-1]$, giving $\mathsf{have\_retained} \stackrel{\checkmark}{\subseteq} \mathsf{want\_retained}$.
\end{proof}

\begin{sublemma}[$\mathsf{want\_retained} \subseteq \mathsf{have\_retained}$]
\label{thm:tilted-subsetr}
Set $\mathsf{want\_retained}$ subsets $\mathsf{have\_retained}$,
\begin{align*}
&\{
  2^{\colorh'}\Big(\left\lfloor \frac{\colorT}{2^{\colorh'}} \right\rfloor - j'\Big) - 1
  :
  j' \in [0 \twodots 2n-1]
  \text{ and }
  \colorh' \in \mathbb{N}
\}\\
&\subseteq
\{
  2^{\colorh + 1}\Big( \left\lfloor\frac{\colorT - 2^{\colorh}}{2^{\colorh + 1}} \right\rfloor - j\Big) + 2^{\colorh} - 1
  :
  j \in [0 \twodots n-1]
  \text{ and }
  \colorh \in \mathbb{N}
\}.
\end{align*}

\end{sublemma}
\begin{proof}
Suppose $\colorTbar \in \mathsf{want\_retained}$.
Then $\exists j' \in [0 \twodots 2n - 1]$ and $\colorh' \in \mathbb{N}$ such that
\begin{align*}
\colorTbar
&= 2^{\colorh'}\Big(\left\lfloor \frac{\colorT}{2^{\colorh'}} \right\rfloor - j'\Big) - 1.
\end{align*}

Begin by calculating how many factors of two divide $\colorTbar + 1$, $\colorH(\colorTbar)$.
Note that we have $\colorH(\colorTbar) \geq \colorh'$ because $2^{\colorh'}$ divides $\colorTbar + 1 = 2^{\colorh'}(\left\lfloor \frac{\colorT}{2^{\colorh'}} \right\rfloor - j')$.
With this fact in hand, we may rearrange our formula for $\colorTbar$,
\begin{align}
\colorTbar
&= 2^{\colorh'}\Big(\left\lfloor \frac{\colorT}{2^{\colorh'}} \right\rfloor - j'\Big) - 1 \nonumber\\
&= 2^{\eqnmarkbox[yellow]{}{\colorH(\colorTbar)}} \Big(
\frac{\left\lfloor \frac{\colorT}{2^{\colorh'}} \right\rfloor - j'}{2^{\eqnmarkbox[yellow]{}{\colorH(\colorTbar) - \colorh'}}}
\Big)
- 1 \nonumber\\
&= 2^{\colorH(\colorTbar)} \Big(
\frac{\left\lfloor \frac{\colorT}{2^{\colorh'}} \right\rfloor \eqnmarkbox[orange]{}{- 2^{\colorH(\colorTbar) - \colorh'}}}{2^{\colorH(\colorTbar) - \colorh'}}
\eqnmarkbox[orange]{}{+ 1}
- \frac{j'}{2^{\colorH(\colorTbar) - \colorh'}}
\Big)
- 1 \nonumber\\
&= 2^{\colorH(\colorTbar)} \Big(
\frac{\left\lfloor \frac{\colorT}{2^{\colorh'}} \eqnmarkbox[orange]{}{- 2^{\colorH(\colorTbar) - \colorh'}} \right\rfloor}{2^{\colorH(\colorTbar) - \colorh'}}
- \frac{j'}{2^{\colorH(\colorTbar) - \colorh'}}
\Big)
\eqnmarkbox[orange]{}{+ 2^{\colorH(\colorTbar)}}
- 1 \nonumber\\
&= 2^{\colorH(\colorTbar) \eqnmarkbox[Orchid]{}{+ 1}} \Big(
\frac{\left\lfloor \frac{\colorT}{2^{\colorh'}} \eqnmarkbox[YellowGreen]{}{- 2^{\colorH(\colorTbar) - \colorh'}} \right\rfloor}{2^{\colorH(\colorTbar) - \colorh' \eqnmarkbox[Orchid]{}{+ 1}}}
-
\frac{j'}{2^{\colorH(\colorTbar) - \colorh' \eqnmarkbox[Orchid]{}{+ 1}}}
\Big)
+ 2^{\colorH(\colorTbar)}
- 1 \nonumber\\
&= 2^{\colorH(\colorTbar) + 1} \Big(
\frac{\left\lfloor \frac{\colorT \eqnmarkbox[YellowGreen]{}{- 2^{\colorH(\colorTbar)}}}{2^{\eqnmarkbox[SkyBlue]{}{\colorh'}}} \right\rfloor}{2^{\colorH(\colorTbar) \eqnmarkbox[SkyBlue]{}{-\colorh'} + 1}}
- \frac{j'}{2^{\colorH(\colorTbar) - \colorh' + 1}}
\Big)
+ 2^{\colorH(\colorTbar)}
- 1.
\nonumber
\end{align}

Letting $\epsilon \in [0, 1/2)$ denote a continuity correction factor for the integer floor,
\begin{align}
\colorTbar
&= 2^{\colorH(\colorTbar) + 1} \Big(
\left\lfloor
\frac{\colorT - 2^{\colorH(\colorTbar)}}{2^{\colorH(\colorTbar) + 1}}
\right\rfloor
\eqnmarkbox[SkyBlue]{}{+ \epsilon}
-
\frac{j'}{2^{\colorH(\colorTbar) - \colorh' + 1}}
\Big)
+ 2^{\colorH(\colorTbar)}
- 1 \nonumber \\
&= 2^{\colorH(\colorTbar) + 1} \Big(
\left\lfloor
\frac{\colorT - 2^{\colorH(\colorTbar)}}{2^{\colorH(\colorTbar) + 1}}
\right\rfloor
-
\eqnmarkbox[gray]{proofgoal}{
  \Big(
  \frac{j'}{2^{\colorH(\colorTbar) - \colorh' + 1}}
  - \epsilon
  \Big)
}
\Big)
+ 2^{\colorH(\colorTbar)}
- 1.
\label{eqn:needtoshowj}\\
\nonumber
\annotate[yshift=0em]{below,right}{proofgoal}{need to show $\in [0 \twodots n-1]$}
\end{align}

By definition, $2^{\colorH(\colorTbar)}$ divides $\colorTbar + 1$ and the quotient $(\colorTbar + 1) / 2^{\colorH(\colorTbar)}$ is an odd, positive integer.
So, $(\eqnmarkbox[yellow]{}{\colorTbar} + 1) / 2^{\colorH(\colorTbar)} - 1$ is an even, non-negative integer.
Applying this observation to our expression for $\colorTbar$ from Equation \ref{eqn:needtoshowj},
\begin{align*}
  \frac{
  \eqnmarkbox[yellow]{}{
  2^{\colorboxed{orange}{\colorH(\colorTbar)} + 1} \Big(
    \left\lfloor
    \frac{\colorT - 2^{\colorH(\colorTbar)}}{2^{\colorH(\colorTbar) + 1}}
    \right\rfloor
    -
    \Big(
    \frac{j'}{2^{\colorH(\colorTbar) - \colorh' + 1}}
    - \epsilon
    \Big)
    \Big)
    \colorboxed{YellowGreen}{+ 2^{\colorH(\colorTbar)}}
    \colorboxed{Orchid}{- 1}
  }
  \colorboxed{Orchid}{+ 1}
  }{2^{\colorboxed{orange}{\colorH(\colorTbar)}}}
  \colorboxed{YellowGreen}{- 1}
  &= 2 \Big(
    \left\lfloor
    \frac{\colorT - 2^{\colorH(\colorTbar)}}{2^{\colorH(\colorTbar) + 1}}
    \right\rfloor
    -
    \Big(
    \frac{j'}{2^{\colorH(\colorTbar) - \colorh' + 1}}
    - \epsilon
    \Big)
    \Big)\\
    &\in [0, 2, 4,\;\; \ldots].
\end{align*}

Dividing by 2,
\begin{align*}
  \left\lfloor
  \frac{\colorT - 2^{\colorH(\colorTbar)}}{2^{\colorH(\colorTbar) + 1}}
  \right\rfloor
  -
  \Big(
  \frac{j'}{2^{\colorH(\colorTbar) - \colorh' + 1}}
  - \epsilon
  \Big)
  &\in \mathbb{N}.
\end{align*}
Because
\begin{align*}
\left\lfloor
\frac{
  \colorT - 2^{\colorH(\colorTbar)}
}{
  2^{\colorH(\colorTbar) + 1}
}
\right\rfloor
- \Big(
\frac{j'}{2^{\colorH(\colorTbar) - \colorh' + 1}}
- \epsilon
\Big)
\in \mathbb{Z},
\end{align*}
we necessarily have
\begin{align*}
\frac{j'}{2^{\colorH(\colorTbar) - \colorh' + 1}}
- \epsilon
\in \mathbb{Z}.
\end{align*}
Further, because $\epsilon \in [0, 1/2)$, $j' \in [0 \twodots 2n-1]$, and $\colorH(\colorTbar) - \colorh'  + 1 \geq 1$,
\begin{align*}
\frac{j'}{2^{\colorH(\colorTbar) - \colorh' + 1}}
- \epsilon
\stackrel{\checkmark}{\in}
[0 \twodots n - 1].
\end{align*}
With $\colorh = \colorH(\colorTbar) \stackrel{\checkmark}{\in} \mathbb{N}$, we have the result: $\mathsf{want\_retained} \stackrel{\checkmark}{\subseteq} \mathsf{have\_retained}$.
\end{proof}

\begin{lemma}[Tilted gap size ratio, given last $n$ items per \hv{}]
\label{thm:gap-size-ratio-tilted}
If the most recent $n$ data items $\colorH(\colorTbar) = \colorh$ for each \hv{} $\colorh$ are retained at time $\colorT$, then gap size ratio is bounded,
\begin{align*}
\mathsf{cost\_tilted}(\colorT)
&\leq
\frac{1}{n - 1/2}.
\end{align*}
\end{lemma}
\begin{proof}

From Lemma \ref{thm:retained-equivalence-tilted}, we have retained data items
\begin{align*}
\mathsf{want\_retained} =
\{
2^{\colorh'}\Big(\left\lfloor \colorT / 2^{\colorh'} \right\rfloor - j\Big) - 1
  :
  j' \in [0 \twodots 2n-1]
  \text{ and }
  \colorh' \in \mathbb{N}
\}.
\end{align*}

Begin by finding the smallest $m \in 2^{\mathbb{N}}$ such that
\begin{align*}
m\Big( \left\lfloor \colorT /m \right\rfloor - (2n - 1)\Big) - 1
&\leq
\colorTbar.
\end{align*}
Solving for $m$,
\begin{align*}
m \left\lfloor \colorT/m \right\rfloor - m(2n - 1)
&\leq \colorTbar + 1\\
m
&\geq \frac{
m \left\lfloor \colorT/m \right\rfloor - \colorTbar - 1
}{2n - 1}\\
m
&\geq \frac{
\colorT - \colorTbar - 1
}{2n - 1}\\
m
&= \left\lceil \frac{\colorT - 1 - \colorTbar}{2n - 1} \right\rceil_{\mathrm{bin}}.
\end{align*}

So, $\colorG_{\colorT}(\colorTbar) \leq \left\lceil \frac{\colorT - 1 - \colorTbar}{2n - 1} \right\rceil_{\mathrm{bin}} - 1$.
Thus, over $\colorTbar < \colorT - 1$, $\mathsf{cost\_tilted}(\colorT)$ can be bounded
\begin{align*}
\frac{\colorG_{\colorT}(\colorTbar)}{\colorT - 1 - \colorTbar}
&\leq
\frac{
\left\lceil \frac{\colorT - 1 - \colorTbar}{2n - 1} \right\rceil_{\mathrm{bin}} - 1
}{
\colorT - 1 - \colorTbar
}\\
&\leq
\frac{
2\frac{\colorT - 1 - \colorTbar}{2n - 1}
 - 1
}{
\colorT - 1 - \colorTbar
}\\
&\leq
\frac{1
}{n - 1/2}
-
\frac{1
}{
\colorT - 1 - \colorTbar
}\\
&\stackrel{\checkmark}{\leq}
\frac{1}{n - 1/2}.
\end{align*}

\end{proof}

\section{Reference Implementations}

\subsection{Steady Algorithm Site Selection Reference Implementation}

\begin{lstlisting}[caption={\code{steady_site_selection.py} implements Algorithm \ref{alg:steady-site-selection}}, label={lst:steady_site_selection.py}, language=Python]
import typing


def ctz(x: int) -> int:
    """Count trailing zeros."""
    assert x > 0
    return (x & -x).bit_length() - 1


def bit_floor(x: int) -> int:
    """Return the largest power of two less than or equal to x."""
    assert x > 0
    return 1 << (x.bit_length() - 1)


def steady_site_selection(S: int, T: int) -> typing.Optional[int]:
    """Site selection algorithm for steady curation.

    Parameters
    ----------
    S : int
        Buffer size. Must be a power of two.
    T : int
        Current logical time.

    Returns
    -------
    typing.Optional[int]
        Selected site, if any.
    """
    s = S.bit_length() - 1
    t = T.bit_length() - s  # Current epoch (or negative)
    h = ctz(T + 1)  # Current hanoi value
    if h < t:  # If not a top n(T) hanoi value...
        return None  # ...discard without storing

    i = T >> (h + 1)  # Hanoi value incidence (i.e., num seen)
    if i == 0:  # Special case the 0th bunch
        k_b = 0  # Bunch position
        o = 0  # Within-bunch offset
        w = s + 1  # Segment width
    else:
        j = bit_floor(i) - 1  # Num full-bunch segments
        B = j.bit_length()  # Num full bunches
        k_b = (1 << B) * (s - B + 1)  # Bunch position
        w = h - t + 1  # Segment width
        assert w > 0
        o = w * (i - j - 1)  # Within-bunch offset

    p = h % w  # Within-segment offset
    return k_b + o + p  # Calculate placement site
\end{lstlisting}

\subsection{Steady Algorithm Lookup Reference Implementation}

\begin{lstlisting}[caption={\code{steady_time_lookup.py} implements Algorithm \ref{alg:steady-time-lookup}}, label={lst:steady_time_lookup.py}, language=Python]
import typing


def steady_time_lookup(
    S: int, T: int
) -> typing.Iterable[typing.Optional[int]]:
    """Ingest time lookup algorithm for steady curation.

    Parameters
    ----------
    S : int
        Buffer size. Must be a power of two.
    T : int
        Current logical time.

    Returns
    -------
    typing.Optional[int]
        Ingest time, if any.
    """
    if T < S:  # Patch for before buffer is filled...
        yield from (v if v < T else None for v in steady_lookup_impl(S, S))
    else:  # ... assume buffer has been filled
        yield from steady_lookup_impl(S, T)


def steady_lookup_impl(S: int, T: int) -> typing.Iterable[int]:
    """Implementation detail for `steady_time_lookup`."""
    assert T >= S - 1  # T <= S redirected to T = S - 1 by steady_time_lookup
    s = S.bit_length() - 1
    t = T.bit_length() - s  # Current epoch

    b = 0  # Bunch physical index (left-to right)
    m_b__ = 1  # Countdown on segments traversed within bunch
    b_star = True  # Have traversed all segments in bunch?
    k_m__ = s + 1  # Countdown on sites traversed within segment
    h_ = None  # Candidate hanoi value__

    for k in range(S):  # Iterate over buffer sites, except unused last one
        # Calculate info about current segment...
        epsilon_w = b == 0  # Correction on segment width if first segment
        # Number of sites in current segment (i.e., segment size)
        w = s - b + epsilon_w
        m = (1 << b) - m_b__  # Calc left-to-right index of current segment
        h_max = t + w - 1  # Max possible hanoi value in segment during epoch

        # Calculate candidate hanoi value...
        _h0, h_ = h_, h_max - (h_max + k_m__) % w
        assert (_h0 == h_) or b_star  # Can skip h calc if b_star is False...
        del _h0  # ... i.e., skip calc within each bunch [[see below]]

        # Decode ingest time of assigned h.v. from segment index g, ...
        # ... i.e., how many instances of that h.v. seen before
        T_bar_k_ = ((2 * m + 1) << h_) - 1  # Guess ingest time
        epsilon_h = (T_bar_k_ >= T) * w  # Correction on h.v. if not yet seen
        h = h_ - epsilon_h  # Corrected true resident h.v.
        T_bar_k = ((2 * m + 1) << h) - 1  # True ingest time
        yield T_bar_k

        # Update within-segment state for next site...
        k_m__ = (k_m__ or w) - 1  # Bump to next site within segment

        # Update h for next site...
        # ... only needed if not calculating h fresh every iter [[see above]]
        h_ += 1 - (h_ >= h_max) * w

        # Update within-bunch state for next site...
        m_b__ -= not k_m__  # Bump to next segment within bunch
        b_star = not (m_b__ or k_m__)  # Should bump to next bunch?
        b += b_star  # Do bump to next bunch, if any
        # Set within-bunch segment countdown, if bumping to next bunch
        m_b__ = m_b__ or (1 << (b - 1))
\end{lstlisting}

\subsection{Stretched Algorithm Site Selection Reference Implementation}

\begin{lstlisting}[caption={\code{stretched_site_selection.py} implements Algorithm \ref{alg:stretched-site-selection}}, label={lst:stretched_site_selection.py}, language=Python]
import typing


def ctz(x: int) -> int:
    """Count trailing zeros."""
    assert x > 0
    return (x & -x).bit_length() - 1


def bit_floor(n: int) -> int:
    """Calculate the largest power of two not greater than n.

    If zero, returns zero.
    """
    mask = 1 << (n >> 1).bit_length()
    return n & mask


def stretched_site_selection(S: int, T: int) -> typing.Optional[int]:
    """Site selection algorithm for stretched curation.

    Parameters
    ----------
    S : int
        Buffer size. Must be a power of two.
    T : int
        Current logical time. Must be less than 2**S - 1.

    Returns
    -------
    typing.Optional[int]
        Selected site, if any.
    """
    s = S.bit_length() - 1
    t = max((T).bit_length() - s, 0)  # Current epoch
    h = ctz(T + 1)  # Current hanoi value
    i = T >> (h + 1)  # Hanoi value incidence (i.e., num seen)

    blt = t.bit_length()  # Bit length of t
    epsilon_tau = bit_floor(t << 1) > t + blt  # Correction factor
    tau = blt - epsilon_tau  # Current meta-epoch
    b = S >> (tau + 1) or 1  # Num bunches available to h.v.
    if i >= b:  # If seen more than sites reserved to hanoi value...
        return None  # ... discard without storing

    b_l = i  # Logical bunch index...
    # ... i.e., in order filled (increasing nestedness/decreasing init size r)

    # Need to calculate physical bunch index...
    # ... i.e., position among bunches left-to-right in buffer space
    v = b_l.bit_length()  # Nestedness depth level of physical bunch
    w = (S >> v) * bool(v)  # Num bunches spaced between bunches in nest level
    o = w >> 1  # Offset of nestedness level in physical bunch order
    p = b_l - bit_floor(b_l)  # Bunch position within nestedness level
    b_p = o + w * p  # Physical bunch index...
    # ... i.e., in left-to-right sequential bunch order

    # Need to calculate buffer position of b_p'th bunch
    epsilon_k_b = bool(b_l)  # Correction factor for zeroth bunch...
    # ... i.e., bunch r=s at site k=0
    k_b = (  # Site index of bunch
        (b_p << 1) + ((S << 1) - b_p).bit_count() - 1 - epsilon_k_b
    )

    return k_b + h  # Calculate placement site...
    # ... where h.v. h is offset within bunch
\end{lstlisting}

\subsection{Stretched Algorithm Lookup Reference Implementation}

\begin{lstlisting}[caption={\code{stretched_time_lookup.py} implements Algorithm \ref{alg:stretched-time-lookup}}, label={lst:stretched_time_lookup.py}, language=Python]
import typing


def ctz(x: int) -> int:
    """Count trailing zeros."""
    assert x > 0
    return (x & -x).bit_length() - 1


def bit_floor(n: int) -> int:
    """Calculate the largest power of two not greater than n.

    If zero, returns zero.
    """
    mask = 1 << (n >> 1).bit_length()
    return n & mask


def stretched_time_lookup(
    S: int, T: int
) -> typing.Iterable[typing.Optional[int]]:
    """Ingest time lookup algorithm for stretched curation.

    Parameters
    ----------
    S : int
        Buffer size. Must be a power of two.
    T : int
        Current logical time.

    Returns
    -------
    typing.Optional[int]
        Ingest time, if any.
    """
    if T < S:  # Patch for before buffer is filled...
        yield from (v if v < T else None for v in stretched_lookup_impl(S, S))
    else:  # ... assume buffer has been filled
        yield from stretched_lookup_impl(S, T)


def stretched_lookup_impl(S: int, T: int) -> typing.Iterable[int]:
    """Implementation detail for `stretched_time_lookup`."""
    assert T >= S  # T < S redirected to T = S by stretched_time_lookup

    s = S.bit_length() - 1
    t = (T).bit_length() - s  # Current epoch

    blt = t.bit_length()  # Bit length of t
    epsilon_tau = bit_floor(t << 1) > t + blt  # Correction factor
    tau0 = blt - epsilon_tau  # Current meta-epoch
    tau1 = tau0 + 1  # Next meta-epoch

    M = (S >> tau1) or 1  # Num invading segments present at current epoch
    w0 = (1 << tau0) - 1  # Smallest segment size at current epoch start
    w1 = (1 << tau1) - 1  # Smallest segment size at next epoch start

    h_ = 0  # Assigned hanoi value of 0th site
    m_p = 0  # Calc left-to-right index of 0th segment (physical segment idx)
    for k in range(S):  # For each site in buffer...
        b_l = ctz(M + m_p)  # Logical bunch index...
        # ... REVERSE fill order (decreasing nestedness/increasing init size r)

        epsilon_w = m_p == 0  # Correction factor for segment size
        w = w1 + b_l + epsilon_w  # Number of sites in current segment

        # Determine correction factors for not-yet-seen data items, Tbar_ >= T
        i_ = (M + m_p) >> (b_l + 1)  # Guess h.v. incidence (i.e., num seen)
        Tbar_k_ = ((2 * i_ + 1) << h_) - 1  # Guess ingest time
        epsilon_h = (Tbar_k_ >= T) * (w - w0)  # Correction factor, h
        epsilon_i = (Tbar_k_ >= T) * (m_p + M - i_)  # Correction factor, i

        # Decode ingest time for ith instance of assigned h.v.
        h = h_ - epsilon_h  # True hanoi value
        i = i_ + epsilon_i  # True h.v. incidence
        yield ((2 * i + 1) << h) - 1  # True ingest time, Tbar_k

        # Update state for next site...
        h_ += 1  # Assigned h.v. increases within each segment
        m_p += h_ == w  # Bump to next segment if current is filled
        h_ *= h_ != w  # Reset h.v. if segment is filled
\end{lstlisting}

\subsection{Tilted Algorithm Site Selection Reference Implementation}

\begin{lstlisting}[caption={\code{tilted_site_selection.py} implements Algorithm \ref{alg:tilted-site-selection}}, label={lst:tilted_site_selection.py}, language=Python]
import typing


def modpow2(dividend: int, divisor: int) -> int:
    """Perform fast mod using bitwise operations.

    Parameters
    ----------
    dividend : int
        The dividend of the mod operation. Must be a positive integer.
    divisor : int
        The divisor of the mod operation. Must be a positive integer and a
        power of 2.

    Returns
    -------
    int
        The remainder of dividing the dividend by the divisor.
    """
    return dividend & (divisor - 1)


def ctz(x: int) -> int:
    """Count trailing zeros."""
    assert x > 0
    return (x & -x).bit_length() - 1


def bit_floor(n: int) -> int:
    """Calculate the largest power of two not greater than n.

    If zero, returns zero.
    """
    mask = 1 << (n >> 1).bit_length()
    return n & mask


def tilted_site_selection(S: int, T: int) -> typing.Optional[int]:
    """Site selection algorithm for tilted curation.

    Parameters
    ----------
    S : int
        Buffer size. Must be a power of two.
    T : int
        Current logical time. Must be less than 2**S - 1.

    Returns
    -------
    typing.Optional[int]
        Selected site, if any.
    """
    s = S.bit_length() - 1
    t = max((T).bit_length() - s, 0)  # Current epoch
    h = ctz(T + 1)  # Current hanoi value
    i = T >> (h + 1)  # Hanoi value incidence (i.e., num seen)

    blt = t.bit_length()  # Bit length of t
    epsilon_tau = bit_floor(t << 1) > t + blt  # Correction factor
    tau = blt - epsilon_tau  # Current meta-epoch
    t_0 = (1 << tau) - tau  # Opening epoch of meta-epoch
    t_1 = (1 << (tau + 1)) - (tau + 1)  # Opening epoch of next meta-epoch
    epsilon_b = t < h + t_0 < t_1  # Uninvaded correction factor
    B = S >> (tau + 1 - epsilon_b) or 1  # Num bunches available to h.v.

    b_l = modpow2(i, B)  # Logical bunch index...
    # ... i.e., in order filled (increasing nestedness/decreasing init size r)

    # Need to calculate physical bunch index...
    # ... i.e., position among bunches left-to-right in buffer space
    v = b_l.bit_length()  # Nestedness depth level of physical bunch
    w = (S >> v) * bool(v)  # Num bunches spaced between bunches in nest level
    o = w >> 1  # Offset of nestedness level in physical bunch order
    p = b_l - bit_floor(b_l)  # Bunch position within nestedness level
    b_p = o + w * p  # Physical bunch index...
    # ... i.e., in left-to-right sequential bunch order

    # Need to calculate buffer position of b_p'th bunch
    epsilon_k_b = bool(b_l)  # Correction factor for zeroth bunch...
    # ... i.e., bunch r=s at site k=0
    k_b = (  # Site index of bunch
        (b_p << 1) + ((S << 1) - b_p).bit_count() - 1 - epsilon_k_b
    )

    return k_b + h  # Calculate placement site...
    # ... where h.v. h is offset within bunch
\end{lstlisting}

\subsection{Tilted Algorithm Lookup Reference Implementation}

\begin{lstlisting}[caption={\code{tilted_time_lookup.py} implements Algorithm \ref{alg:tilted-time-lookup}}, label={lst:tilted_time_lookup.py}, language=Python]
import typing


def modpow2(dividend: int, divisor: int) -> int:
    """Perform fast mod using bitwise operations.

    Parameters
    ----------
    dividend : int
        The dividend of the mod operation. Must be a positive integer.
    divisor : int
        The divisor of the mod operation. Must be a positive integer and a
        power of 2.

    Returns
    -------
    int
        The remainder of dividing the dividend by the divisor.
    """
    assert divisor.bit_count() == 1  # Assert divisor is a power of two
    return dividend & (divisor - 1)


def ctz(x: int) -> int:
    """Count trailing zeros."""
    assert x > 0
    return (x & -x).bit_length() - 1


def bit_floor(n: int) -> int:
    """Calculate the largest power of two not greater than n.

    If zero, returns zero.
    """
    mask = 1 << (n >> 1).bit_length()
    return n & mask


def tilted_time_lookup(
    S: int, T: int
) -> typing.Iterable[typing.Optional[int]]:
    """Ingest time lookup algorithm for tilted curation.

    Parameters
    ----------
    S : int
        Buffer size. Must be a power of two.
    T : int
        Current logical time.

    Returns
    -------
    typing.Optional[int]
        Ingest time, if any.
    """
    if T < S:  # Patch for before buffer is filled...
        yield from (v if v < T else None for v in tilted_lookup_impl(S, S))
    else:  # ... assume buffer has been filled
        yield from tilted_lookup_impl(S, T)


def tilted_lookup_impl(S: int, T: int) -> typing.Iterable[int]:
    """Implementation detail for `tilted_time_lookup`."""
    assert T >= S  # T < S redirected to T = S by tilted_time_lookup

    s = S.bit_length() - 1
    t = (T).bit_length() - s  # Current epoch

    blt = t.bit_length()  # Bit length of t
    epsilon_tau = bit_floor(t << 1) > t + blt  # Correction factor
    tau0 = blt - epsilon_tau  # Current meta-epoch
    tau1 = tau0 + 1  # Next meta-epoch
    t0 = (1 << tau0) - tau0  # Opening epoch of current meta-epoch
    T0 = 1 << (t + s - 1)  # Opening time of current epoch

    M_ = S >> tau1 or 1  # Number of invading segments present at current epoch
    w0 = (1 << tau0) - 1  # Smallest segment size at current epoch start
    w1 = (1 << tau1) - 1  # Smallest segment size at next epoch start

    h_ = 0  # Assigned hanoi value of 0th site
    m_p = 0  # Left-to-right (physical) segment index
    for k in range(S):  # For each site in buffer...
        b_l = ctz(M_ + m_p)  # Reverse fill order (logical) bunch index
        epsilon_w = m_p == 0  # Correction factor for segment size
        w = w1 + b_l + epsilon_w  # Number of sites in current segment
        m_l_ = (M_ + m_p) >> (b_l + 1)  # Logical (fill order) segment index

        # Detect scenario...
        # Scenario A: site in invaded segment, h.v. ring buffer intact
        X_A = h_ - (t - t0) > w - w0  # To be invaded in future epoch t in tau?
        T_i = ((2 * m_l_ + 1) << h_) - 1  # When overwritten by invader?
        X_A_ = h_ - (t - t0) == w - w0 and T_i >= T  # Invaded at this epoch?

        # Scenario B site in invading segment, h.v. ring buffer intact
        X_B = (t - t0 < h_ < w0) and (t < S - s)  # At future epoch t in tau?
        T_r = T0 + T_i  # When is site refilled after ring buffer halves?
        X_B_ = (h_ == t - t0) and (t < S - s) and (T_r >= T)  # At this epoch?

        assert X_A + X_A_ + X_B + X_B_ <= 1  # scenarios are mutually exclusive

        # Calculate corrected values...
        epsilon_G = (X_A or X_A_ or X_B or X_B_) * M_
        epsilon_h = (X_A or X_A_) * (w - w0)
        epsilon_T = (X_A_ or X_B_) * (T - T0)  # Snap back to start of epoch

        M = M_ + epsilon_G
        h = h_ - epsilon_h
        Tc = T - epsilon_T  # Corrected time
        m_l = (X_A or X_A_) * (M_ + m_p) or m_l_

        # Decode what h.v. instance fell on site k...
        j = ((Tc + (1 << h)) >> (h + 1)) - 1  # Num seen, less one
        i = j - modpow2(j - m_l + M, M)  # H.v. incidence resident at site k
        # ... then decode ingest time for that ith h.v. instance
        yield ((2 * i + 1) << h) - 1  # True ingest time, Tbar_k

        # Update state for next site...
        h_ += 1  # Assigned h.v. increases within each segment
        m_p += h_ == w  # Bump to next segment if current is filled
        h_ *= h_ != w  # Reset h.v. if segment is filled
\end{lstlisting}

\section{Reference Implementation Tests}

\subsection{Steady Algorithm Site Selection Tests}

\begin{lstlisting}[caption={\code{test_steady_site_selection.py} tests Listing \ref{lst:steady_site_selection.py}}, label={lst:test_steady_site_selection.py}, language=Python]
import functools
import itertools as it
from random import randrange as rand
import typing

from .steady_site_selection import bit_floor, ctz, steady_site_selection


def test_ctz():
    # fmt: off
    assert [*map(ctz, range(1, 17))] == [
        0, 1, 0, 2, 0, 1, 0, 3, 0, 1, 0, 2, 0, 1, 0, 4
    ]


def test_bit_floor():
    # fmt: off
    assert [*map(bit_floor, range(1, 17))] == [
        1, 2, 2, 4, 4, 4, 4, 8, 8, 8, 8, 8, 8, 8, 8, 16
    ]


def validate_steady_site_selection(fn: typing.Callable) -> typing.Callable:
    """Decorator to validate pre- and post-conditions on site selection."""

    @functools.wraps(fn)
    def wrapper(S: int, T: int) -> typing.Optional[int]:
        assert S.bit_count() == 1  # Assert S is a power of two
        assert 0 <= T  # Assert T is non-negative
        res = fn(S, T)
        assert res is None or 0 <= res < S  # Assert valid output
        return res

    return wrapper


site_selection = validate_steady_site_selection(steady_site_selection)


def test_steady_site_selection8():
    # fmt: off
    actual = (site_selection(8, T) for T in it.count())
    expected = [
        0, 1, 4, 2, 6, 5, 7, 3,  # T 0-7
        None, 6, None, 4, None, 7, None, 0,  # T 8-15
        None, None, None, 6, None, None, None, 5,  # T 16-23
        None, None, None, 7, None, None, None, 1,  # T 24-31
        None, None, None, None, None, None, None, 6 # T 32-39
    ]
    assert all(x == y for x, y in zip(actual, expected))


def test_steady_site_selection16():
    # fmt: off
    actual = (site_selection(16, T) for T in it.count())
    expected = [
        0, 1, 5, 2, 8, 6, 10, 3,  # T 0-7 --- hv 0,1,0,2,0,1,0,3
        12, 9, 13, 7, 14, 11, 15, 4,  # T 8-15 --- hv 0,1,0,2,0,1,0,4
        None, 12, None  # T 16-18 --- hv 0,1,0
    ]
    assert all(x == y for x, y in zip(actual, expected))


def test_steady_site_selection_fuzz():
    testS = (1 << s for s in range(33))
    testT = it.chain(range(10**5), (rand(2**128) for _ in range(10**5)))
    for S, T in it.product(testS, testT):
        site_selection(S, T)  # Validated via wrapper


def test_steady_site_selection_epoch0():
    for S in (1 << s for s in range(1, 21)):
        actual = {site_selection(S, T) for T in range(S)}
        expected = set(range(S))
        assert actual == expected
\end{lstlisting}

\subsection{Steady Algorithm Lookup Tests}

\begin{lstlisting}[caption={\code{test_steady_time_lookup.py} tests Listing \ref{lst:steady_time_lookup.py}}, label={lst:test_steady_time_lookup.py}, language=Python]
import functools
import typing

from .steady_site_selection import steady_site_selection as site_selection
from .steady_time_lookup import steady_time_lookup


def validate_steady_time_lookup(fn: typing.Callable) -> typing.Callable:
    """Decorator to validate pre- and post-conditions on time lookup."""

    @functools.wraps(fn)
    def wrapper(S: int, T: int) -> typing.Iterable[typing.Optional[int]]:
        assert S.bit_count() == 1  # Assert S is a power of two
        assert 0 <= T  # Assert T is non-negative
        res = fn(S, T)
        for v in res:
            assert v is None or 0 <= v < T  # Assert valid output
            yield v

    return wrapper


time_lookup = validate_steady_time_lookup(steady_time_lookup)


def test_steady_time_lookup_against_site_selection():
    for s in range(1, 12):
        S = 1 << s
        T_max = min(1 << 17 - s, 2**S - 1)
        expected = [None] * S
        for T in range(T_max):
            actual = time_lookup(S, T)
            assert all(x == y for x, y in zip(expected, actual))

            site = site_selection(S, T)
            if site is not None:
                expected[site] = T
\end{lstlisting}

\subsection{Stretched Algorithm Site Selection Tests}

\begin{lstlisting}[caption={\code{test_stretched_site_selection.py} tests Listing \ref{lst:stretched_site_selection.py}}, label={lst:test_stretched_site_selection.py}, language=Python]
import functools
import itertools as it
import typing

from .stretched_site_selection import bit_floor, ctz, stretched_site_selection


def test_ctz():
    # fmt: off
    assert [*map(ctz, range(1, 17))] == [
        0, 1, 0, 2, 0, 1, 0, 3, 0, 1, 0, 2, 0, 1, 0, 4
    ]


def test_bit_floor():
    # fmt: off
    assert [*map(bit_floor, range(1, 17))] == [
        1, 2, 2, 4, 4, 4, 4, 8, 8, 8, 8, 8, 8, 8, 8, 16
    ]


def validate_stretched_site_selection(fn: typing.Callable) -> typing.Callable:
    """Decorator to validate pre- and post-conditions on site selection."""

    @functools.wraps(fn)
    def wrapper(S: int, T: int) -> typing.Optional[int]:
        assert S.bit_count() == 1  # Assert S is a power of two
        assert S >= 8  # Assert S is at least 8
        assert 0 <= T  # Assert T is non-negative
        res = fn(S, T)
        assert res is None or 0 <= res < S  # Assert valid output
        return res

    return wrapper


site_selection = validate_stretched_site_selection(stretched_site_selection)


def test_stretched_site_selection8():
    # fmt: off
    actual = (site_selection(8, T) for T in it.count())
    expected = [
        0, 1, 5, 2, 4, 6, 7, 3,  # T 0-7
        None, None, None, 7, None, None, None, 4,  # T 8-15
        None, None, None, None, None, None, None, None,  # T 16-23
        None, None, None, None, None, None, None, 5,  # T 24-31
        None, None, None, None, None, None, None, None # T 32-39
    ]
    assert all(x == y for x, y in zip(actual, expected))


def test_stretched_site_selection16():
    # fmt: off
    actual = (site_selection(16, T) for T in it.count())
    expected = [
        0, 1, 9, 2, 6, 10, 13, 3,  # T 0-7 --- hv 0,1,0,2,0,1,0,3
        5, 7, 8, 11, 12, 14, 15, 4,  # T 8-15 --- hv 0,1,0,2,0,1,0,4
        None, None, None, 8, None, None, None, 12,  # T 16-24 --- hv 0,1,0, ...
        None, None, None, 15, None, None, None, 5,  # T 24-31
        None, None, None, None, None, None, None, None # T 32-39
    ]
    assert all(x == y for x, y in zip(actual, expected))


def test_stretched_site_selection_fuzz():
    for S in (1 << s for s in range(3, 17)):
        for T in range(S - 1):
            site_selection(S, T)  # Validated via wrapper


def test_stretched_site_selection_epoch0():
    for S in (1 << s for s in range(3, 17)):
        actual = {site_selection(S, T) for T in range(S)}
        expected = set(range(S))
        assert actual == expected
\end{lstlisting}

\subsection{Stretched Algorithm Lookup Tests}

\begin{lstlisting}[caption={\code{test_stretched_time_lookup.py} tests Listing \ref{lst:stretched_time_lookup.py}}, label={lst:test_stretched_time_lookup.py}, language=Python]
import functools
import typing

from .stretched_site_selection import (
    stretched_site_selection as site_selection,
)
from .stretched_time_lookup import stretched_time_lookup


def validate_stretched_time_lookup(fn: typing.Callable) -> typing.Callable:
    """Decorator to validate pre- and post-conditions on time lookup."""

    @functools.wraps(fn)
    def wrapper(S: int, T: int) -> typing.Iterable[typing.Optional[int]]:
        assert S.bit_count() == 1  # Assert S is a power of two
        assert 0 <= T  # Assert T is non-negative
        res = fn(S, T)
        for v in res:
            assert v is None or 0 <= v < T  # Assert valid output
            yield v

    return wrapper


time_lookup = validate_stretched_time_lookup(stretched_time_lookup)


def test_stretched_time_lookup_against_site_selection():
    for s in range(1, 12):
        S = 1 << s
        T_max = min(1 << 17 - s, 2**S - 1)
        expected = [None] * S
        for T in range(T_max):
            actual = time_lookup(S, T)
            assert all(x == y for x, y in zip(expected, actual))

            site = site_selection(S, T)
            if site is not None:
                expected[site] = T
\end{lstlisting}

\subsection{Tilted Algorithm Site Selection Tests}

\begin{lstlisting}[caption={\code{test_tilted_site_selection.py} tests Listing \ref{lst:tilted_site_selection.py}}, label={lst:test_tilted_site_selection.py}, language=Python]
import functools
import itertools as it
import typing

from .tilted_site_selection import (
    bit_floor,
    ctz,
    modpow2,
    tilted_site_selection,
)


def test_ctz():
    # fmt: off
    assert [*map(ctz, range(1, 17))] == [
        0, 1, 0, 2, 0, 1, 0, 3, 0, 1, 0, 2, 0, 1, 0, 4
    ]


def test_bit_floor():
    # fmt: off
    assert [*map(bit_floor, range(1, 17))] == [
        1, 2, 2, 4, 4, 4, 4, 8, 8, 8, 8, 8, 8, 8, 8, 16
    ]


def test_modpow2():
    assert modpow2(10, 2) == 0  # 10 % 2 = 0
    assert modpow2(10, 4) == 2  # 10 % 4 = 2
    assert modpow2(10, 8) == 2  # 10 % 8 = 2
    assert modpow2(15, 8) == 7  # 15 % 8 = 7
    assert modpow2(20, 16) == 4  # 20 % 16 = 4
    assert modpow2(16, 16) == 0  # 16 % 16 = 0
    assert modpow2(1, 2) == 1  # 1 % 2 = 1
    assert modpow2(3, 8) == 3  # 3 % 8 = 3
    assert modpow2(1023, 1024) == 1023  # 1023 % 1024 = 1023
    assert modpow2(0, 8) == 0  # 0 % 8 = 0


def validate_tilted_site_selection(fn: typing.Callable) -> typing.Callable:
    """Decorator to validate pre- and post-conditions on site selection."""

    @functools.wraps(fn)
    def wrapper(S: int, T: int) -> typing.Optional[int]:
        assert S.bit_count() == 1  # Assert S is a power of two
        assert S >= 8  # Assert S is at least 8
        assert 0 <= T  # Assert T is non-negative
        res = fn(S, T)
        assert 0 <= res < S  # Assert valid output
        return res

    return wrapper


site_selection = validate_tilted_site_selection(tilted_site_selection)


def test_tilted_site_selection8():
    # fmt: off
    actual = (site_selection(8, T) for T in it.count())
    expected = [
        0, 1, 5, 2, 4, 6, 7, 3,  # T 0-7
        0, 1, 5, 7, 0, 6, 5, 4,  # T 8-15
        0, 1, 0, 2, 0, 6, 0, 3,  # T 16-23
        0, 1, 0, 7, 0, 6, 0, 5,  # T 24-31
        0, 1, 0, 2, 0, 1, 0, 3 # T 32-39
    ]
    assert all(x == y for x, y in zip(actual, expected))


def test_tilted_site_selection16():
    # fmt: off
    actual = (site_selection(16, T) for T in it.count())
    expected = [
        0, 1, 9, 2, 6, 10, 13, 3,  # T 0-7 --- hv 0,1,0,2,0,1,0,3
        5, 7, 8, 11, 12, 14, 15, 4,  # T 8-15 --- hv 0,1,0,2,0,1,0,4
        0, 1, 9, 8, 6, 10, 13, 12,  # T 16-24 --- hv 0,1,0, ...
        0, 7, 9, 15, 6, 14, 13, 5,  # T 24-31
        0, 1, 9, 2, 0, 10, 9, 3 # T 32-39
    ]
    assert all(x == y for x, y in zip(actual, expected))


def test_tilted_site_selection_fuzz():
    for S in (1 << s for s in range(3, 17)):
        for T in range(S - 1):
            site_selection(S, T)  # Validated via wrapper


def test_tilted_site_selection_epoch0():
    for S in (1 << s for s in range(3, 17)):
        actual = {site_selection(S, T) for T in range(S)}
        expected = set(range(S))
        assert actual == expected
\end{lstlisting}

\subsection{Tilted Algorithm Lookup Tests}

\begin{lstlisting}[caption={\code{test_tilted_time_lookup.py} tests Listing \ref{lst:test_tilted_time_lookup.py}}, label={lst:test_tilted_time_lookup.py}, language=Python]
import functools
import typing

from .tilted_site_selection import tilted_site_selection as site_selection
from .tilted_time_lookup import tilted_time_lookup


def validate_tilted_time_lookup(fn: typing.Callable) -> typing.Callable:
    """Decorator to validate pre- and post-conditions on time lookup."""

    @functools.wraps(fn)
    def wrapper(S: int, T: int) -> typing.Iterable[typing.Optional[int]]:
        assert S.bit_count() == 1  # Assert S is a power of two
        assert 0 <= T  # Assert T is non-negative
        res = fn(S, T)
        for v in res:
            assert v is None or 0 <= v < T  # Assert valid output
            yield v

    return wrapper


time_lookup = validate_tilted_time_lookup(tilted_time_lookup)


def test_tilted_time_lookup_against_site_selection():
    for s in range(1, 12):
        S = 1 << s
        T_max = min(1 << 17 - s, 2**S - 1)
        expected = [None] * S
        for T in range(T_max):
            actual = time_lookup(S, T)
            assert all(x == y for x, y in zip(expected, actual))

            site = site_selection(S, T)
            if site is not None:
                expected[site] = T
\end{lstlisting}

\putbib

\end{bibunit}


\begin{thebibliography}{}

\bibitem[Abdulla et~al., 2004]{abdulla2004simulation}
Abdulla, G., Critchlow, T., \& Arrighi, W. (2004).
\newblock Simulation data as data streams.
\newblock {\em ACM SIGMOD Record}, 33(1), 89--94.
\newblock \url{https://doi.org/10.1145/974121.974137}

\bibitem[Agarwal et~al., 2009]{agarwal2009faster}
Agarwal, V., Bader, D.~A., Dan, L., Liu, L.-K., Pasetto, D., Perrone, M., \&
  Petrini, F. (2009).
\newblock Faster fast: multicore acceleration of streaming financial data.
\newblock {\em Computer Science - Research and Development}, 23(3-4), 249--257.
\newblock \url{https://doi.org/10.1007/s00450-009-0093-5}

\bibitem[Aggarwal et~al., 2003]{aggarwal2003framework}
Aggarwal, C.~C., Yu, P.~S., Han, J., \& Wang, J. (2003).
\newblock A framework for clustering evolving data streams.
\newblock {\em Proceedings 2003 VLDB Conference}, 81--92.
\newblock \url{https://doi.org/10.1016/b978-012722442-8/50016-1}

\bibitem[Aupy et~al., 2013]{aupy2013combination}
Aupy, G., Benoit, A., Herault, T., Robert, Y., Vivien, F., \& Zaidouni, D.
  (2013).
\newblock On the combination of silent error detection and checkpointing.
\newblock {\em 2013 IEEE 19th Pacific Rim International Symposium on Dependable
  Computing}.
\newblock \url{https://doi.org/10.1109/prdc.2013.10}

\bibitem[Cai et~al., 2004]{cai2004maids}
Cai, Y.~D., Clutter, D., Pape, G., Han, J., Welge, M., \& Auvil, L. (2004).
\newblock Maids: mining alarming incidents from data streams.
\newblock {\em Proceedings of the 2004 ACM SIGMOD international conference on
  Management of data}, Sigmod/pods04, 919--920.
\newblock \url{https://doi.org/10.1145/1007568.1007695}

\bibitem[Cordeiro \& Gama, 2016]{cordeiro2016online}
Cordeiro, M. \& Gama, J. (2016).
\newblock Online social networks event detection: a survey.
\newblock {\em Solving Large Scale Learning Tasks. Challenges and Algorithms:
  Essays Dedicated to Katharina Morik on the Occasion of Her 60th Birthday},
  1--41.
\newblock \url{https://doi.org/10.1007/978-3-319-41706-6_1}

\bibitem[Cormode \& Jowhari, 2019]{cormode2019lp}
Cormode, G. \& Jowhari, H. (2019).
\newblock Lp samplers and their applications: A survey.
\newblock {\em ACM Computing Surveys}, 52(1), 1–31.
\newblock \url{https://doi.org/10.1145/3297715}

\bibitem[Elnahrawy, 2003]{elnahrawy2003research}
Elnahrawy, E. (2003).
\newblock Research directions in sensor data streams: solutions and challenges.
\newblock {\em Rutgers University, Tech. Rep. DCIS-TR-527}, 2, D3.

\bibitem[Fischer et~al., 2012]{fischer2012real}
Fischer, F., Mansmann, F., \& Keim, D.~A. (2012).
\newblock Real-time visual analytics for event data streams.
\newblock {\em Proceedings of the 27th Annual ACM Symposium on Applied
  Computing}, Sac 2012.
\newblock \url{https://doi.org/10.1145/2245276.2245432}

\bibitem[Foundation, 2024]{oeis}
Foundation, O. (2024).
\newblock {\em The {O}n-{L}ine {E}ncyclopedia of {I}nteger {S}equences}.
\newblock Published electronically at \url{http://oeis.org}.

\bibitem[Gaber et~al., 2005]{gaber2005mining}
Gaber, M.~M., Zaslavsky, A., \& Krishnaswamy, S. (2005).
\newblock Mining data streams: A review.
\newblock {\em SIGMOD Rec.}, 34(2), 18--26.
\newblock \url{https://doi.org/10.1145/1083784.1083789}

\bibitem[Gama \& Rodrigues, 2007]{gama2007data}
Gama, J. \& Rodrigues, P.~P. (2007).
\newblock Data stream processing.
\newblock {\em Learning from data streams: Processing techniques in sensor
  networks}, 25--39.
\newblock \url{https://doi.org/10.1007/3-540-73679-4_3}

\bibitem[Giannella et~al., 2003]{giannella2003mining}
Giannella, C., Han, J., Pei, J., Yan, X., \& Yu, P.~S. (2003).
\newblock Mining frequent patterns in data streams at multiple time
  granularities.
\newblock {\em Next generation data mining}, volume 212, 191--212. MIT Press.

\bibitem[Graham et~al., 2012]{graham2012data}
Graham, M.~J., Djorgovski, S.~G., Mahabal, A., Donalek, C., Drake, A., \&
  Longo, G. (2012).
\newblock Data challenges of time domain astronomy.
\newblock {\em Distributed and Parallel Databases}, 30(5-6), 371--384.
\newblock \url{https://doi.org/10.1007/s10619-012-7101-7}

\bibitem[Han et~al., 2005]{han2005stream}
Han, J., Chen, Y., Dong, G., Pei, J., Wah, B.~W., Wang, J., \& Cai, Y.~D.
  (2005).
\newblock Stream cube: An architecture for multi-dimensional analysis of data
  streams.
\newblock {\em Distributed and Parallel Databases}, 18(2), 173--197.
\newblock \url{https://doi.org/10.1007/s10619-005-3296-1}

\bibitem[Harris et~al., 2020]{harris2020array}
Harris, C.~R., Millman, K.~J., van~der Walt, S.~J., Gommers, R., Virtanen, P.,
  Cournapeau, D., Wieser, E., Taylor, J., Berg, S., Smith, N.~J., Kern, R.,
  Picus, M., Hoyer, S., van Kerkwijk, M.~H., Brett, M., Haldane, A., del
  R{\'{i}}o, J.~F., Wiebe, M., Peterson, P., G{\'{e}}rard-Marchant, P.,
  Sheppard, K., Reddy, T., Weckesser, W., Abbasi, H., Gohlke, C., \& Oliphant,
  T.~E. (2020).
\newblock Array programming with {NumPy}.
\newblock {\em Nature}, 585(7825), 357--362.
\newblock \url{https://doi.org/10.1038/s41586-020-2649-2}

\bibitem[He et~al., 2010]{he2010comet}
He, B., Yang, M., Guo, Z., Chen, R., Su, B., Lin, W., \& Zhou, L. (2010).
\newblock Comet: batched stream processing for data intensive distributed
  computing.
\newblock {\em Proceedings of the 1st ACM symposium on Cloud computing}, Socc
  '10.
\newblock \url{https://doi.org/10.1145/1807128.1807139}

\bibitem[Henzinger et~al., 1999]{henzinger1998computing}
Henzinger, M., Raghavan, P., \& Rajagopalan, S. (1999).
\newblock {\em Computing on data streams}, 107--118.
\newblock American Mathematical Society.
\newblock \url{https://doi.org/10.1090/dimacs/050/05}

\bibitem[Hill et~al., 2009]{hill2009real}
Hill, D.~J., Minsker, B.~S., \& Amir, E. (2009).
\newblock Real-time bayesian anomaly detection in streaming environmental data.
\newblock {\em Water Resources Research}, 45(4).
\newblock \url{https://doi.org/10.1029/2008wr006956}

\bibitem[Hunter, 2007]{hunter2007matplotlib}
Hunter, J.~D. (2007).
\newblock Matplotlib: A 2d graphics environment.
\newblock {\em Computing in Science \& Engineering}, 9(3), 90--95.
\newblock \url{https://doi.org/10.1109/mcse.2007.55}

\bibitem[Jain et~al., 2022]{jain2022survey}
Jain, S., Verma, R.~K., Pattanaik, K.~K., \& Shukla, A. (2022).
\newblock A survey on event-driven and query-driven hierarchical routing
  protocols for mobile sink-based wireless sensor networks.
\newblock {\em The Journal of Supercomputing}, 78(9), 11492--11538.
\newblock \url{https://doi.org/10.1007/s11227-022-04327-4}

\bibitem[Jiang \& Gruenwald, 2006]{jiang2006research}
Jiang, N. \& Gruenwald, L. (2006).
\newblock Research issues in data stream association rule mining.
\newblock {\em ACM SIGMOD Record}, 35(1), 14--19.
\newblock \url{https://doi.org/10.1145/1121995.1121998}

\bibitem[Johnson et~al., 2005]{johnson2005streams}
Johnson, T., Muthukrishnan, S., Spatscheck, O., \& Srivastava, D. (2005).
\newblock {\em Streams, Security and Scalability}, 1--15.
\newblock Springer Berlin Heidelberg.
\newblock \url{https://doi.org/10.1007/11535706_1}

\bibitem[Kent \& Souppaya, 2006]{kent2006guide}
Kent, K. \& Souppaya, M.~P. (2006).
\newblock {\em Guide to computer security log management}.
\newblock \url{https://doi.org/10.6028/nist.sp.800-92}

\bibitem[Lin et~al., 2004]{lin2004continuously}
Lin, X., Lu, H., Xu, J., \& Yu, J. (2004).
\newblock Continuously maintaining quantile summaries of the most recent n
  elements over a data stream.
\newblock {\em Proceedings. 20th International Conference on Data Engineering},
  Icde-04.
\newblock \url{https://doi.org/10.1109/icde.2004.1320011}

\bibitem[Lucas, 1889]{lucas1889jeux}
Lucas, {\`E}. (1889).
\newblock Jeux scientifiques pour servir {\`a} l’histoire, {\`a}
  l’enseignement et {\`a} la pratique du calcul et du dessin.
\newblock {\em Paris: EL, Imp. Girard et fils}.

\bibitem[Manku \& Motwani, 2002]{manku2002approximate}
Manku, G.~S. \& Motwani, R. (2002).
\newblock {\em Approximate Frequency Counts over Data Streams}, 346--357.
\newblock Elsevier.
\newblock \url{https://doi.org/10.1016/b978-155860869-6/50038-x}

\bibitem[Miebach, 2002]{miebach2002hubble}
Miebach, M.~P. (2002).
\newblock Hubble space telescope on-line telemetry archive for monitoring
  science instruments.
\newblock {\em Observatory Operations to Optimize Scientific Return III},
  volume 4844, 408--416.
\newblock \url{https://doi.org/10.1117/12.460637}

\bibitem[Moreno, 2023]{moreno2023teeplot}
Moreno, M.~A. (2023).
\newblock {\em mmore500/teeplot}.
\newblock \url{https://doi.org/10.5281/zenodo.10440670}

\bibitem[Moreno, 2024a]{moreno2024downstream}
Moreno, M.~A. (2024a).
\newblock {\em mmore500/downstream}.
\newblock \url{https://doi.org/10.5281/zenodo.10866541}

\bibitem[Moreno, 2024b]{moreno2024hsurf}
Moreno, M.~A. (2024b).
\newblock {\em mmore500/hstrat-surface-concept}.
\newblock \url{https://doi.org/10.5281/zenodo.10779240}

\bibitem[Moreno et~al., 2022a]{moreno2022hereditary}
Moreno, M.~A., Dolson, E., \& Ofria, C. (2022a).
\newblock Hereditary stratigraphy: Genome annotations to enable phylogenetic
  inference over distributed populations.
\newblock {\em The 2022 Conference on Artificial Life}, Alife 2022, 64.
\newblock \url{https://doi.org/10.1162/isal_a_00550}

\bibitem[Moreno et~al., 2022b]{moreno2022hstrat}
Moreno, M.~A., Dolson, E., \& Ofria, C. (2022b).
\newblock hstrat: a python package for phylogenetic inference on distributed
  digital evolution populations.
\newblock {\em Journal of Open Source Software}, 7(80), 4866.
\newblock \url{https://doi.org/10.21105/joss.04866}

\bibitem[Moreno et~al., 2024a]{moreno2024guide}
Moreno, M.~A., Ranjan, A., Dolson, E., \& Zaman, L. (2024a).
\newblock {\em A guide to tracking phylogenies in parallel and distributed
  agent-based evolution models}.
\newblock \url{https://doi.org/10.48550/arXiv.2405.10183}

\bibitem[Moreno et~al., 2024b]{moreno2024algorithms}
Moreno, M.~A., {Rodriguez Papa}, S., \& Dolson, E. (2024b).
\newblock {\em Algorithms for efficient, compact online data stream curation}.
\newblock \url{https://doi.org/10.48550/arXiv.2403.00266}

\bibitem[Moreno \& Yang, 2024]{moreno2024wse}
Moreno, M.~A. \& Yang, C. (2024).
\newblock {\em mmore500/wse-sketches}.
\newblock \url{https://doi.org/10.5281/zenodo.10779280}

\bibitem[Moreno et~al., 2024c]{moreno2024trackable}
Moreno, M.~A., Yang, C., Dolson, E., \& Zaman, L. (2024c).
\newblock Trackable agent-based evolution models at wafer scale.
\newblock {\em The 2024 Conference on Artificial Life}.
\newblock \url{https://doi.org/10.48550/arXiv.2404.10861}

\bibitem[Muthukrishnan, 2005]{muthukrishnan2005data}
Muthukrishnan, S. (2005).
\newblock Data streams: Algorithms and applications.
\newblock {\em Foundations and Trends® in Theoretical Computer Science}, 1(2),
  117--236.
\newblock \url{https://doi.org/10.1561/0400000002}

\bibitem[Palpanas et~al., 2004]{palpanas2004online}
Palpanas, T., Vlachos, M., Keogh, E., Gunopulos, D., \& Truppel, W. (2004).
\newblock Online amnesic approximation of streaming time series.
\newblock {\em Proceedings. 20th International Conference on Data Engineering},
  Icde-04.
\newblock \url{https://doi.org/10.1109/icde.2004.1320009}

\bibitem[pandas developers, 2020]{reback2020pandas}
pandas developers (2020).
\newblock pandas-dev/pandas: Pandas.
\newblock {\em Zenodo}.
\newblock \url{https://doi.org/10.5281/zenodo.3509134}

\bibitem[Phithakkitnukoon \& Ratti, 2010]{phithakkitnukoon2010recent}
Phithakkitnukoon, S. \& Ratti, C. (2010).
\newblock A recent-pattern biased dimension-reduction framework for time series
  data.
\newblock {\em Journal of Advances in Information Technology}, 1(4), 168--180.
\newblock \url{https://doi.org/10.4304/jait.1.4.168-180}

\bibitem[Rajeshwari \& Babu, 2016]{rajeshwari2016real}
Rajeshwari, U. \& Babu, B.~S. (2016).
\newblock Real-time credit card fraud detection using streaming analytics.
\newblock {\em 2016 2nd International Conference on Applied and Theoretical
  Computing and Communication Technology (iCATccT)}.
\newblock \url{https://doi.org/10.1109/icatcct.2016.7912039}

\bibitem[Schoellhammer et~al., 2024]{schoellhammer2024lightweight}
Schoellhammer, T., Greenstein, B., Osterweil, E., Wimbrow, M., \& Estrin, D.
  (2024).
\newblock Lightweight temporal compression of microclimate datasets [wireless
  sensor networks].
\newblock {\em 29th Annual IEEE International Conference on Local Computer
  Networks}, Lcn-04.
\newblock \url{https://doi.org/10.1109/lcn.2004.72}

\bibitem[Schützel et~al., 2014]{schutzel2014stream}
Schützel, J., Meyer, H., \& Uhrmacher, A.~M. (2014).
\newblock A stream-based architecture for the management and on-line analysis
  of unbounded amounts of simulation data.
\newblock {\em Proceedings of the 2nd ACM SIGSIM Conference on Principles of
  Advanced Discrete Simulation}, Sigsim-pads '14.
\newblock \url{https://doi.org/10.1145/2601381.2601399}

\bibitem[Sibai et~al., 2016]{sibai2016sampling}
Sibai, R.~E., Chabchoub, Y., Demerjian, J., Kazi-Aoul, Z., \& Barbar, K.
  (2016).
\newblock Sampling algorithms in data stream environments.
\newblock {\em 2016 International Conference on Digital Economy (ICDEc)}.
\newblock \url{https://doi.org/10.1109/icdec.2016.7563142}

\bibitem[Silva et~al., 2013]{silva2013data}
Silva, J.~A., Faria, E.~R., Barros, R.~C., Hruschka, E.~R., Carvalho, A. C. P.
  L. F.~d., \& Gama, J.~a. (2013).
\newblock Data stream clustering: A survey.
\newblock {\em ACM Comput. Surv.}, 46(1).
\newblock \url{https://doi.org/10.1145/2522968.2522981}

\bibitem[Virtanen et~al., 2020]{2020SciPy-NMeth}
Virtanen, P., Gommers, R., Oliphant, T.~E., Haberland, M., Reddy, T.,
  Cournapeau, D., Burovski, E., Peterson, P., Weckesser, W., Bright, J., {van
  der Walt}, S.~J., Brett, M., Wilson, J., Millman, K.~J., Mayorov, N., Nelson,
  A. R.~J., Jones, E., Kern, R., Larson, E., Carey, C.~J., Polat, {\.I}., Feng,
  Y., Moore, E.~W., {VanderPlas}, J., Laxalde, D., Perktold, J., Cimrman, R.,
  Henriksen, I., Quintero, E.~A., Harris, C.~R., Archibald, A.~M., Ribeiro,
  A.~H., Pedregosa, F., {van Mulbregt}, P., \& {SciPy 1.0 Contributors} (2020).
\newblock {{SciPy} 1.0: Fundamental Algorithms for Scientific Computing in
  Python}.
\newblock {\em Nature Methods}, 17, 261--272.
\newblock \url{https://doi.org/10.1038/s41592-019-0686-2}

\bibitem[Waskom, 2021]{waskom2021seaborn}
Waskom, M.~L. (2021).
\newblock seaborn: statistical data visualization.
\newblock {\em Journal of Open Source Software}, 6(60), 3021.
\newblock \url{https://doi.org/10.21105/joss.03021}

\bibitem[{W}es {M}c{K}inney, 2010]{mckinney-proc-scipy-2010}
{W}es {M}c{K}inney (2010).
\newblock {D}ata {S}tructures for {S}tatistical {C}omputing in {P}ython.
\newblock {\em {P}roceedings of the 9th {P}ython in {S}cience {C}onference},
  56--61.
\newblock \url{https://doi.org/10.25080/Majora-92bf1922-00a}

\bibitem[Zhao \& Zhang, 2006]{zhao2005generalized}
Zhao, Y. \& Zhang, S. (2006).
\newblock Generalized dimension-reduction framework for recent-biased time
  series analysis.
\newblock {\em IEEE Transactions on Knowledge and Data Engineering}, 18(2),
  231--244.
\newblock \url{https://doi.org/10.1109/tkde.2006.30}

\end{thebibliography}


\begin{thebibliography}{}

\bibitem[De~Biasi \& Ophelders, 2016]{de2016complexity}
De~Biasi, M. \& Ophelders, T. (2016).
\newblock The complexity of snake.
\newblock \url{https://doi.org/10.4230/lipics.fun.2016.11}

\end{thebibliography}
\end{document}